\documentclass[12pt,twoside,a4paper]{article}
\usepackage{titlesec}
\titleformat{\section}{\large\bfseries}{\thesection.}{0.5em}{}
\usepackage{bm} 
\usepackage{bbm}   	
\usepackage{braket}
\usepackage{graphicx}
\usepackage{array}
\usepackage{needspace}
\usepackage{theorem}
\usepackage{amsmath,multirow,bigdelim}
\usepackage{amssymb,mathrsfs}
\usepackage{latexsym}
\usepackage{blkarray}
\usepackage{cases}
\usepackage{algorithm}
\usepackage{algorithmic}
\usepackage[subrefformat=parens]{subcaption}
\usepackage{url}
\usepackage[T1]{fontenc}
\usepackage{placeins}
\usepackage[numbers]{natbib}
\usepackage[hidelinks]{hyperref}

\newcommand{\one}{\mbox{$1$}\hspace{-0.25em}{\rm l}}
\theorembodyfont{\itshape}

\makeatletter
\gdef\th@plain{\normalfont\itshape
  \def\@begintheorem##1##2{%
        \item[\hskip\labelsep \theorem@headerfont ##1\ ##2.]}%
  \def\@opargbegintheorem##1##2##3{%
        \item[\hskip\labelsep \theorem@headerfont ##1\ ##2\ (##3).]}}
\makeatother

\newtheorem{thm}{Theorem}[section]
\newtheorem{lem}[thm]{Lemma}

\newtheorem{rem}[thm]{Remark}

\newtheorem{defn}[thm]{Definition}

\newtheorem{prob}[thm]{Problem}
\newcommand{\qed}{\hspace*{\fill}$\Box$}
\newenvironment{proof}{\noindent{\it Proof.~}}{\qed}

\makeatletter
\def\eqnarray{\stepcounter{equation}\let\@currentlabel=\theequation
\global\@eqnswtrue
\global\@eqcnt\z@\tabskip\@centering\let\\=\@eqncr
$$\halign to \displaywidth\bgroup\@eqnsel\hskip\@centering
  $\displaystyle\tabskip\z@{##}$&\global\@eqcnt\@ne 
  \hfil$\;{##}\;$\hfil
  &\global\@eqcnt\tw@ $\displaystyle\tabskip\z@{##}$\hfil 
   \tabskip\@centering&\llap{##}\tabskip\z@\cr}
\makeatother
\makeatletter
\newif\if@borderstar
\def\bordermatrix{\@ifnextchar*{%
 \@borderstartrue\@bordermatrix@i}{\@borderstarfalse\@bordermatrix@i*}%
}
\def\@bordermatrix@i*{\@ifnextchar[{\@bordermatrix@ii}{\@bordermatrix@ii[()]}}
\def\@bordermatrix@ii[#1]#2{%
\begingroup
 \m@th\@tempdima8.75\p@\setbox\z@\vbox{%
 \def\cr{\crcr\noalign{\kern 2\p@\global\let\cr\endline }}%
 \ialign {$##$\hfil\kern 2\p@\kern\@tempdima & \thinspace %
  \hfil $##$\hfil && \quad\hfil $##$\hfil\crcr\omit\strut %
  \hfil\crcr\noalign{\kern -\baselineskip}#2\crcr\omit %
  \strut\cr}}%
 \setbox\tw@\vbox{\unvcopy\z@\global\setbox\@ne\lastbox}%
 \setbox\tw@\hbox{\unhbox\@ne\unskip\global\setbox\@ne\lastbox}%
 \setbox\tw@\hbox{%
  $\kern\wd\@ne\kern -\@tempdima\left\@firstoftwo#1%
  \if@borderstar\kern 2pt\else\kern -\wd\@ne\fi%
 \global\setbox\@ne\vbox{\box\@ne\if@borderstar\else\kern 2\p@\fi}%
 \vcenter{\if@borderstar\else\kern -\ht\@ne\fi%
  \unvbox\z@\kern -\if@borderstar2\fi\baselineskip}%
 \if@borderstar\kern-2\@tempdima\kern2\p@\else\,\fi\right\@secondoftwo#1 $%
 }\null \;\vbox{\kern\ht\@ne\box\tw@}%
\endgroup
}
\makeatother
\makeatletter
\DeclareRobustCommand\widecheck[1]{{\mathpalette\@widecheck{#1}}}
\def\@widecheck#1#2{%
    \setbox\z@\hbox{\m@th$#1#2$}%
    \setbox\tw@\hbox{\m@th$#1%
       \widehat{%
          \vrule\@width\z@\@height\ht\z@
          \vrule\@height\z@\@width\wd\z@}$}%
    \dp\tw@-\ht\z@
    \@tempdima\ht\z@ \advance\@tempdima2\ht\tw@ \divide\@tempdima\thr@@
    \setbox\tw@\hbox{%
       \raise\@tempdima\hbox{\scalebox{1}[-1]{\lower\@tempdima\box
\tw@}}}%
    {\ooalign{\box\tw@ \cr \box\z@}}}
\makeatother

\newcommand{\ol}{\overline}

\newcommand{\vc}{\bm}

\newcommand{\sep}{\unskip, }

\newcommand{\dm}{\displaystyle}

\newcommand{\EE}{\mathsf{E}}
\newcommand{\PP}{\mathsf{P}}

\newcommand{\calC}{\mathcal{C}}

\newcommand{\bbK}{\mathbb{K}}

\newcommand{\bbN}{\mathbb{N}}

\newcommand{\bbZ}{\mathbb{Z}}
\renewcommand{\labelenumi}{(\roman{enumi})}
\makeatletter
\long\def\@makefntext#1{%
  \parindent=0pt
  \noindent
  #1%
}
\makeatother
\renewcommand{\thefootnote}{\fnsymbol{footnote}}
\makeatother

\begin{document}
\thispagestyle{empty}

\hfill

\vspace{-10mm}

\begin{center}
{\large\bfseries
PureRank: A parameter-free recursive importance measure for network nodes
\par}

\vspace{3mm}

Hiroyuki Masuyama\begingroup
\renewcommand{\thefootnote}{}%
\footnote{%
Email address: masuyama@tmu.ac.jp.\\
Published version: \textit{Information Sciences} 751 (2026) 123507, available online 15 April 2026.\\
DOI: \url{https://doi.org/10.1016/j.ins.2026.123507}\\
Article history: received 11 July 2025; received in revised form 11 April 2026; accepted 11 April 2026.}%
\addtocounter{footnote}{-1}%
\endgroup

\vspace{1mm}

{\small\itshape
Graduate School of Management, Tokyo Metropolitan University, 1-1 Minami-Osawa, Hachioji, 192--0397, Tokyo, Japan
\par}
\end{center}

\bigskip
\medskip

\begin{center}
{\small
\textbf{Abstract}

\medskip

\begin{tabular}{p{0.85\textwidth}}
This study develops PureRank, a parameter-free importance measure for network nodes based on the recursive definition of importance (RDI). For any directed network, PureRank uniquely determines an importance score vector without user-specified parameters. PureRank can thus provide a neutral reference for parameter-dependent importance measures. PureRank is constructed in three steps: (i) nodes are classified into {\it recurrent}, {\it transient}, and {\it dangling} classes via strongly connected component decomposition; (ii) for each class, the local importance vector is obtained by choosing the parameters of the Katz equation on the class-restricted subnetwork according to the RDI principle; and (iii) the local importance vectors are aggregated into the PureRank vector. This modular design supports parallel and incremental computation while retaining a unified random-surfer interpretation. Numerical experiments on three SNAP networks show that PageRank has a computational advantage over PureRank except when the damping factor $d$ is close to one, and that the similarity of PageRank to PureRank depends on $d$ and the node classification. In the fully recurrent network, similarity increases monotonically with $d$ and reaches Kendall's $\tau_b=0.966$ and Pearson correlation coefficient $=1.000$ at $d=0.999$, whereas in the two transient-dominated networks, similarity varies nonmonotonically with $d$. PureRank is extended to multi-attribute networks.
\end{tabular}
}
\end{center}

\begin{center}
\begin{tabular}{p{0.90\textwidth}}
{\small

\medskip

%
}
\end{tabular}

\end{center}

\section{Introduction: background and contributions}\label{sec:introduction}

Section~\ref{sec:introduction} is organized into three parts. We first describe our motivation and network model for recursive importance (centrality) measures. We next summarize state-of-the-art techniques and their limitations, and finally we outline PureRank and the main contributions.

\subsection{Motivation and network model}
This paper develops a {\it parameter-free} recursive importance (centrality) measure for nodes in directed networks with nonnegative link weights. Recursive importance measures are based on the recursive definition of importance (RDI), which captures the idea that a node should be regarded as important when it is supported by important neighbors, and that this support propagates along network links. Network importance measures, including recursive ones, have become standard tools across a wide range of applications, such as information retrieval, recommendation systems, social network analysis, and bioinformatics (see, e.g., Newman~\cite{Newm10}; for a survey of centrality measures, see Saxena and Iyengar~\cite{Saxe20}).

We consider a weighted directed network with $N$ nodes labeled $1,2,\dots,N$. Let $V=\{1,2,\dots,N\}$, called the {\it node set}. Let $w_{i,j}\ge 0$ denote the weight of the directed link from node $i$ to node $j$, where $w_{i,j}=0$ means that there is no directed link from node $i$ to node $j$. The original network is represented as $G:=(V,\vc{W})$, where $\vc{W}=(w_{i,j})_{i,j\in V}$ is referred to as the {\it weight matrix} (also known as the weighted adjacency matrix). For later use, let $w_i^{out}:=\sum_{j\in V}w_{i,j}$ denote the (weighted) out-degree of node $i$, and let $D:=\{i\in V:w_i^{out}=0\}$ denote the set of {\it dangling nodes}. Furthermore, let $\vc{P}:=(p_{i,j})_{i,j\in V}$ denote the {\it normalized weight matrix} such that
\begin{align}
p_{i,j} =
\left\{
\begin{array}{ll}
\displaystyle{w_{i,j} \over w_i^{out}}, & \mbox{if $i \not\in D$, i.e., $w_i^{out} > 0$,}
\\[1ex]
0, & \mbox{if $i \in D$, i.e., $w_i^{out} = 0$.}
\end{array}
\right.
\label{defn:P}
\end{align}
Eq.~(\ref{defn:P}) implies that $\vc{P}$ is a substochastic matrix whose rows corresponding to nodes in $D$ are zero, whereas each remaining row sums to one.

We introduce the auxiliary normalized network $\ol{G}:=(V,\vc{P})$ for later analysis. The network $\ol{G}$ has the same node set $V$ as $G$ and uses the normalized weight matrix $\vc{P}$ as its weight matrix. For each $S\subseteq V$, let $\ol{G}_S:=(S,\vc{P}_S)$ denote the restriction of $\ol{G}$ to $S$, where $\vc{P}_S:=(p_{i,j})_{i,j\in S}$ is the submatrix of $\vc{P}$ indexed by $S$. The original network $G=(V,\vc{W})$ remains the primary object of this paper, and we use $\ol{G}$ and $\ol{G}_S$ only to streamline later arguments formulated in terms of $\vc{P}$.

Throughout the paper, we follow the convention that bold lowercase Greek letters denote row vectors and bold lowercase English letters denote column vectors. We write the $i$th entry of a vector $\vc{v}$ as $[\vc{v}]_i$, and we use the same notation for row vectors. We also write the $(i,j)$th entry of a matrix $\vc{A}$ as $[\vc{A}]_{i,j}$. For example, $[\vc{W}]_{i,j}$ denotes the $(i,j)$th entry of the weight matrix $\vc{W}$, that is, $[\vc{W}]_{i,j}=w_{i,j}$. Furthermore, for any square matrix $\vc{A}$, $\rho(\vc{A})$ denotes its spectral radius.
Finally, for any finite set $S$, $\vc{e}_S$ denotes the $|S|\times 1$ vector of ones, with the subscript omitted when the underlying set is clear from the context, and $|S|$ denotes the cardinality of $S$.

\subsection{State-of-the-art techniques and limitations}

Seeley~\cite{Seel49} pioneered the RDI principle, which states ``A's popularity is a function of the `popularity' of those who chose him; and their popularity is a function of those who chose them, and so ad infinitum.'' Based on the RDI principle, Seeley~\cite{Seel49} proposed {\it Seeley centrality}, whose score vector $\vc{\sigma}:=(\sigma_j)_{j\in V}$ satisfies
\begin{align}
\sigma_j = \sum_{i \in V}\sigma_i p_{i,j},
\qquad j\in V.
\label{eq:Seeley_formula}
\end{align}
Eq.~(\ref{eq:Seeley_formula}) is written in vector-matrix form as
\begin{align}
\vc{\sigma} = \vc{\sigma}\vc{P}.
\label{eq:Seeley_formula_vec_mat}
\end{align}
We call (\ref{eq:Seeley_formula_vec_mat}) the {\it Seeley equation} on the network $G=(V,\vc{W})$. This equation can also be interpreted as the Seeley equation on the auxiliary normalized network $\ol{G}=(V,\vc{P})$.

Bonacich~\cite{Bona72} also proposed an RDI-based centrality, known as eigenvector centrality. When $\rho(\vc{W})>0$, the corresponding score vector $\vc{\eta}:=(\eta_j)_{j\in V}$ satisfies
\begin{align}
\eta_j = {1 \over \rho(\vc{W})} \sum_{i \in V}\eta_i w_{i,j},
\qquad j\in V.
\label{eq:EV_formula}
\end{align}
Eq.~(\ref{eq:EV_formula}) is written in vector-matrix form as
\begin{align*}
\vc{\eta} = {1 \over \rho(\vc{W})}\vc{\eta}\vc{W}.
\end{align*}

Seeley centrality and eigenvector centrality can exhibit two limitations on general directed networks: the {\it lack of uniqueness} and the {\it lack of completeness}. Both limitations are governed by the decomposition of the network into strongly connected components (SCCs), where an SCC is called closed if it has no outgoing links to any other SCC, and open otherwise. The lack of uniqueness depends on how the network is decomposed into SCCs, whereas the lack of completeness arises from the possibility that nodes in open SCCs receive a score of zero.

The lack of uniqueness means that the score vector is not uniquely determined after normalization. For Seeley centrality, this failure occurs when the network contains multiple closed SCCs, because the eigenvalue $1=\rho(\vc{P})$ of $\vc{P}$ is not simple. For eigenvector centrality, by contrast, multiple closed SCCs do not by themselves preclude uniqueness; a unique normalized nonnegative score vector can still exist provided that $\rho(\vc{W})$ is a simple eigenvalue of $\vc{W}$.

The lack of completeness means that some nodes receive a score of zero in every admissible score vector. For Seeley centrality, all nodes in open SCCs always receive a score of zero. By contrast, for eigenvector centrality, nodes in open SCCs can receive positive scores, depending on the left eigenspace associated with the dominant eigenvalue $\rho(\vc{W})$ of $\vc{W}$. The lack of completeness is especially problematic for networks with large open SCCs, such as the Web~\cite{Brod00}.

Katz centrality~\cite{Katz53} and PageRank~\cite{Brin98} address these two limitations by introducing free parameters, which allow each node in an arbitrary network to receive a unique positive score. Katz~\cite{Katz53} proposed {\it Katz centrality}, which uses two free parameters, called the {\it Katz parameters}: (i) the {\it damping factor} $d>0$ satisfying $d\rho(\vc{W})<1$, and (ii) the uniform {\it baseline score} $\beta\ge0$ assigned to each node. By definition, the {\it Katz equation} $(V,\vc{W},d,\beta)$ for the unknown row vector $\vc{\kappa}:=(\kappa_j)_{j\in V}$ is defined by
\begin{align}
\vc{\kappa} = d\vc{\kappa}\vc{W} + \beta\vc{e}_V^{\top}.
\label{defn:Katz}
\end{align}
According to this notation, (\ref{eq:Seeley_formula_vec_mat}) is interpreted as the Katz equation $(V,\vc{P},1,0)$.

The Katz equation $(V,\vc{W},d,\beta)$ has a unique positive solution whenever $d>0$ satisfies $d\rho(\vc{W})<1$ and $\beta>0$. This solution, called the {\it Katz centrality vector}, is given by
\begin{align}
\vc{\kappa} = \beta\vc{e}_V^{\top}(\vc{I}-d\vc{W})^{-1},
\label{eq:Katz-01}
\end{align}
where $\vc{I}$ denotes the identity matrix. We normalize the Katz centrality vector $\vc{\kappa}$ as
\begin{align}
\vc{\kappa}
=
{
\vc{e}_V^{\top}(\vc{I}-d\vc{W})^{-1}
\over \vc{e}_V^{\top}(\vc{I}-d\vc{W})^{-1}\vc{e}_V
}
=
{
\vc{\mu}_V(\vc{I}-d\vc{W})^{-1}
\over \vc{\mu}_V(\vc{I}-d\vc{W})^{-1}\vc{e}_V
},
\label{eq:Katz-02}
\end{align}
where $\vc{\mu}_S:=\vc{e}_S^{\top}/|S|$ denotes the uniform distribution vector on a finite set $S$. This normalization corresponds to setting the baseline score $\beta$ in (\ref{eq:Katz-01}) to
\begin{align*}
\beta = {1 \over \vc{e}_V^{\top}(\vc{I}-d\vc{W})^{-1}\vc{e}_V}.
\end{align*}

PageRank is a Katz-like importance measure that combines movement along network links with uniform teleportation through the damping factor $d\in(0,1)$. For $d\in(0,1)$, let $\vc{\gamma}^{(d)}:=(\gamma_j^{(d)})_{j\in V}$ denote the PageRank vector with damping factor $d$, where $\gamma_j^{(d)}$ is the PageRank score of node $j\in V$. The PageRank vector $\vc{\gamma}^{(d)}$ satisfies
\begin{align}
\vc{\gamma}^{(d)}
&=d\vc{\gamma}^{(d)}\vc{Q}_V
+ (1-d)\vc{\mu}_V
= d\vc{\gamma}^{(d)}\vc{Q}_V
+ {1-d \over N}\vc{e}_V^{\top},
\label{defn:PageRank-01}
\end{align}
where $\vc{Q}_V$ denotes the stochastic matrix defined by
\begin{align}
\vc{Q}_V=\vc{P}+(\vc{e}_V-\vc{P}\vc{e}_V)\vc{\mu}_V.
\label{eq:olP}
\end{align}
Note that (\ref{defn:PageRank-01}) is interpreted as the Katz equation $(V,\vc{Q}_V,d,(1-d)/N)$. In addition, the term $(\vc{e}_V-\vc{P}\vc{e}_V)\vc{\mu}_V$ in (\ref{eq:olP}) corresponds to teleportation links from each dangling node to all nodes (see \cite[Section~4.5]{Lang06}), whereas the term $(1-d)\vc{\mu}_V$ in (\ref{defn:PageRank-01}) corresponds to a uniform baseline score assigned to all nodes. Moreover, (\ref{defn:PageRank-01}) yields
\begin{align}
\vc{\gamma}^{(d)}
=(1-d)\vc{\mu}_V(\vc{I}-d\vc{Q}_V)^{-1}
=\vc{\mu}_V\sum_{n=0}^{\infty}(1-d)d^n (\vc{Q}_V)^n,
\label{defn:PageRank-02}
\end{align}
which shows that the PageRank vector $\vc{\gamma}^{(d)}$ is uniquely determined. PageRank is also axiomatically characterized in \cite{Was23}.

The damping factor $d\in(0,1)$ controls the balance between the network structure and the uniform baseline in the PageRank vector $\vc{\gamma}^{(d)}$. As $d \downarrow 0$, $\vc{\gamma}^{(d)}$ approaches the uniform vector $\vc{\mu}_V$ on $V$. As $d\uparrow 1$, the baseline contribution $(1-d)/N$ vanishes, and $\vc{\gamma}^{(d)}$ increasingly reflects the original link structure of the network $G$, except for the dangling-node teleportation encoded in $\vc{Q}_V$. The corresponding classwise limit results are presented in Theorem~\ref{thm:classwise_PR_limit}; see also Remark~\ref{rem:PureRank_in_fully_recurrent_NW}.

The damping factor $d$ also affects the convergence rate of the power method for computing $\vc{\gamma}^{(d)}$ (see \cite[Section~4.6 and 5.1]{Lang06}). By \cite[Theorem~4.7.1]{Lang06}, the second largest eigenvalue modulus of the Google matrix $\vc{G}:=d\vc{Q}_V + (1-d)\vc{e}_V\vc{\mu}_V$ equals $d$ times that of $\vc{Q}_V$. Therefore, a smaller $d$ typically accelerates the convergence of the power method, whereas a larger $d$ slows it. More generally, computing PageRank on large-scale networks is itself a critical issue. This issue has motivated extensive algorithmic work on personalized PageRank, which includes standard PageRank as a special case; see, for example, the surveys \cite{Park19,Yang24}, the dynamic algorithms under edge insertions and deletions \cite{Jaya24}, and the massively parallel methods \cite{Hou21}.

The choice of the damping factor $d$ in PageRank remains delicate and debated, although PageRank is widely used in journal impact metrics \cite{Berg08}, web spam detection \cite{Gyon04}, bipartite-graph ranking for recommendation \cite{He17}, gene prioritization \cite{Morr05}, systemic risk assessment in financial networks \cite{Batt12}, and graph-based text ranking \cite{Miha04}; see \cite{Glei15} for a survey of applications beyond the Web. Brin and Page~\cite{Brin98} recommended $d=0.85$; however, smaller values such as $0.5$ or $0.6$ may be more appropriate in some applications, for example, ranking in sports leagues; see Govan~\cite{Gova08} and \cite[Chapter~6]{Lang12}. Avrachenkov et al.~\cite{Avra08} also examined the choice of $d$ from a fairness viewpoint and suggested a value around $0.5$, based on heuristics derived from a mathematical analysis of bow-tie web graphs. Notably, in worst-case networks constructed in~\cite{Bres10}, the top $k$ nodes can appear in all $k!$ possible orders even when $d$ varies within an arbitrarily small interval, such as $[0.84999,0.85001]$ around the recommended value. These observations indicate that PageRank resolves the lack of uniqueness and completeness at the cost of introducing an application-dependent parameter whose choice can substantially affect the resulting ranking.

The sensitivity of PageRank to the damping factor $d$ has motivated parameter-free variants that eliminate the damping factor. One such variant is LeaderRank~\cite{Lu11}, which follows a Seeley-type construction: it augments the original network with a ground node that has bidirectional links to all nodes, computes the Seeley centrality vector on the augmented network, and redistributes the ground-node score uniformly among the original nodes. Nevertheless, this ground-node augmentation modifies the original network itself, thereby raising questions about how the modification should be justified and interpreted.

\subsection{PureRank and contributions}

We propose {\it PureRank}, a parameter-free, RDI-based importance measure defined solely by the intrinsic structure of the original network $G$. PureRank is constructed in three steps. First, nodes are classified into {\it recurrent}, {\it transient}, and {\it dangling} classes via SCC decomposition together with the out-degree condition, and the resulting node classification is denoted by $\calC$, a family of mutually exclusive and collectively exhaustive classes; see Section~\ref{subsec:node_classification} for details. Second, for each Class $S\in\calC$, the local importance vector is obtained by solving the Katz equation on the class-restricted subnetwork $\ol{G}_S=(S,\vc{P}_S)$ with its Katz parameters chosen according to the RDI principle so that the equation is as close as possible to the Seeley equation. Third, these local importance vectors are aggregated into the global importance vector, namely, the {\it PureRank vector}, through the RDI-based local-to-global (RDI-L2G) construction; see Section~\ref{subsec:L2G} for details. The three-step construction requires no user-specified parameters, such as the PageRank damping factor, and does not modify the network $G$ by adding artificial nodes or links. This design embodies ``purity'' by avoiding empirical or heuristic tuning and by not modifying the structure of the target network, which motivates the name PureRank.

Table~\ref{tab:sota_RDI} compares PureRank with representative recursive importance measures in terms of general-network guarantee, parameter dependence, and network modification.
\begin{table}[t]
\centering
\caption{Comparison of representative recursive importance measures with respect to key design trade-offs. ``General-network guarantee'': Yes: the method guarantees a unique positive score vector for arbitrary directed networks; No: such a guarantee is not available in general. ``Parameter dependence'': Yes: the method requires user-specified parameter values; No: it requires no such parameters. ``Network modification'': Yes: the method explicitly augments the original network by introducing artificial nodes or links; No: it introduces no artificial nodes or links.}
\label{tab:sota_RDI}
\resizebox{\linewidth}{!}{%
\begin{tabular}{lcccccc}
\hline
 & Eigenvector & Seeley & Katz & PageRank & LeaderRank & PureRank
\\
\hline
General-network guarantee & No & No & Yes & Yes & Yes & Yes
\\
Parameter dependence & No & No & Yes & Yes & No & No
\\
Network modification & No & No & No & Yes & Yes & No
\\
\hline
\end{tabular}
}
\end{table}
Seeley centrality is parameter-free and, in this sense, ``pure'', but it lacks a general-network guarantee. PureRank preserves this purity while guaranteeing a unique positive score vector for any directed network. Additionally, PureRank avoids both PageRank-style teleportation and LeaderRank-style augmentation of the original network by artificial nodes or links, and instead assigns scores through an explicit decomposition into recurrent, transient, and dangling classes.

PureRank, as a parameter-free realization of the RDI principle, has clear strengths but also intrinsic limitations. A central advantage of PureRank is that it avoids the need for empirical or heuristic parameter tuning while still guaranteeing a unique score vector for any given network. PureRank can therefore serve as a neutral reference against which application-specific, parameter-dependent importance measures can be compared. At the same time, PureRank provides no tuning parameter to mitigate computational cost, and its scores may be sensitive to structural changes that alter the SCC decomposition. Section~\ref{sec:concluding_remarks} revisits these limitations, summarizes their implications, and outlines directions for future work.

The numerical evaluation in this paper uses three large real-world networks from the SNAP repository~\cite{snap_datasets}. We study a Twitter follower network (dominated by the transient class), a high-energy physics citation network (dominated by the transient class, with only seven recurrent nodes), and a co-authorship network (fully recurrent and dominated by a single large recurrent class). The three networks represent social, citation, and collaboration domains and exhibit markedly different node classifications. To compare PureRank with PageRank, we use three similarity measures: (i) Top-100 Overlap, (ii) Kendall's $\tau_b$, and (iii) the Pearson correlation coefficient (PCC). These measures reveal that the relationship between PureRank and PageRank depends strongly on the damping factor $d$ and on the node classification of the network. In the fully recurrent network, all three similarity measures increase monotonically with $d$ and indicate near coincidence at $d=0.999$. In the two transient-dominated networks, Kendall's $\tau_b$ increases with $d$, whereas Top-100 Overlap and PCC do not exhibit a monotonic dependence on $d$.

In addition to its deterministic formulation, PureRank admits a random-surfer model that is fully characterized by the submatrices of the normalized weight matrix $\vc{P}$. This model clarifies the scoring mechanism of PureRank, connects PureRank to Markov chain theory, and suggests potential customizations such as incorporating auxiliary information on nodes and links through the choice of initial distributions or transition rules. However, since the purpose of this paper is to introduce PureRank as a versatile parameter-free RDI-based measure, such extensions are beyond the scope of the present study and are left for future work.

We extend PureRank to networks with multiple link attributes, including biological networks whose links represent distinct actions such as activation or inhibition~\cite{Frah20}. The extension uses the {\it splitting-network construction}, which splits each node into copies indexed by link attributes and converts the original multi-attribute network into an enlarged single-attribute {\it splitting network}. PureRank is then computed on the resulting network. A related node-duplication mechanism was introduced in the Ising-PageRank model~\cite{Frah19} and later applied to biological networks in~\cite{Frah20}. However, the purposes differ: the splitting-network construction uses node splitting to encode multiple link attributes so that PureRank can be applied, whereas the Ising-PageRank model duplicates each node into two opinion-indexed copies to model opinion formation before applying PageRank.

The main contributions of this paper are summarized as follows.
\begin{enumerate}\renewcommand{\labelenumi}{\arabic{enumi}.}
\item We formulate PureRank as a parameter-free, RDI-based importance measure for arbitrary directed networks through node classification based on SCC decomposition and the out-degree condition, classwise computation of local importance vectors, and their aggregation into the PureRank vector by the RDI-L2G construction.
\item We develop a random-surfer interpretation of PureRank based on the submatrices of the normalized weight matrix $\vc{P}$, thereby clarifying the scoring mechanism of PureRank and linking PureRank to Markov chain theory.
\item We extend PureRank to multi-attribute networks via the splitting-network construction, which converts a multi-attribute network into a single-attribute splitting network and yields attribute-specific PureRank scores for each original node.
\item We numerically compare PureRank and PageRank on three large real-world networks from SNAP, focusing on computational cost and on how their ranking and scoring similarity vary with node classification and PageRank's damping factor.
\end{enumerate}
As described above, the primary contributions are theoretical and methodological. The empirical part is a quantitative evaluation based on numerical experiments on network datasets.

The remainder of this paper is organized as follows. Section~\ref{sec:PureRank} formulates PureRank and details its computational procedure. Section~\ref{sec:random_surfer} develops a random-surfer interpretation of PureRank from a Markov chain perspective. Section~\ref{sec:numerical_experiments} studies PureRank and PageRank through numerical experiments on three large real-world networks from SNAP~\cite{snap_datasets}, emphasizing computational cost and the effects of node classification and PageRank's damping factor on ranking and scoring similarity. Section~\ref{sec:multi_attribute_networks} extends PureRank to multi-attribute networks via the splitting-network construction. Finally, Section~\ref{sec:concluding_remarks} discusses implications, limitations, and future directions.

\section{Method: PureRank as a parameter-free RDI-based importance measure}\label{sec:PureRank}

This section formulates PureRank as a parameter-free RDI-based importance measure and presents its theoretical foundation and computational procedure. Section~\ref{subsec:node_classification} describes the node classification $\calC$ and presents the corresponding block-partitioned forms of $\vc{W}$ and $\vc{P}$. Section~\ref{subsec:local_centrality} characterizes, for each Class $S\in\calC$, the local importance vector through a bi-objective optimization problem for choosing the parameters of the Katz equation on $\ol{G}_S=(S,\vc{P}_S)$ according to the RDI principle. Section~\ref{subsec:L2G} presents the {\it RDI-based local-to-global (RDI-L2G) construction}, which aggregates the local importance vectors into the global importance vector, namely, the {\it PureRank vector}. Finally, Section~\ref{sec:computing} discusses the computation of PureRank.

\subsection{Classification of nodes}\label{subsec:node_classification}

This subsection describes the classification of nodes in the network $G=(V,\vc{W})$ and the associated notation. We define the recurrent, transient, and dangling classes for network nodes and introduce the notation $\calC$ for the resulting node classification. Based on this classification, we present the block-partitioned forms of $\vc{W}$ and $\vc{P}$. In particular, the submatrices in the block-partitioned form of $\vc{P}$ are used in Section~\ref{subsec:local_centrality} to define the local importance vector for each class.

We formally define the recurrent, transient, and dangling classes. Although the dangling class was already introduced in Section~\ref{sec:introduction}, we restate it here for completeness and in parallel with the definitions of the other two classes.
\begin{defn}[Dangling nodes]\label{defn:class_D}
Node $i \in V$ is said to be a {\it dangling node} if and only if
\begin{align*}
w_i^{out}=\sum_{j \in V} w_{i,j}=0.
\end{align*}
Equivalently, node $i$ has no outlinks to any node (including itself). The set of dangling nodes is denoted by $D$ and referred to as the dangling class, or Class $D$.
\end{defn}
\begin{defn}[Recurrent class]\label{defn:class_R}
A set $S \subseteq V \setminus D$ is said to be a {\it recurrent class} if and only if $S$ is a closed SCC of the network $G$; equivalently, $S$ satisfies the following conditions:
\begin{align*}
\sum_{n=1}^{\infty}[\vc{W}^n]_{i,j} &> 0,\quad \forall i,j \in S, \\
\sum_{n=1}^{\infty}[\vc{W}^n]_{i,j} &= 0,\quad \forall i \in S,\ j \in V \setminus S.
\end{align*}
Each node in a recurrent class is said to be recurrent. Let $K \in \bbZ_+:=\{0,1,2,\dots\}$ denote the number of recurrent classes, and let $R_k$ denote the $k$th recurrent class for $k \in \bbK:=\{1,2,\dots,K\}$. If $K=0$, then $\bbK=\varnothing$ and no recurrent classes exist. The recurrent classes $R_1,R_2,\dots,R_K$ are disjoint, and their union is denoted by $R:=\bigsqcup_{k=1}^K R_k$ (the symbol ``$\sqcup$'' denotes the disjoint union of sets). The set $R$ is referred to as Class $R$.
\end{defn}
\begin{defn}[Transient class]\label{defn:class_T}
A node $i \in V$ is said to be {\it transient} if and only if $i \not\in R \cup D$. The set of transient nodes is denoted by $T:=V \setminus (R \sqcup D)$ and referred to as the transient class, or Class $T$.
\end{defn}

\begin{rem}
Definitions~\ref{defn:class_D} and \ref{defn:class_R} distinguish between a single dangling node and a recurrent class consisting of a single node. The former has no outlinks and thus no self-loop, whereas the latter has a self-loop as its only outlink. Although one may regard these two types of nodes as identical, we do not adopt this convention in this paper.
\end{rem}

According to Definitions~\ref{defn:class_D}--\ref{defn:class_T}, the node set $V$ is partitioned into Classes $R_1$, $R_2$, $\dots$, $R_K$, $T$, and $D$. We denote this classification by
\begin{align}
\calC := \{R_1,R_2,\dots,R_K,T,D\}.
\label{eq:defn_C}
\end{align}
For notational convenience, we use the notation in (\ref{eq:defn_C}) throughout the paper. The formulation of PureRank and the analysis of its properties are most naturally stated in the general setting where multiple recurrent classes coexist with the transient and dangling classes, although some of these classes may in fact be absent in a given network. Sets introduced later may also be empty when they are determined by the empty classes in $\calC$. Unless otherwise stated, whenever such a set is empty, any associated quantities, including subvectors, submatrices, blocks, block rows, block columns, and terms, are understood to be absent. In particular, sums over empty index sets are interpreted as zero. This convention avoids repeated case distinctions regarding which classes in $\calC$ are empty.

The node classification $\calC$ of the network $G=(V,\vc{W})$ depends on the network connectivity. The classes $R_1,R_2,\dots,R_K,T,D$ are determined by the out-degree condition in Definition~\ref{defn:class_D} and the SCC decomposition of the subnetwork induced by $V\setminus D$ (see Algorithm~\ref{algo:node_classification}), and some of these classes may be empty. In fact, this classification depends only on the zero pattern of $\vc{W}$. Therefore, modifying link weights without changing the zero pattern leaves $\calC$ unchanged, whereas the addition or removal of nodes or links can change $\calC$ by altering the SCC decomposition or out-degrees.

We describe the connectivity among the recurrent, transient, and dangling classes by partitioning the weight matrix $\vc{W}$ with respect to the node classification $\calC = \{R_1,R_2,\dots,R_K,T,D\}$. Rearranging the rows and columns of $\vc{W}$ according to $\calC$, we write $\vc{W}$ in block-partitioned form:
\begin{align}
\vc{W}
=
\begin{blockarray}{c ccc ccc ccc}
       & R_1    & \cdots & R_K &   & T   &   &   & D   \\
\begin{block}{c (ccc|ccc|ccc)}
R_1          & \vc{W}_{R_1} &        & \mbox{\large $\vc{O}$}  &        &        &        &        &        \\
\vdots       &              & \ddots &                        &        & \mbox{\large $\vc{O}$}  &        &        & \mbox{\large $\vc{O}$}  \\
R_K          & \mbox{\large $\vc{O}$} &  & \vc{W}_{R_K} &        &        &        &        &        \\
\BAhhline{~---------}
             &              &        &              &        &        &        &        &        \\
T              & \vc{W}_{T,R_1} & \cdots & \vc{W}_{T,R_K} &        & \mbox{\large $\vc{W}_T$}  &        &        & \mbox{\large $\vc{W}_{T,D}$}  \\
             &              &        &              &        &        &        &        &        \\
\BAhhline{~---------}
             &              &        &              &        &        &        &        &        \\
D             &  \mbox{\large $\vc{O}$} & \cdots &  \mbox{\large $\vc{O}$}  &        & \mbox{\large $\vc{O}$}  &        &        & \mbox{\large $\vc{O}$}  \\
             &              &        &              &        &        &        &        &        \\
\end{block}
\end{blockarray}
~.
\label{partition_W}
\end{align}
Definition~\ref{defn:class_D} shows that all rows corresponding to Class $D$ are zero. Definition~\ref{defn:class_R} also states that, for each $k \in \bbK$, Class $R_k$ is a closed SCC and thus $\vc{W}_{R_k}$ is irreducible. Eq.~(\ref{partition_W}) illustrates that the connectivity structure of the network $G$ is as depicted in Fig.~\ref{fig:classification_of_V}.
\begin{figure}[tbh]
\centering
\includegraphics[bb=0 0 545 327, scale=0.4]{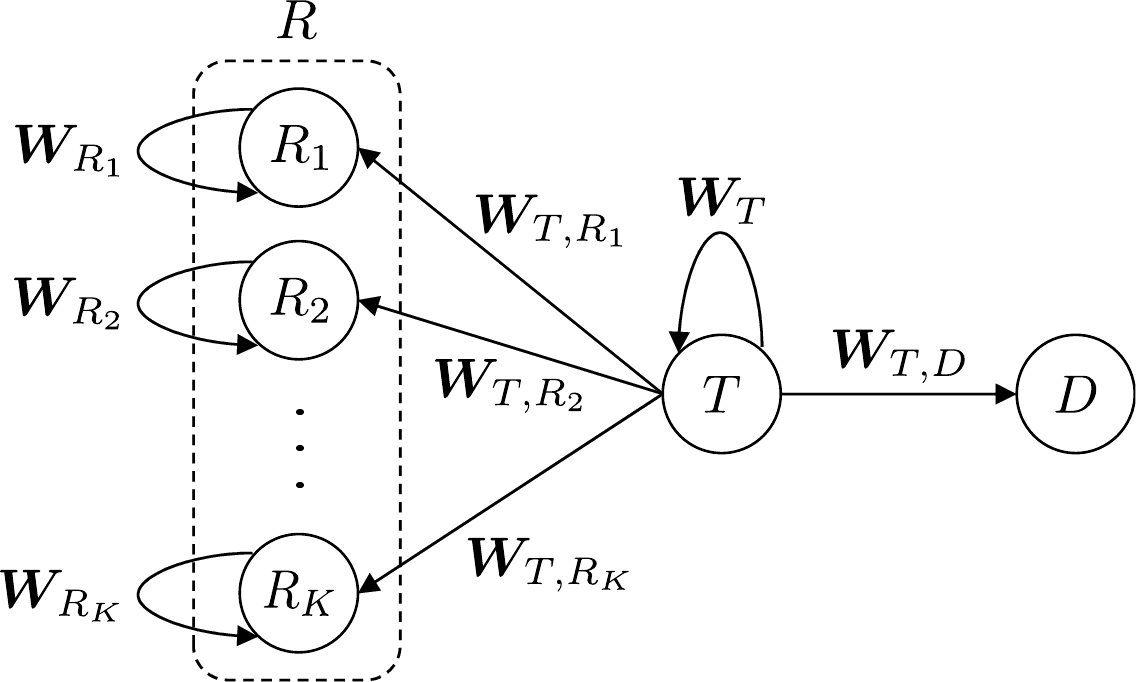}
\caption{Connectivity structure of $G=(V,\vc{W})$ under the node classification $\calC=\{R_1,R_2,\dots,R_K,T,D\}$.}
\label{fig:classification_of_V}
\end{figure}

As with $\vc{W}$, we partition the normalized weight matrix $\vc{P}$ in (\ref{defn:P}) using the same node classification $\calC$. The block-partitioned form of $\vc{P}$ is as follows:
\begin{align}
\vc{P}
=
\begin{blockarray}{c ccc ccc ccc}
       & R_1    & \cdots & R_K &   & T   &   &   & D   &   \\
\begin{block}{c(ccc|ccc|ccc)}
R_1          	&
\vc{P}_{R_1}	&
        		&
\mbox{\large $\vc{O}$}  &
        		&
				&
				&
				&
				&
\\
\vdots       	&
              	&
\ddots 			&
				&
				&
\mbox{\large $\vc{O}$}  &
				&
				&
\mbox{\large $\vc{O}$}  &
\\
R_K          	&
\mbox{\large $\vc{O}$} &
				&
\vc{P}_{R_K} 	&
				&
				&
				&
				&
				&
\\
\BAhhline{~---------}
             	&
				&
				&
				&
				&
				&
				&
				&
				&
\\
T              	&
\vc{P}_{T,R_1} 	&
\cdots 			&
\vc{P}_{T,R_K} 	&
				&
\mbox{\large $\vc{P}_T$} &
				&
				&
\mbox{\large $\vc{P}_{T,D}$} &
\\
             	&
				&
				&
				&
				&
				&
				&
				&
				&
\\
\BAhhline{~---------}
				&
				&
				&
				&
				&
				&
				&
				&
				&
\\
D             	&
\mbox{\large $\vc{O}$}  &
\cdots &
\mbox{\large $\vc{O}$}  &
				&
\mbox{\large $\vc{O}$}  &
				&
				&
\mbox{\large $\vc{P}_D$~($=\vc{O}$)}  &
\\
             	&
				&
				&
				&
				&
				&
				&
				&
				&
\\
\end{block}
\end{blockarray}
~,
\label{partition_P}
\end{align}
where each $\vc{P}_{R_k}$ is an irreducible stochastic matrix, as follows from (\ref{defn:P}) and the irreducibility of $\vc{W}_{R_k}$. For brevity, we also provide a more compact partition of $\vc{P}$:
\begin{align}
\vc{P}
&=
\begin{blockarray}{cccc}
  	& R        	   & T    	  & D \\
\begin{block}{c(ccc)}
  R & \vc{P}_R     & \vc{O}   & \vc{O} \\
  T & \vc{P}_{T,R} & \vc{P}_T & \vc{P}_{T,D} \\
  D & \vc{O}       & \vc{O}   & \vc{O} \\
\end{block}
\end{blockarray}
~,
\label{compact_partition_P}
\end{align}
where
\begin{align}
\vc{P}_R
&=
\begin{blockarray}{ccccc}
  & R_1 & R_2 & \cdots  & R_K \\
\begin{block}{c(cccc)}
R_1 & \vc{P}_{R_1} & \vc{O}       & \cdots       & \vc{O} \\
R_2 & \vc{O}       & \vc{P}_{R_2} & \ddots       & \vdots \\
\vdots & \vdots  & \ddots  & \ddots  & \vc{O} \\
R_K & \vc{O}  & \cdots  & \vc{O}  & \vc{P}_{R_K} \\
\end{block}
\end{blockarray}
~,
\label{defn:P_R}
\\
\vc{P}_{T,R}
&=
\begin{blockarray}{ccccc}
  & R_1 & R_2 & \cdots  & R_K \\
\begin{block}{c(cccc)}
T & \vc{P}_{T,R_1} & \vc{P}_{T,R_2} & \cdots       & \vc{P}_{T,R_K} \\
\end{block}
\end{blockarray}
.
\label{defn:P_{T,R}}
\end{align}
Eq.~(\ref{defn:P}) implies that $\vc{P}$ is substochastic. In addition, because each row of $\vc{P}$ indexed by $T$ sums to one, (\ref{partition_P}) and (\ref{defn:P_{T,R}}) yield
\begin{align}
\vc{e}_T
&=
\sum_{k=1}^K \vc{P}_{T,R_k}\vc{e}_{R_k}
+ \vc{P}_T\vc{e}_T
+ \vc{P}_{T,D}\vc{e}_D
\nonumber
\\
&= \vc{P}_{T,R}\vc{e}_R
+ \vc{P}_T\vc{e}_T
+ \vc{P}_{T,D}\vc{e}_D.
\label{sum_rows_in_T}
\end{align}

\begin{rem}\label{rem:P_T}
Definition~\ref{defn:class_T} implies that after rearranging its rows and columns, $\vc{P}_T$ has an upper block triangular form and each diagonal block corresponds to an open SCC contained in $T$. Let $U$ denote such an open SCC, and let $\vc{P}_U=(p_{i,j})_{i,j\in U}$. There exists some $m\in\bbN$ such that $(\vc{P}_U)^m\vc{e}_U<\vc{e}_U$, which implies that $\rho(\vc{P}_U)<1$. Therefore, the Neumann series $(\vc{I}-\vc{P}_T)^{-1}=\sum_{n=0}^{\infty}(\vc{P}_T)^n$ converges.
\end{rem}

\subsection{Local importance vectors}\label{subsec:local_centrality}

This subsection introduces the local importance vector for each class. For each nonempty $S\in\calC$, we consider the Katz equation on the class-restricted subnetwork $\ol{G}_S=(S,\vc{P}_S)$ and choose its Katz parameters through a bi-objective optimization problem. This parameter choice follows the RDI principle by making the equation as close as possible to the Seeley equation. The unique optimal solution of this problem yields the local importance vector $\vc{\lambda}_S$.

For each nonempty $S \in \calC$, we consider the following bi-objective optimization problem:
\begin{prob}\label{prob_local_importance}
\begin{subequations}\label{prob_1}
\begin{alignat}{2}
&\mbox{Minimize}    &  & \quad	(\delta_S,\beta_S)
\notag
\\
&\mbox{Subject to} 	&  & \quad  	\vc{x}_S^{\top} = (1-\delta_S) \vc{x}_S^{\top}\vc{P}_S + \beta_S\vc{e}_S^{\top},
\label{cond_01}
\\
&					&  & \quad  	\vc{x}_S^{\top}\vc{e}_S =1,
\label{cond_02}
\\					
&					&  & \quad  	\vc{x}_S \ge \vc{0},
\label{cond_03}
\\
&					&  & \quad  	0 \le \delta_S \le 1,\ \beta_S \ge 0.
\label{cond_04}	
\end{alignat}
\end{subequations}
\end{prob}

Problem~\ref{prob_local_importance} chooses the parameters $(\delta_S,\beta_S)$ in the Katz equation $(S,\vc{P}_S,1-\delta_S,\beta_S)$ on the class-restricted subnetwork $\ol{G}_S$ so that the resulting Katz equation is the smallest modification of the Seeley-centrality equation:
\begin{align}
\vc{x}_S^{\top}=\vc{x}_S^{\top}\vc{P}_S.
\label{eq:Seeley-centrality}
\end{align}
Accordingly, within Class $S$, Problem~\ref{prob_local_importance} seeks a feasible importance score vector that preserves, as far as possible, the RDI principle realized by Seeley centrality.

To characterize the optimal solution to Problem~\ref{prob_local_importance}, we introduce the class-specific quantities $\beta_S^*$ and $\vc{\lambda}_S$ for nonempty $S \in \calC$. Let $\beta_S^*$ be defined by
\begin{align}
\beta_S^* &= \left\{
\begin{array}{ll}
\dm{1 \over |D|}, & S=D,
\\
\rule[-2mm]{0mm}{7mm}
0, & S \in \{R_1,R_2,\dots,R_K\},
\\
\dm{1 \over \vc{e}_T^{\top}(\vc{I} - \vc{P}_T)^{-1}\vc{e}_T}, & S=T,
\end{array}
\right.
\label{defn:beta_S}
\end{align}
In addition, let $\vc{\lambda}_D$ and $\vc{\lambda}_T$ be defined by
\begin{align}
\vc{\lambda}_D 
&=\vc{\mu}_D={\vc{e}_D^{\top} \over |D|},
\label{defn:lambda_D}
\\
\vc{\lambda}_T
&=
{
\vc{\mu}_T(\vc{I} - \vc{P}_T)^{-1}
\over
\vc{\mu}_T
 (\vc{I} - \vc{P}_T)^{-1}\vc{e}_T
}
=
{
\vc{e}_T^{\top}(\vc{I} - \vc{P}_T)^{-1}
\over
\vc{e}_T^{\top}
 (\vc{I} - \vc{P}_T)^{-1}\vc{e}_T
}.
\label{defn:lambda_T}
\end{align}
For $k \in \bbK$, let $\vc{\lambda}_{R_k}$ denote the unique stationary distribution vector of $\vc{P}_{R_k}$, i.e., the unique vector such that
\begin{align}
\vc{\lambda}_{R_k} \vc{P}_{R_k} = \vc{\lambda}_{R_k},
\quad
\vc{\lambda}_{R_k} \vc{e}_{R_k} = 1,
\quad
\vc{\lambda}_{R_k} \ge \vc{0}^{\top}.
\label{defn:lambda_R_k}
\end{align} 

The optimal solution to Problem~\ref{prob_local_importance} is given by the following theorem.
\begin{thm}\label{thm:Katz-Centrality}
\hfill\\
For each nonempty $S\in\calC$, Problem~\ref{prob_local_importance} admits a unique optimal solution $(\vc{x}_S^{\top},\delta_S,\beta_S)=(\vc{\lambda}_S,0,\beta_S^*)$. Moreover, this optimal solution is independent of whether $\delta_S$ or $\beta_S$ is minimized first, because both orders select the same solution.
\end{thm}

\begin{proof}
See Appendix~\ref{proof:thm:Katz-Centrality}.
\end{proof}

\medskip

Theorem~\ref{thm:Katz-Centrality} motivates the following definition of the local importance vector of each class.
\begin{defn}[Local importance vectors]\label{defn:local_PureRank}
For $S \in \calC$, let $\lambda_{S}(j)$, $j\in S$, denote the {\it local importance} score of node $j$ within Class $S$. The vector $\vc{\lambda}_{S}:= (\lambda_{S}(j))_{j\in S}$ is referred to as the {\it local importance vector} of Class $S$.
\end{defn}

Theorem~\ref{thm:Katz-Centrality} shows that the unique optimal solution of Problem~\ref{prob_local_importance} satisfies $\vc{x}_S^{\top}=\vc{\lambda}_S$. Thus, the local importance vector is obtained directly from the optimal solution, without invoking generic optimization solvers. In particular, $\vc{\lambda}_D$ is the uniform distribution vector on $D$, each $\vc{\lambda}_{R_k}$ is given by (\ref{defn:lambda_R_k}), and $\vc{\lambda}_T$ is given by (\ref{defn:lambda_T}).

\begin{rem}
The uniqueness of $\vc{\lambda}_{R_k}$ follows from the irreducibility of $\vc{P}_{R_k}$ (see, e.g., \cite[Theorems~3.2.6 and 3.2.8]{Brem20}).
\end{rem}

\begin{rem}
Avrachenkov et al.~\cite{Avra10} proposed four centralities for the transient class (which they refer to as the extended SCC), one of which is equivalent to the local importance vector $\vc{\lambda}_T$ (see \cite[Definition~1]{Avra10}). However, their analysis focused on {\it local} centrality within the transient class and did not discuss a {\it global} centrality that accounts for connections between the transient class and other classes.
\end{rem}

\subsection{RDI-based local-to-global construction of PureRank}\label{subsec:L2G}

This subsection introduces PureRank, which measures the \textit{global} importance score of every node in $V$ by using its {\it local} importance score together with its class size and connectivity with Class $T$. We first present the definition of the PureRank vector $\vc{\pi}$, which contains the PureRank scores of nodes. We next describe the \textit{RDI-based local-to-global (RDI-L2G) construction}, which interprets the PureRank vector $\vc{\pi}$ as an aggregation of the local importance vectors $\{\vc{\lambda}_S\}_{S\in\calC}$ by incorporating class sizes and connectivity with Class $T$. We then discuss the relationship between $\{\vc{\lambda}_S\}_{S\in\calC}$ and the PageRank vector $\vc{\gamma}^{(d)}$ as $d\uparrow1$. Finally, we explain the rationale for the name PureRank.

We define PureRank as a parameter-free recursive importance measure for nodes in the network $G=(V,\vc{W})$. To this end, we introduce a key quantity $\theta_T$ defined by
\begin{align}
\theta_T
&= \vc{\lambda}_T\vc{P}_{T,R}\vc{e}_R + \vc{\lambda}_T\vc{P}_{T,D}\vc{e}_D
= 1 - \vc{\lambda}_T\vc{P}_T\vc{e}_T,
\label{eqn:theta_T}
\end{align}
which quantifies the total amount of local importance that leaks from Class $T$ when $\vc{\lambda}_T$ is propagated by $\vc{P}_T$ in one step. The second equality in (\ref{eqn:theta_T}) follows from (\ref{defn:P_{T,R}}) and (\ref{sum_rows_in_T}). Using this key quantity $\theta_T$, we present the definition of PureRank below.
\begin{defn}\label{defn:PureRank}
Let $\pi_j$, $j \in V$, denote the \textit{PureRank} score of node $j$, and let $\vc{\pi}:=(\pi_j)_{j \in V}$ be the PureRank vector. For each $S \in \calC$, the PureRank subvector $\vc{\pi}_S:=(\pi_j)_{j\in S}$ of Class $S$ is constructed from the local importance vectors $\{\vc{\lambda}_S\}_{S\in\calC}$ as follows:
\begin{subequations}\label{eqn:PureRank_vectors}
\begin{align}
\vc{\pi}_{R_k}
&=
{1 \over N}
\left( |R_k|  \vc{\lambda}_{R_k}
+ { |T| \over 1 + \theta_T} \vc{\lambda}_T \vc{P}_{T,R_k}
\right),
\quad k \in \bbK,
\label{formula_pi_R_k}
\\
\vc{\pi}_D
&=
{1 \over N}
\left( |D| \vc{\mu}_D
+ { |T|  \over 1 + \theta_T} \vc{\lambda}_T\vc{P}_{T,D}
\right),
\label{formula_pi_D}
\\
\vc{\pi}_T
&=
{1 \over N}
{|T|  \over 1 + \theta_T} \vc{\lambda}_T.
\label{formula_pi_T}
\end{align}
\end{subequations}
Furthermore, (\ref{formula_pi_R_k}) and (\ref{formula_pi_D}) can be rewritten as
\begin{subequations}
\begin{align*}
\vc{\pi}_{R_k}
&= {|R_k| \over N}\vc{\lambda}_{R_k} + \vc{\pi}_T \vc{P}_{T,R_k},
\quad k \in \bbK,
\\
\vc{\pi}_D
&=
{|D| \over N} \vc{\mu}_D + \vc{\pi}_T\vc{P}_{T,D}.
\end{align*}
\end{subequations}
\end{defn}

\begin{rem}\label{rem:strongly_connected}
Suppose that the network $G=(V,\vc{W})$ is strongly connected. The strong connectivity of $G$ implies that $V=R_1$, $K=1$, $T=D=\varnothing$, and $\vc{P}=\vc{P}_{R_1}$. Definition~\ref{defn:PureRank} together with (\ref{defn:lambda_R_k}) yields
\begin{align*}
\vc{\pi}
=
\vc{\pi}_{R_1}
= { |R_1| \over N }\vc{\lambda}_{R_1}
= \vc{\lambda}_{R_1},
\end{align*}
because $|R_1|=N$. Hence, the PureRank vector $\vc{\pi}$ coincides with the unique stationary distribution vector of $\vc{P}$, namely, the Seeley centrality vector of $G=(V,\vc{W})$.
\end{rem}

\begin{rem}\label{rem:theta_T}
By (\ref{eqn:theta_T}), (\ref{defn:lambda_T}), (\ref{defn:beta_S}), and $\vc{\mu}_T\vc{e}_T=1$, we obtain
\begin{align}
\theta_T
&= 1 - \vc{\lambda}_T\vc{P}_T\vc{e}_T
= 1 -
{ \vc{\mu}_T(\vc{I} - \vc{P}_T)^{-1}\vc{P}_T\vc{e}_T \over
\vc{\mu}_T(\vc{I} - \vc{P}_T)^{-1}\vc{e}_T}
\nonumber
\\
&=
{ \vc{\mu}_T(\vc{I} - \vc{P}_T)^{-1}(\vc{I} - \vc{P}_T)\vc{e}_T
\over
\vc{\mu}_T(\vc{I} - \vc{P}_T)^{-1}\vc{e}_T}
= {1 \over \vc{\mu}_T(\vc{I} - \vc{P}_T)^{-1}\vc{e}_T}
\nonumber
\\
&= {|T| \over \vc{e}_T^{\top}(\vc{I} - \vc{P}_T)^{-1}\vc{e}_T}
= |T|\beta_T^*.
\label{eqn:remark_theta_T}
\end{align}
Therefore, $\theta_T$ equals the total baseline score of Class $T$, namely, $|T|\beta_T^*$. Combining (\ref{defn:lambda_T}) with (\ref{eqn:remark_theta_T}) further yields
\begin{align*}
\vc{\lambda}_T(\vc{I} - \vc{P}_T)
= { \vc{\mu}_T \over \vc{\mu}_T(\vc{I} - \vc{P}_T)^{-1}\vc{e}_T}
= \theta_T\vc{\mu}_T.
\end{align*}
Equivalently,
\begin{align}
\vc{\lambda}_T
= \vc{\lambda}_T\vc{P}_T + \theta_T\vc{\mu}_T.
\label{eq:decomposition_theta_T}
\end{align}
Hence, the local importance vector $\vc{\lambda}_T$ of Class $T$ decomposes into two terms: (i) $\vc{\lambda}_T\vc{P}_T$, the importance redistributed within Class $T$, and (ii) $\theta_T\vc{\mu}_T$, the uniform compensation for the leakage from Class $T$.
\end{rem}

We now explain the RDI-L2G construction of the PureRank vector $\vc{\pi}$. The RDI-L2G construction interprets Definition~\ref{defn:PureRank} by decomposing the PureRank score into the direct contribution from classwise local importance and the leakage of importance from Class $T$ to Classes $R_1,R_2,\dots,R_K$ and $D$ (see Fig.~\ref{fig:RDI_L2G}).
\begin{enumerate}
\item \textbf{Transient Class ($T$):}
Class $T$ is not closed, and hence local importance leaks from Class $T$. Remark~\ref{rem:theta_T} shows that $\vc{\lambda}_T$ decomposes into the local importance redistributed within Class $T$ and a uniform compensation for the leakage, and the total amount of this compensation is equal to $\theta_T=|T|\beta_T^*>0$. We therefore define the PureRank subvector $\vc{\pi}_T$ by scaling $\vc{\lambda}_T$ by the class size $|T|$ and the factor $1/(1+\theta_T)$. The factor $1/(1+\theta_T)$ offsets the compensation incorporated into the local importance vector $\vc{\lambda}_T$. This adjustment places the contribution of Class $T$ on the same basis as those of the other classes, whose local importance vectors contain no such compensation. Accordingly,
\begin{align}
\vc{\pi}_T = {1 \over Z} {|T| \over 1 + \theta_T}\vc{\lambda}_T,
\label{defn:pi_T}
\end{align}
where $Z>0$ is a normalizing constant (equal to $N$, as shown in Theorem~\ref{thm:Z=N}).

\item \textbf{Recurrent Classes ($R_1,R_2,\dots,R_K$):}
Each Class $R_k$ ($k=1,2,\dots,K$) is a closed SCC, and hence local importance does not leak from Class $R_k$. Theorem~\ref{thm:Katz-Centrality} shows that $\beta_{R_k}^*=0$, and therefore the local importance vector $\vc{\lambda}_{R_k}$ contains no compensation term. Hence, the PureRank subvector $\vc{\pi}_{R_k}$ consists of the size-scaled local importance vector of Class $R_k$ (with no offsetting factor applied) and the leakage from Class $T$ into Class $R_k$:
\begin{align}
\vc{\pi}_{R_k}
&= {1 \over Z} \left(
|R_k| \vc{\lambda}_{R_k}
+ {|T| \over 1 + \theta_T} \vc{\lambda}_T\vc{P}_{T,R_k}
\right),
\qquad k=1,2,\dots,K.
\label{defn:pi_{R_k}}
\end{align}

\item \textbf{Dangling Class ($D$):}
Class $D$ consists of nodes with no outlinks, and thus $\vc{P}_D=\vc{O}$. Hence, local importance does not leak from Class $D$. Unlike the recurrent classes, Theorem~\ref{thm:Katz-Centrality} shows that $\beta_D^*=1/|D|>0$. This positive baseline score reflects the presence of the dangling nodes rather than compensation for leakage. Accordingly, the PureRank subvector $\vc{\pi}_D$ consists of the size-scaled local importance vector of Class $D$ and the leakage from Class $T$ into Class $D$:
\begin{align}
\vc{\pi}_D
&= {1 \over Z} \left(
|D| \vc{\lambda}_D
+ {|T| \over 1 + \theta_T} \vc{\lambda}_T \vc{P}_{T,D}
\right).
\label{defn:pi_D}
\end{align}

\item \textbf{Normalization:}
Finally, we normalize the PureRank subvectors so that
\begin{align}
\sum_{k=1}^K \vc{\pi}_{R_k}\vc{e}_{R_k} + \vc{\pi}_T\vc{e}_T + \vc{\pi}_D\vc{e}_D = 1.
\label{eqn:sum_lambda_S}
\end{align}
\end{enumerate}

\begin{figure}[tbh]
\centering
\includegraphics[bb=0 0 788 370, scale=0.4]{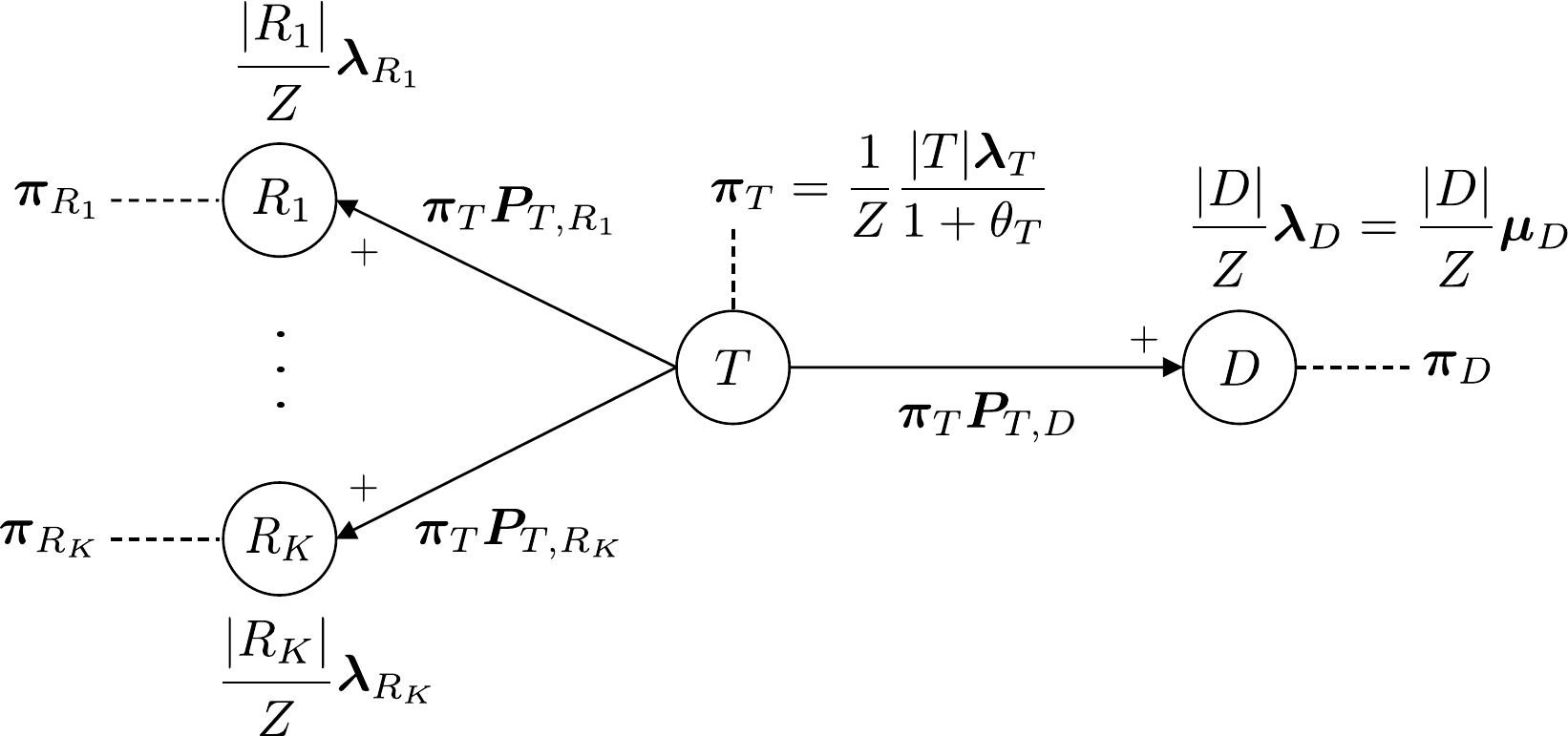}
\caption{RDI-based local-to-global (RDI-L2G) construction of the PureRank vector ($Z$ is the normalizing constant).}
\label{fig:RDI_L2G}
\end{figure}

Theorem~\ref{thm:Z=N} shows that $Z=N$, which confirms that the RDI-L2G construction realizes the PureRank vector in Definition~\ref{defn:PureRank}.
\begin{thm}\label{thm:Z=N}
The normalizing constant $Z$ is equal to the total number $N$ of nodes.
\end{thm}

\begin{proof}
See Appendix~\ref{proof:thm:GPR-01}.
\end{proof}

\medskip

The following theorem characterizes the limits as $d\uparrow1$ of the normalized classwise PageRank subvectors. Related numerical experiments and discussion are presented in Section~\ref{subsec:similarity_ranking_scoring}.
\begin{thm}\label{thm:classwise_PR_limit}
For each nonempty $S \in \mathcal{C}$, let $\widetilde{\vc{\gamma}}_S^{(d)}$ be defined by
\begin{align}
\widetilde{\vc{\gamma}}_S^{(d)} = {\vc{\gamma}_S^{(d)} \over \vc{\gamma}_S^{(d)}\vc{e}_S},\label{eq:widetilde{gamma}_S^{(d)}}
\end{align}
where $\vc{\gamma}_S^{(d)} := (\gamma_j^{(d)})_{j\in S}$ denotes the Class-$S$ subvector of the PageRank vector $\vc{\gamma}^{(d)}$. The following statements hold:
\begin{enumerate}
\item For each nonempty $R_k \in \calC$,
\begin{align}
\lim_{d \uparrow 1}\widetilde{\vc{\gamma}}_{R_k}^{(d)} = \vc{\lambda}_{R_k},
\label{eq:lim_tilde_gamma_R_k}
\end{align}
where $\vc{\lambda}_{R_k}$ is the stationary distribution vector of $\vc{P}_{R_k}$.
\item If $T \neq \varnothing$, then
\begin{align}
\lim_{d \uparrow 1}\widetilde{\vc{\gamma}}_T^{(d)} = \vc{\lambda}_{T},
\label{eq:lim_tilde_gamma_T}
\end{align}
where $\vc{\lambda}_T$ is defined in (\ref{defn:lambda_T}).
\item If $D \neq \varnothing$, then
\begin{align}
\lim_{d \uparrow 1}\widetilde{\vc{\gamma}}_D^{(d)}
= 
\left\{
\begin{array}{ll}
\dm{ 
|D|\vc{\mu}_D + \dm{|T| \vc{\lambda}_T \vc{P}_{T,D} \over \theta_T}
\over
|D| + \dm{|T| \vc{\lambda}_T \vc{P}_{T,D}\vc{e}_D \over \theta_T}
}, & \mbox{$T \neq \varnothing$},
\\
\vc{\mu}_D, & \mbox{$T = \varnothing$},
\end{array}
\right.
\label{eq:lim_tilde_gamma_D}
\end{align}
where $\theta_T$ is defined in (\ref{eqn:theta_T}). 
\end{enumerate}
\end{thm}

\begin{proof}
See Appendix~\ref{proof:thm:classwise_PR_limit}. 
\end{proof}

\medskip

To conclude this subsection, we explain the rationale for the name \textit{PureRank}. The term \textit{PureRank} reflects the absence of tunable parameters in its formulation. The PureRank vector is obtained in three steps: (i) classifying nodes into recurrent, transient, and dangling classes using SCC decomposition together with the out-degree condition, (ii) constructing the local importance vector for each class from the unique optimal solution of Problem~\ref{prob_local_importance}, and (iii) aggregating the local importance vectors into the PureRank vector $\vc{\pi}$ through the RDI-L2G construction. PureRank is free from empirical or heuristic parameter tuning, such as selecting a damping factor or inserting teleportation links, and thus preserves ``purity'' as a recursive importance measure, since it relies only on the network data encoded in $\vc{W}$ and $\vc{P}$. This purity also makes PureRank a neutral reference when comparing application-specific, parameter-dependent importance measures. The following remark further supports this purity by examining the specific case of fully recurrent networks.

\begin{rem}\label{rem:PureRank_in_fully_recurrent_NW}
Suppose that $G$ is fully recurrent, that is, $T=\varnothing$ and $D=\varnothing$. Under this assumption, $\vc{P}$ reduces to $\vc{P}_R$ in (\ref{defn:P_R}), and Definition~\ref{defn:PureRank} implies that
\begin{align}
\vc{\pi}
={1 \over N}
\left(
|R_1|\vc{\lambda}_{R_1},|R_2|\vc{\lambda}_{R_2},\dots,|R_K|\vc{\lambda}_{R_K}
\right),
\label{eq:PureRank_for_full_recurrent}
\end{align}
where each $\vc{\lambda}_{R_k}$ is the Seeley centrality vector for Class $R_k$. Thus, (\ref{eq:PureRank_for_full_recurrent}) shows that, under the RDI principle, PureRank first scores each recurrent class independently by Seeley centrality and then forms $\vc{\pi}$ by weighting the resulting vectors by class size and placing them in the corresponding block positions. This construction requires neither parameter tuning nor modification of the network structure and is therefore ``pure''. By contrast, PageRank connects all nodes through teleportation and therefore does not score the recurrent classes independently, even when $G$ is fully recurrent. Nevertheless, in the present special case, $\vc{\gamma}^{(d)}$ converges to the vector $\vc{\pi}$ in (\ref{eq:PureRank_for_full_recurrent}) as $d\uparrow1$. Indeed, since $\vc{Q}_V=\vc{P}=\vc{P}_R$, (\ref{defn:PageRank-02}) and (\ref{defn:P_R}) yield
\begin{align*}
\vc{\gamma}^{(d)}
&=
\left(
{|R_1| \over N}\vc{\mu}_{R_1}(1-d)(\vc{I}-d\vc{P}_{R_1})^{-1},
{|R_2| \over N}\vc{\mu}_{R_2}(1-d)(\vc{I}-d\vc{P}_{R_2})^{-1},
\right.\\
&\hspace{2em}\left.
\dots,
{|R_K| \over N}\vc{\mu}_{R_K}(1-d)(\vc{I}-d\vc{P}_{R_K})^{-1}
\right).
\end{align*}
Combining this with (\ref{eq:Abel_limit_R_k}), and using (\ref{eq:PureRank_for_full_recurrent}), we have
\begin{align*}
\lim_{d\uparrow1}
\vc{\gamma}^{(d)}
= {1 \over N}\bigl(|R_1|\vc{\lambda}_{R_1},|R_2|\vc{\lambda}_{R_2},\dots,|R_K|\vc{\lambda}_{R_K}\bigr)=\vc{\pi}.
\end{align*}
\end{rem}

\subsection{A procedure for computing PureRank}\label{sec:computing}

This subsection presents a modular procedure for computing the PureRank vector $\vc{\pi}$ of the target network $G=(V,\vc{W})$. The procedure uses the node classification $\calC$ introduced in Section~\ref{subsec:node_classification} to compute the local importance vectors $\{\vc{\lambda}_S\}_{S\in\calC}$ separately for each class and to aggregate them into $\vc{\pi}$ via (\ref{eqn:PureRank_vectors}). The modular structure of the procedure decomposes the major computations into classwise tasks. This decomposition clarifies which tasks can be performed in parallel. The decomposition also shows which quantities may need to be recomputed after a local update of the network, provided that the node classification remains unchanged.

We describe the procedure that generates the node classification $\calC=\{R_1,R_2,\dots,R_K,T,D\}$. The procedure identifies Class $D$ from the out-degrees and then applies SCC decomposition to the subnetwork induced by $V\setminus D$. Standard SCC algorithms, such as Kosaraju's algorithm \cite[Algorithm~4.6]{Sedg11} and Tarjan's algorithm \cite{Tarj72}, run in $O(|V|+|E|)$ time, where $E=\{(i,j)\in V^2:w_{i,j}>0\}$ is the set of links in the network $G$ and $O(\cdot)$ denotes Big-O notation. Each SCC $S_1,S_2,\dots,S_L$ in $V\setminus D$ is assigned to a class according to its connectivity: each closed SCC forms Class $R_k$ for some $k\in\bbK$, whereas each non-closed SCC belongs to Class $T$. Algorithm~\ref{algo:node_classification} summarizes this classification procedure.
\begin{algorithm}[htb]
\caption{Node classification procedure.}
\label{algo:node_classification}
\hfill\\
\textbf{Input}: Network $G = (V, \vc{W})$ \\
\textbf{Output}: Node classification $\calC=\{R_1, R_2, \dots, R_K, T, D\}$

\begin{algorithmic}[1]
\STATE \textbf{Step 1: Identify Dangling Nodes}
\FOR{each node $i \in V$}
    \IF{$w_i^{out} = 0$}
        \STATE Assign $i$ to Class $D$
    \ENDIF
\ENDFOR

\STATE \textbf{Step 2: SCC Decomposition (excluding $D$)}
\STATE Identify all SCCs $\{S_1, S_2, \dots, S_L\}$ in $V \setminus D$  ($L$: number of SCCs)

\STATE \textbf{Step 3: Assign Classes to Each SCC}
\STATE $k \leftarrow 1$
\FOR{each SCC $S_{\ell}$ ($\ell=1,2,\dots,L$)}
    \IF{$S_{\ell}$ has no outgoing links to any node outside $S_{\ell}$}
        \STATE Assign all nodes in $S_{\ell}$ to Class $R_k$
        \STATE $k \leftarrow k + 1$
    \ELSE
        \STATE Assign all nodes in $S_{\ell}$ to Class $T$
    \ENDIF
\ENDFOR
\end{algorithmic}
\end{algorithm}

We compute the local importance vectors $\{\vc{\lambda}_S\}_{S\in\calC}$ separately for the classes in $\calC$. This classwise construction decomposes the dominant iterative computations into independent subproblems. The computation of $\vc{\lambda}_D$ is immediate because $\vc{\lambda}_D=\vc{\mu}_D$. For each nonempty $S\in\calC\setminus\{D\}$, we compute $\vc{\lambda}_S$ from the corresponding principal submatrix $\vc{P}_S$. These computations are independent across distinct $S\in\calC\setminus\{D\}$ and can therefore be carried out in parallel.

For each Class $R_k$, the local importance vector $\vc{\lambda}_{R_k}$ is the stationary distribution vector of the irreducible stochastic matrix $\vc{P}_{R_k}$. We therefore compute $\vc{\lambda}_{R_k}$ by the power method. If $\vc{P}_{R_k}$ is periodic, the plain power iteration may fail to converge for some initial distributions. In that case, convergence is ensured by applying the power method to the aperiodic modification $(1-c)\vc{P}_{R_k}+c\vc{I}_{R_k}$ for some $c\in(0,1/2]$. This modification preserves the stationary distribution vector $\vc{\lambda}_{R_k}$ and serves only to guarantee convergence of the iteration. In the degenerate case $|R_k|=1$, $\vc{\lambda}_{R_k}$ reduces to the scalar $1$.

For Class $T$, the local importance vector $\vc{\lambda}_T$ can be computed by a convergent iterative procedure, which avoids the direct use of (\ref{defn:lambda_T}) and the inverse matrix $(\vc{I}-\vc{P}_T)^{-1}$.
\begin{thm}\label{thm:LPR-class_T}
\hfill
\begin{enumerate}
\item $\vc{\lambda}_T$ is the unique stationary probability vector of $\vc{Q}_T$, where $\vc{Q}_T$ denotes the stochastic matrix defined by
\begin{align}
\vc{Q}_T = \vc{P}_T + (\vc{e}_T - \vc{P}_T\vc{e}_T) \vc{\mu}_T.
\label{defn:Q_T}
\end{align}
Furthermore, $\vc{Q}_T$ is irreducible and aperiodic. Therefore, by \cite[Theorem 4.2.1]{Brem20},
\begin{align}
\lim_{n\to \infty} (\vc{Q}_T)^n = \vc{e}_T\vc{\lambda}_T.
\label{eqn:lim_(Q_T)^n}
\end{align}
\item Moreover, for any distribution vector $\vc{\xi}_T$ on Class $T$, $\vc{\lambda}_T$ is the limit of the sequence $\{\vc{\lambda}_T(n):n\in\bbZ_+\}$ generated by the recursion
\begin{subequations}\label{lim:lambda_T-02}
\begin{align}
\vc{\lambda}_T(0) &= \vc{\xi}_T,
\label{lim:lambda_T-02a}
\\
\vc{\lambda}_T(n+1) &= \vc{\lambda}_T(n)\vc{P}_T + [1 - \vc{\lambda}_T(n)\vc{P}_T\vc{e}_T] \vc{\mu}_T, \quad n=0,1,\dots.
\label{lim:lambda_T-02b}
\end{align}
\end{subequations}
\end{enumerate}
\end{thm}

\begin{proof}
See Appendix~\ref{proof:thm:LPR-class_T}.
\end{proof}

\medskip

\begin{algorithm}[htb]
\caption{Computation of the PureRank vector.}\label{algo:PureRank}
\hfill\\
\textbf{Input}: Network $G = (V, \vc{W})$ \\
\textbf{Output}: PureRank vector $\vc{\pi} =(\vc{\pi}_S)_{S \in \calC}$  ($\calC$: Node classification)

\begin{algorithmic}[1]
\STATE \textbf{Step 1: Node Classification}
\STATE \quad Obtain $\calC = \{R_1, R_2, \dots, R_K, T, D\}$ by Algorithm~\ref{algo:node_classification}, where $K \in \bbZ_+$.
\STATE \quad If $K=0$, the family $\{R_k\}_{k=1}^K$ is empty.

\STATE \textbf{Step 2: Computation of Local Importance Vectors $\{\vc{\lambda}_S\}_{S\in\calC}$}
\STATE \quad Compute the normalized matrix $\vc{P}$ by (\ref{defn:P}).
\STATE \quad Partition $\vc{P}$ as in (\ref{partition_P}).
\STATE \quad \textit{(a):} If $D\neq\varnothing$, compute $\vc{\lambda}_D$ using (\ref{defn:lambda_D}).
\STATE \quad \textit{(b):} \textbf{For} $k=1,2,\dots,K$ \textbf{do}
\STATE \qquad\qquad\quad Compute $\vc{\lambda}_{R_k}$ as the stationary distribution vector of $\vc{P}_{R_k}$.
\STATE \quad\qquad \textbf{EndFor}
\STATE \quad \textit{(c):} If $T\neq\varnothing$, compute $\vc{\lambda}_T$ by iterating the recursion in (\ref{lim:lambda_T-02}).

\STATE \textbf{Step 3: Aggregation into PureRank Vector $\vc{\pi} =(\vc{\pi}_S)_{S \in \calC}$}
\STATE \quad \textbf{For} each nonempty $S \in \calC$ \textbf{do}
\STATE \qquad\quad Compute $\vc{\pi}_S$ using (\ref{eqn:PureRank_vectors}).
\STATE \quad \textbf{EndFor}
\end{algorithmic}
\end{algorithm}

We summarize the procedure for computing PureRank in Algorithm~\ref{algo:PureRank} and illustrate its modular design in Fig.~\ref{fig:PureRank_framework}. Step~1 performs a global preprocessing stage through node classification. Step~2 computes the local importance vectors $\{\vc{\lambda}_S\}_{S\in\calC}$ classwise. These classwise computations are independent and can therefore be carried out in parallel. Step~3 aggregates the local importance vectors into the PureRank vector $\vc{\pi}$ through (\ref{eqn:PureRank_vectors}). This step can also be parallelized across classes once the local importance vectors and the required submatrices of $\vc{P}$ are available. After a local update of the network, recomputation can be restricted to the affected classes, provided that the node classification remains unchanged.

We discuss the computational cost of Algorithm~\ref{algo:PureRank} on large-scale networks. Step~1 essentially performs an SCC decomposition in time linear in the number of nodes and links, and its runtime is often smaller than those of Steps~2 and 3. Step~3 consists of a non-iterative aggregation, whereas Step~2 involves power iterations to compute the local importance vectors for the recurrent and transient classes. Thus, the dominant cost of computing $\vc{\pi}$ typically arises from the iterative computation of the local importance vector for the largest class among the recurrent and transient classes.

\begin{figure}[tbh]
\centering
\includegraphics[bb=0 0 523 228, scale=0.65]{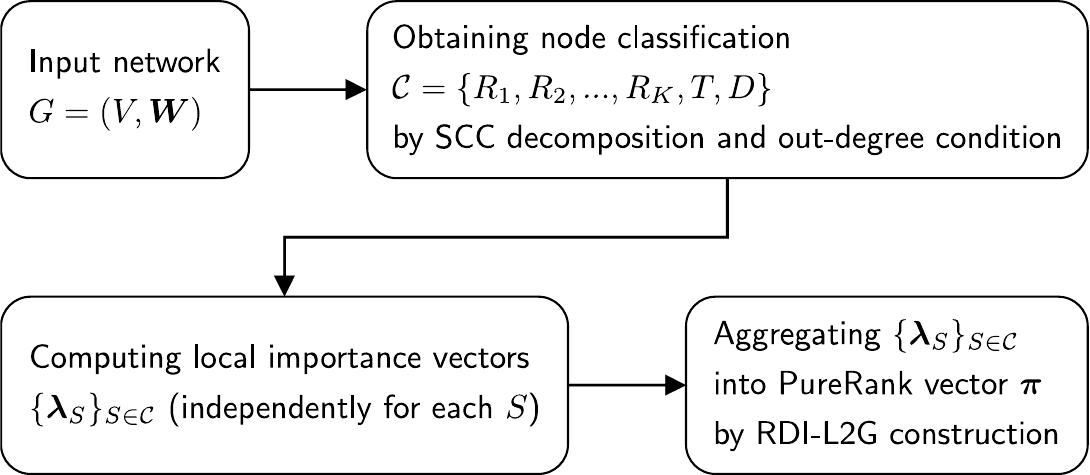}
\caption{Modular procedure for computing the PureRank vector $\vc{\pi}$ for $G=(V,\vc{W})$.}
\label{fig:PureRank_framework}
\end{figure}

\section{Interpretation: the random-surfer model of PureRank}\label{sec:random_surfer}

This section develops a random-surfer interpretation of the PureRank vector $\vc{\pi}=(\pi_j)_{j\in V}$. We first introduce a Markov transition matrix $\vc{M}$ and establish a key relation between $\vc{M}$ and $\vc{\pi}$. We then formulate a Markov chain on an extended state space, driven by $\vc{M}$, and use its long-run behavior to characterize PureRank probabilistically. This probabilistic perspective clarifies that PureRank is a parameter-free recursive importance measure determined solely by the network structure.

To define the Markov transition matrix $\vc{M}$, we first introduce auxiliary notation. When $T\neq\varnothing$, let $j'$ denote the copy of each $j \in R \sqcup D$, and let
\begin{align*}
D' &= \{j'; j \in D\},
\\
R' &= \bigsqcup_{k=1}^K R_k',\quad R_k' = \{j'; j \in R_k\},\quad k \in \bbK,
\end{align*}
where $D'$, $R'$, and $R_k'$ denote the copies of $D$, $R$, and $R_k$, respectively. When $T=\varnothing$, we set $D'=\varnothing$ and $R_k'=\varnothing$ for all $k \in \bbK$, so that $R'=\varnothing$. We also introduce the sets $\widehat{T}$ and $\widehat{V}$ by
\begin{align*}
\widehat{T} &=
\left\{
\begin{array}{ll}
T \sqcup R' \sqcup D', & \mbox{if $T\neq\varnothing$},
\\
\varnothing, & \mbox{if $T=\varnothing$},
\end{array}
\right.
\\
\widehat{V} &= R \sqcup \widehat{T} \sqcup D,
\end{align*}
and refer to $\widehat{T}$ and $\widehat{V}$ as the {\it extended transient class} and {\it extended node set}, respectively. When $T=\varnothing$, the sets $R'$, $D'$, and $\widehat{T}$ are empty; accordingly, by the convention introduced in Section~\ref{subsec:node_classification}, the blocks, block rows, and block columns labeled by $T$, $\widehat{T}$, $R'$, and $D'$ are understood to be absent.

We now define the Markov transition matrix $\vc{M}$ as the $|\widehat{V}| \times |\widehat{V}|$ stochastic matrix
\begin{align}
\vc{M}
&=
\begin{blockarray}{cccccc}
 & R & R'& T & D' & D \\ 
\begin{block}{c(c|ccc|c)}
R        & 
\vc{P}_R &
\vc{O}   &  
\vc{O}   & 
\vc{O}   &  
\vc{O}
\\  
\BAhhline{~-----}
R'        & 
\vc{O}   &  
\vc{O}   &  
\vc{e}_R\vc{\mu}_T   & 
\vc{O}   &  
\vc{O}
\\  
T        & 
\vc{O}   & 
\vc{P}_{T,R}   &  
\vc{P}_T   & 
\vc{P}_{T,D}   &  
\vc{O}
\\  
D'        & 
\vc{O}   &  
\vc{O}   &  
\vc{e}_D\vc{\mu}_T   & 
\vc{O}   &  
\vc{O}
\\ 
\BAhhline{~-----}
D        & 
\vc{O}   &  
\vc{O}   & 
\vc{O}   &  
\vc{O}   &
\vc{I}_D
\\
 \end{block}
\end{blockarray}~
=
\begin{blockarray}{c@{~~~}ccc}
 & R& \widehat{T} & D  \\ 
\begin{block}{c@{~~~}(ccc)}
R~        	& 
\vc{P}_R	& 
\vc{O}   	&   
\vc{O}   	 
\\  
\widehat{T}        & 
\vc{O}   &  
\vc{M}_{\widehat{T}}   & 
\vc{O}   
\\  
D~        	&  
\vc{O}   	&  
\vc{O}   	& 
\vc{I}_D   	
\\  
 \end{block}
\end{blockarray}
~,
\label{defn:matrix_M}
\end{align}
where $\vc{P}_R$ and $\vc{P}_{T,R}$ are given in (\ref{defn:P_R}) and (\ref{defn:P_{T,R}}), respectively, and $\vc{M}_{\widehat{T}}$ denotes the principal submatrix of $\vc{M}$ indexed by $\widehat{T}$, namely,
\begin{align}
\vc{M}_{\widehat{T}}
=
\begin{blockarray}{c@{~~}ccc}
 & R'& T & D'  \\ 
\begin{block}{c@{~~}(ccc)}
R'~        & 
\vc{O}   &  
\vc{e}_R\vc{\mu}_T   & 
\vc{O}    
\\  
T~        & 
\vc{P}_{T,R}   &  
\vc{P}_T   & 
\vc{P}_{T,D}    
\\  
D'~        &  
\vc{O}   &  
\vc{e}_D\vc{\mu}_T   & 
\vc{O}   
\\  
 \end{block}
\end{blockarray}
~.
\label{defn:M_T^{(e)}}
\end{align}
Eq.~(\ref{defn:matrix_M}) shows that $\vc{M}_{\widehat{T}}$ is stochastic whenever $T\neq\varnothing$.

The following lemma establishes the connection between the Markov transition matrix $\vc{M}$ and the PureRank vector $\vc{\pi}$.
\begin{lem}\label{lem:random_surfer-01}
\hfill
\begin{enumerate}
\item If $T \neq \varnothing$, then the stochastic matrix $\vc{M}_{\widehat{T}}$ has the unique stationary probability vector $\vc{\lambda}_{\widehat{T}}$, where
\begin{align}
\vc{\lambda}_{\widehat{T}} 
&= {1 \over 1 + \theta_T}
\begin{blockarray}{ccc}
R' & T  & D' 
\\ 
\begin{block}{(ccc)}
\vc{\lambda}_T\vc{P}_{T,R} &
\vc{\lambda}_T   &  
\vc{\lambda}_T\vc{P}_{T,D}   
\\  
 \end{block}
\end{blockarray}
\,.
\label{defn:lambda_{T^{(e)}}}
\end{align}
\item The Markov transition matrix $\vc{M}$ satisfies
\begin{align}
\lim_{n\to\infty}{1 \over n} \sum_{\nu=0}^{n-1} \vc{M}^{\nu}
=
\begin{blockarray}{cccccc}
		&	
R_1    	& 
\cdots 	& 
R_K 	& 
\widehat{T}	& 
D   
\\ 
\begin{block}{c(ccc|c|c)}
R_1          				&
\vc{e}_{R_1}\vc{\lambda}_{R_1} 	&
			 				&
\mbox{\large $\vc{O}$}		&
\vc{O}		 &
\vc{O}		 
\\	
\vdots       &
			 &
\ddots		 &
			 &
\vdots		 &
\vdots		 		 
\\
R_K          &
\mbox{\large $\vc{O}$}		 &
			 &
\vc{e}_{R_K}\vc{\lambda}_{R_K} &
\vc{O}		 &
\vc{O}		 
\\	
\BAhhline{~-----}
\widehat{T}              &
\vc{O}		 &
\cdots		 &
\vc{O}		 &
\vc{e}_{\widehat{T}}\vc{\lambda}_{\widehat{T}}		 &
\vc{O}		 
\\
\BAhhline{~-----}		
D			&
\vc{O}		&
\cdots		 &
\vc{O}		&
\vc{O}		&
\vc{I}_D		 
\\			 
\end{block}
\end{blockarray}
~.
\label{eqn:cesaro_limit-01}
\end{align}
\item The PureRank vector $\vc{\pi}=(\pi_j)_{j \in V}$ is given by
\begin{align}
\vc{\pi}
= 
\lim_{n\to\infty}{1 \over n}\sum_{\nu=0}^{n-1}\vc{\varpi} \vc{M}^{\nu} \vc{F},
\label{eqn:vector_pi}
\end{align}
where
\begin{align}
\vc{\varpi} 
&={1 \over N}
\begin{blockarray}{ccccc}
R & R' & T  & D' & D \\ 
\begin{block}{(ccccc)}
\vc{e}_R^{\top} & 
\vc{0}^{\top} &
\vc{e}_T^{\top} & 
\vc{0}^{\top} &
\vc{e}_D^{\top}
\\  
\end{block}
\end{blockarray}
,
\label{defn:eta} 
\\
\vc{F}
&=
\begin{blockarray}{cccc}
 & R & T & D \\ 
\begin{block}{c(c|c|c)}
R        & 
\vc{I}_R &
\vc{O}   &    
\vc{O}
\\  
R'        & 
\vc{I}_R   &  
\vc{O}   &  
\vc{O}
\\  
\BAhhline{~---}
T        & 
\vc{O}   & 
\vc{I}_T   & 
\vc{O}   
\\  
\BAhhline{~---}
D'        & 
\vc{O}   &  
\vc{O}   &  
\vc{I}_D   
\\ 
D        & 
\vc{O}   &  
\vc{O}   &  
\vc{I}_D 
\\
 \end{block}
\end{blockarray}
~.
\label{defn:E}
\end{align}
\end{enumerate}
\end{lem}

\begin{proof}
See Appendix~\ref{proof:lem_random_surfer-01}.
\end{proof}

\medskip

Using Lemma~\ref{lem:random_surfer-01}, we state the following theorem, which specifies the Markov chain on the extended node set $\widehat{V}$ and provides the foundation for interpreting the chain as the random-surfer model of PureRank.
\begin{thm}\label{thm:random_surfer}
Let $\{\widehat{X}_n:n\in\bbZ_+\}$ be a Markov chain on the state space $\widehat{V}$ with initial distribution $\vc{\varpi}$ and Markov transition matrix $\vc{M}$. The following statements hold:
\begin{enumerate}
\item Let $\one(A)$ denote the indicator of an event $A$; that is, $\one(A)=1$ if $A$ occurs and $\one(A)=0$ otherwise. We then have
\begin{align}
\pi_j
&=
\left\{
\begin{alignedat}{1}
\lim_{n\to \infty}
{1 \over n}
\sum_{i \in \widehat{V}} [\vc{\varpi}]_i
\EE\left[ \left. \sum_{\nu=0}^{n-1} \one(\widehat{X}_{\nu} = j) \right| 
\widehat{X}_0 = i\right], & ~~~ j \in T,
\\
\lim_{n\to \infty}
{1 \over n}
\sum_{i \in \widehat{V}} [\vc{\varpi}]_i
\EE\left[\left. 
\sum_{\nu=0}^{n-1} \one(\widehat{X}_{\nu} \in \{j,j'\}) \right| 
\widehat{X}_0 = i \right], & ~~~ j \in R \sqcup D.
\end{alignedat}
\right.
\label{eqn:pi_j}
\end{align}
If $T=\varnothing$, then the first case in (\ref{eqn:pi_j}) does not arise, and in the second case $\{j,j'\}$ is to be read as $\{j\}$.
\item Suppose that $T\neq\varnothing$. Let $\{\tau_T(m): m \in \bbZ_+\}$ be the sequence of random variables defined by
\begin{align*}
\tau_T(0) &= \inf\{n \in \bbZ_+: \widehat{X}_n \in T\},
\\
\tau_T(m) &= \inf\{n > \tau_T(m-1): \widehat{X}_n \in T, \widehat{X}_{n-1} \in R' \sqcup D'\},
\qquad m \in \bbN.
\end{align*}
For each $m \in \bbZ_+$, the quantity $\tau_T(m+1) - 1$ is the first hitting time of $R' \sqcup D'$ after time $\tau_T(m)$. Moreover,
\begin{align*}
\{\tau_T(m), \tau_T(m)+1, \dots, \tau_T(m+1) - 2\}
\end{align*}
is the $m$th sojourn period in Class $T$, and its length is
\begin{align*}
C_T(m):=\tau_T(m+1)-\tau_T(m)-1.
\end{align*}
Under these definitions, the sequence $\{C_T(m): m \in \bbZ_+\}$ is independent and identically distributed (i.i.d.) given $\widehat{X}_0 \in T$, and
\begin{subequations}\label{eqns_theta_T}
\begin{align}
\EE[C_T(m)\mid \widehat{X}_0 \in T]
&= {1 \over \theta_T}
\label{defn:theta_T-a}
\\
&= {1 \over |T| \beta_T^*}
\label{defn:theta_T-b}
\\
&= \vc{\mu}_T(\vc{I} - \vc{P}_T)^{-1}\vc{e}_T.
\label{defn:theta_T-c}
\end{align}
\end{subequations}
\end{enumerate}
\end{thm}

\begin{proof}
See Appendix~\ref{proof:thm:random_surfer}.
\end{proof}

\medskip

Theorem~\ref{thm:random_surfer} yields a random-surfer interpretation of PureRank through the Markov chain $\{\widehat{X}_n\}$. We consider independent surfers whose initial distribution is $\vc{\varpi}$ in (\ref{defn:eta}). Under this initial distribution, a surfer starts in $R$, $T$, or $D$ with probabilities $|R|/N$, $|T|/N$, and $|D|/N$, respectively, and the initial node is sampled uniformly from the selected class. In particular, for each $k \in \bbK$, a surfer starts in $R_k$ with probability $|R_k|/N$. If any of Classes $R$, $T$, and $D$ is empty, the corresponding case below is omitted. Each surfer evolves on the extended state space $\widehat{V}$ according to the following rules.
\begin{enumerate}
\item If the initial node lies in $R_k$ for some $k \in \bbK$, then the surfer evolves according to the transition matrix $\vc{P}_{R_k}$.
\item If the initial node lies in $T$, then the surfer evolves on $\widehat{T}=R' \sqcup T \sqcup D'$ according to the transition matrix $\vc{M}_{\widehat{T}}$ in (\ref{defn:matrix_M}).
\item If the initial node lies in $D$, then the surfer remains at that node indefinitely.
\end{enumerate}

\begin{figure}[tbh]
\centering
\includegraphics[bb=0 0 706 376, scale=0.4]{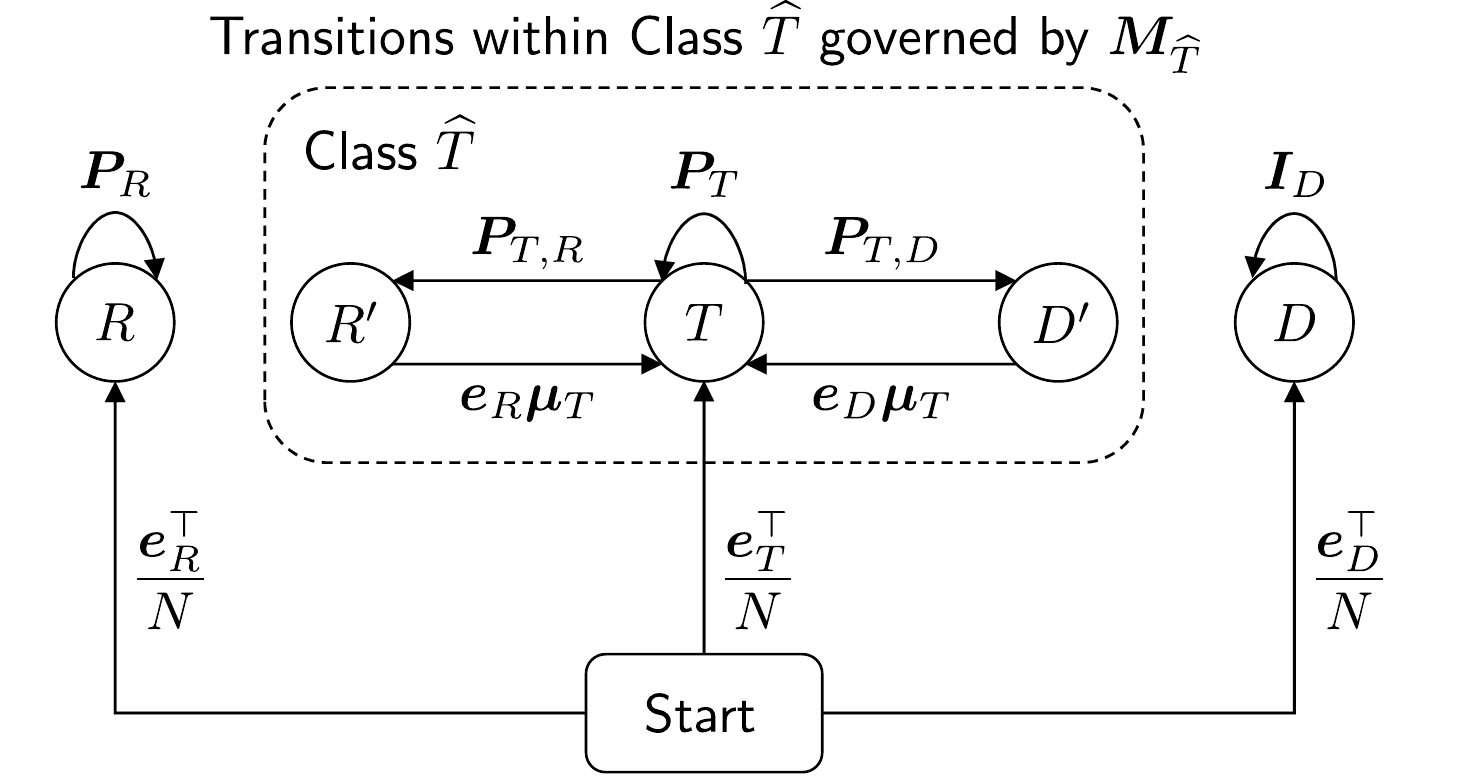}
\caption{Transition dynamics of the random-surfer model $\{\widehat{X}_n\}$.}
\label{fig:random_surfer}
\end{figure}

The random-surfer model explains how the RDI-L2G construction aggregates the local importance vectors into the PureRank vector. For each nonempty recurrent class $R_k$, a surfer starting in $R_k$ follows $\vc{P}_{R_k}$, and the long-run visit frequencies yield $\vc{\lambda}_{R_k}$. If Class $D$ is nonempty, a surfer starting in $D$ remains at the initial node, and the long-run visit frequencies yield $\vc{\lambda}_D=\vc{\mu}_D$. If Class $T$ is nonempty, a surfer starting in $T$ evolves on $\widehat{T}$ according to $\vc{M}_{\widehat{T}}$, and the normalized $T$-part of the long-run visit frequencies yields $\vc{\lambda}_T$. Each exit from $T$ to $R'$ or $D'$ is followed by a one-step return to $T$ and a uniform restart on $T$. Under the RDI principle, the resulting $R'$-part of the long-run visit frequencies on $\widehat{T}$ is aggregated with the local importance vectors $\{\vc{\lambda}_{R_k}\}_{k \in \bbK}$ for the recurrent classes, and the resulting $D'$-part is aggregated with $\vc{\lambda}_D$.

The random-surfer model also clarifies the probabilistic meanings of the PureRank vector $\vc{\pi}$ and the parameter $\theta_T$. Each surfer's trajectory is a sample path of $\{\widehat{X}_n\}$; see Fig.~\ref{fig:random_surfer}. For independent surfers with initial distribution $\vc{\varpi}$, the long-run visit frequency of node $j$, counting visits to its copy when a copy exists, converges to $\pi_j$. If Class $T$ is nonempty, then $\theta_T$ in (\ref{eqns_theta_T}) is the reciprocal of the expected sojourn time in Class $T$ before the surfer reaches $R'$ or $D'$. Equivalently, $\theta_T$ is the long-run frequency of exits from Class $T$.

\section{Experiments: comparison with PageRank}\label{sec:numerical_experiments}
This section investigates PureRank and PageRank on three large-scale real-world networks of different types: the \texttt{twitter\_combined} network~\cite{snap_twitter}, the \texttt{cit-HepPh} network~\cite{snap_cit_hepph}, and the \texttt{ca-AstroPh} network~\cite{snap_ca_astroph}, drawn from the Stanford Large Network Dataset Collection~\cite{snap_datasets}. We first summarize, for each network, the basic network statistics, the numbers of nodes in Classes $R$, $T$, and $D$, together with the size of the largest class in the node classification $\calC$. We next compare the computational cost of PureRank and PageRank, emphasizing that the dominant factors differ between the two importance measures: for PageRank, the runtime depends strongly on the damping factor $d$, whereas for PureRank it is typically determined by the size of the largest class in $\calC \setminus \{D\}$, that is, among the recurrent classes and the transient one. We finally vary the damping factor $d$ to study how the similarity between PureRank and PageRank, in terms of both rankings and scores, depends on whether the largest class in $\calC \setminus \{D\}$ is recurrent or transient.

We used the power method to compute the PageRank vector and the local importance vectors for PureRank in all numerical experiments. We terminated the iteration when the $L_1$ norm of the difference between successive iterates fell below the error tolerance $10^{-10}$.

We conducted all numerical experiments on a machine running Windows 11 Pro (version 25H2). The machine was equipped with an Intel Core i7-1270P CPU (2.20 GHz) and 32 GB of RAM. We implemented the algorithms in Python 3.12.7 using NumPy 1.26.4, SciPy 1.13.1, and NetworkX 3.3. NumPy was linked against Intel MKL (BLAS version 2023.1).

\subsection{Class structure of the target networks}

The node classifications of the target networks differ substantially. The numbers of nodes in Classes $R$, $T$, and $D$, the size of the largest class in $\calC$, and the distribution of recurrent class sizes in each network can be summarized as follows (see Tables~\ref{tab:network_stats_all} and \ref{tab:recurrent_class_distribution_all}):
\begin{itemize}
\item \texttt{twitter\_combined}: A network dominated by Class $T$ (approximately 85\% of all nodes), with a small recurrent population (approximately 0.77\% of all nodes) distributed across many small recurrent classes, most of which have size at most three.
\item \texttt{cit-HepPh}: A network heavily dominated by Class $T$ (approximately 93\% of all nodes), with a negligible recurrent population (approximately 0.02\% of all nodes) consisting only of singletons and pairs.
\item \texttt{ca-AstroPh}: A fully recurrent network (which has no transient or dangling nodes), whose recurrent structure is dominated by a single large recurrent class containing approximately 95\% of all nodes.
\end{itemize}

In Sections~\ref{subsec:comparison_cost} and \ref{subsec:similarity_ranking_scoring}, we examine how differences in node classification affect the computational cost of PureRank and PageRank, especially the power iterations used to compute the PageRank vector and the local importance vectors in PureRank, as well as the similarity between PureRank and PageRank in rankings and scores.

\subsection{Comparison of computational cost}\label{subsec:comparison_cost}

This subsection compares the computational cost of PageRank and PureRank by using the iteration counts and median wall-clock times reported in Panel~B of Table~\ref{tab:comparison_page_pure_all}. We first explain why the PageRank runtime increases with the damping factor $d$. We next compare the runtimes of PureRank with the corresponding PageRank runtimes for different values of $d$. We finally analyze the PureRank runtimes by examining the power iteration for computing the local importance vector of the largest class in $\calC \setminus \{D\}$ and explain the cross-network differences in these runtimes by using the largest-class statistics in Table~\ref{tab:network_stats_all}.

The computational cost of the PageRank vector $\vc{\gamma}^{(d)}$ increases with the damping factor $d$ because the power iteration converges more slowly. The PageRank vector $\vc{\gamma}^{(d)}$ is computed by generating the sequence $\{\vc{\gamma}_n^{(d)}\}_{n\in\bbZ_+}$ from the initial vector $\vc{\gamma}_0^{(d)}=\vc{\mu}_V$:
\begin{align}
\vc{\gamma}_{n+1}^{(d)}
&=\vc{\gamma}_n^{(d)}d\vc{Q}_V+(1-d)\vc{\mu}_V,\qquad n=0,1,\dots.
\label{recursion:gamma_n^{(d)}}
\end{align}
This iteration follows from (\ref{defn:PageRank-01}). In general, the iterative computation dominates the preprocessing cost of constructing $\vc{Q}_V$. Using (\ref{defn:PageRank-01}), (\ref{recursion:gamma_n^{(d)}}), and $\vc{Q}_V\vc{e}_V=\vc{e}_V$, we bound the error $\vc{\gamma}^{(d)}-\vc{\gamma}_{n+1}^{(d)}$ in the $L_1$ norm as follows:
\begin{align*}
\|\vc{\gamma}^{(d)}-\vc{\gamma}_{n+1}^{(d)}\|_1
&=d\|(\vc{\gamma}^{(d)}-\vc{\gamma}_n^{(d)})\vc{Q}_V\|_1
\\
&=d\sum_{j\in V}\left|
\sum_{i\in V}\left\{(\vc{\gamma}^{(d)})_i-(\vc{\gamma}_n^{(d)})_i\right\}
(\vc{Q}_V)_{i,j}
\right|
\\
&\le d\sum_{j\in V}\sum_{i\in V}
\left|(\vc{\gamma}^{(d)})_i-(\vc{\gamma}_n^{(d)})_i\right|
(\vc{Q}_V)_{i,j}
\\
&=d\sum_{i\in V}
\left|(\vc{\gamma}^{(d)})_i-(\vc{\gamma}_n^{(d)})_i\right|
\sum_{j\in V}(\vc{Q}_V)_{i,j}
\\
&=d\sum_{i\in V}
\left|(\vc{\gamma}^{(d)})_i-(\vc{\gamma}_n^{(d)})_i\right|
=d\|\vc{\gamma}^{(d)}-\vc{\gamma}_n^{(d)}\|_1.
\end{align*}
Thus, as is well known (see, e.g., \cite{Lang06}), a larger value of $d$ yields slower convergence of $\{\vc{\gamma}_n^{(d)}\}_{n\in\bbZ_+}$ to $\vc{\gamma}^{(d)}$. Panel~B of Table~\ref{tab:comparison_page_pure_all} is consistent with this contraction bound: in all three networks, larger $d$ leads to more iterations and a longer total runtime.

PageRank has a shorter runtime than PureRank except when $d$ is very close to one. Panel~B of Table~\ref{tab:comparison_page_pure_all} shows that this exception occurs only at $d=0.999$ in all three networks. Moreover, the runtime gap between PureRank and PageRank arises because PureRank does not use teleportation, whereas teleportation accelerates the iterative computation of the PageRank vector. Indeed, it is standard (see, e.g., \cite{Lang06}) that more frequent teleportation, that is, smaller $d$, accelerates the iterative computation of the PageRank vector.

The computation of the local importance vector for the largest class in $\calC \setminus \{D\}$ accounts for a substantial share of the PureRank runtime in the \texttt{twitter\_combined} and \texttt{ca-AstroPh} networks. As noted at the end of Section~\ref{sec:PureRank} (after Algorithm~\ref{algo:PureRank}), the computation of the local importance vector for the largest class in $\calC \setminus \{D\}$ can be a major component of the total PureRank runtime, particularly when the largest class in $\calC \setminus \{D\}$ contains most of the nodes, as in the three networks considered here. Combining Panel~B of Table~\ref{tab:comparison_page_pure_all} with Table~\ref{tab:network_stats_all}, we find that the median time for computing the local importance vector of the largest class is 50.8\% of the median total runtime of PureRank in the \texttt{twitter\_combined} network (9.506244/18.713799), 2.2\% in the \texttt{cit-HepPh} network (0.042390/1.944775), and 40.2\% in the \texttt{ca-AstroPh} network (0.836499/2.079275). These percentages show that the computation of the local importance vector for the largest class can be a major runtime component in the \texttt{twitter\_combined} and \texttt{ca-AstroPh} networks but not in the \texttt{cit-HepPh} network, even though the largest class contains most of the nodes in all three networks.

For the transient-dominated \texttt{twitter\_combined} and \texttt{cit-HepPh} networks, the iteration count for computing the local importance vector of the largest class is consistent with the leakage rate from the largest class. In both networks, the largest class is $T$, and Theorem~\ref{thm:LPR-class_T} shows that $\vc{\lambda}_T$ is obtained from the recursion (\ref{lim:lambda_T-02b}), that is, by power iteration on $\vc{Q}_T$ in (\ref{defn:Q_T}). The second term on the right-hand side of (\ref{defn:Q_T}) redistributes the probability mass that leaves Class~$T$, and its stationary mean is the leakage rate $\theta_T$. Panel~B of Table~\ref{tab:network_stats_all} reports 2{,}816 iterations for the \texttt{twitter\_combined} network but only 45 for the \texttt{cit-HepPh} network. This contrast is consistent with the larger leakage rate in the \texttt{cit-HepPh} network ($\theta_T\approx 0.294$) than in the \texttt{twitter\_combined} network ($\theta_T\approx 0.0357$), and it suggests that a larger leakage rate is associated with faster convergence of the power iteration for $\vc{\lambda}_T$.

Finally, we compare the largest-class computation runtimes of the \texttt{ca-AstroPh} and \break \texttt{twitter\_combined} networks, because in the \texttt{cit-HepPh} network the largest-class computation accounts for only 2.2\% of the total PureRank runtime. In the \texttt{ca-AstroPh} network, the largest class is recurrent rather than transient. Panel~B of Table~\ref{tab:network_stats_all} shows that this class has 17{,}903 nodes, compared with 69{,}473 for the largest class in the \texttt{twitter\_combined} network, and that its iteration count is also smaller (1{,}877 versus 2{,}816). The median time for computing the local importance vector of the largest class is therefore 0.836499 in the \texttt{ca-AstroPh} network, which is about 8.8\% of the corresponding value 9.506244 in the \texttt{twitter\_combined} network. Taken together, these comparisons indicate that the leakage mechanism in transient classes and the size and iteration count of the largest class both contribute to the cross-network differences in the PureRank runtimes.

\subsection{Similarity in ranking and scoring}\label{subsec:similarity_ranking_scoring}

This subsection examines how the similarity between PureRank and PageRank varies with the damping factor $d$, with separate attention to ranking similarity and score similarity. We first introduce the three similarity measures used below: Top-100 Overlap, Kendall's $\tau_b$, and the Pearson correlation coefficient (PCC). We next summarize the behavior of these measures in the two transient-dominated networks, the \texttt{twitter\_combined} and \texttt{cit-HepPh} networks. For these networks, we explain why Kendall's $\tau_b$ is high when $d$ is close to one by using Theorem~\ref{thm:classwise_PR_limit}. We also consider why Top-100 Overlap and PCC can behave differently from Kendall's $\tau_b$ as $d$ varies in the transient-dominated networks. We finally use the fully recurrent \texttt{ca-AstroPh} network as a control case.

We specify how Top-100 Overlap, Kendall's $\tau_b$, and PCC are computed. For Top-100 Overlap, we sort nodes in descending order of score, break ties by ascending node ID, and compute the percentage of common nodes in the resulting top-100 sets. For Kendall's $\tau_b$, we convert each score vector into dense ranks, assign the same rank to tied scores, and apply SciPy's \texttt{kendalltau} function to the resulting rank vectors. We compute PCC directly from the score vectors.

We summarize the empirical behavior of the three similarity measures in the transient-dominated \texttt{twitter\_combined} and \texttt{cit-HepPh} networks, where Class $T$ contains most nodes; see Table~\ref{tab:network_stats_all}. Panel~A of Table~\ref{tab:comparison_page_pure_all} shows that Kendall's $\tau_b$ between the PureRank and PageRank rankings increases with $d$ in both networks, whereas Top-100 Overlap and PCC are nonmonotone and can decline when $d$ is close to one. At $d=0.999$, Kendall's $\tau_b$ reaches its largest value, namely, $0.709$ in the \texttt{twitter\_combined} network and $0.951$ in the \texttt{cit-HepPh} network, whereas the all-node PCC reaches its smallest value, namely, $0.174$ and $0.093$, respectively. At the same value $d=0.999$, Top-100 Overlap also reaches its smallest value, namely, $3\%$, in the \texttt{twitter\_combined} network, whereas in the \texttt{cit-HepPh} network it decreases to $93\%$ from its peak value, $97\%$, at $d=0.95$. These results show that, in the transient-dominated networks, stronger agreement in global ordering does not imply closer agreement in scores or in the identities of the highest-ranked nodes.

We first explain why the all-node Kendall's $\tau_b$ between PureRank and PageRank is high when $d$ is close to one in the transient-dominated \texttt{twitter\_combined} and \texttt{cit-HepPh} networks. Theorem~\ref{thm:classwise_PR_limit} together with Definition~\ref{defn:PureRank} implies that
\begin{align}
\lim_{d \uparrow 1}\widetilde{\vc{\gamma}}_T^{(d)} 
&= \vc{\lambda}_T = {\vc{\pi}_T \over \vc{\pi}_T\vc{e}_T},
\nonumber
\\
\lim_{d \uparrow 1}\widetilde{\vc{\gamma}}_{R_k}^{(d)} 
&= \vc{\lambda}_{R_k},
\label{eq:limit_gamma_Rk}
\\
\lim_{d \uparrow 1}\widetilde{\vc{\gamma}}_D^{(d)} 
&=
{
|D| \vc{\mu}_D + \dm{|T| \over \theta_T} \vc{\lambda}_T\vc{P}_{T,D} 
\over 
|D| + \dm{|T| \over \theta_T} \vc{\lambda}_T\vc{P}_{T,D}\vc{e}_D
}.
\label{eq:limit_gamma_D}
\end{align}
The limits in (\ref{eq:limit_gamma_Rk}) and (\ref{eq:limit_gamma_D}) do not, in general, coincide with the corresponding normalized PureRank vectors. However, when the contribution of $|D|\vc{\mu}_D$ is small relative to the leakage term from Class $T$, (\ref{formula_pi_D}) and (\ref{eq:limit_gamma_D}) yield
\begin{align*}
\widetilde{\vc{\gamma}}_D^{(d)}  
\approx 
{
\vc{\lambda}_T\vc{P}_{T,D} 
\over 
\vc{\lambda}_T\vc{P}_{T,D}\vc{e}_D
}
\approx {\vc{\pi}_D \over \vc{\pi}_D\vc{e}_D},
\qquad d\uparrow 1.
\end{align*}
Thus, when $d$ is close to one, the PageRank and PureRank scores are approximately linearly related within each of Classes $T$ and $D$. Table~\ref{tab:classwise_tau_pcc} supports this interpretation: at $d=0.999$, the within-class Kendall's $\tau_b$ and PCC in Classes $T$ and $D$ are close to one in both networks. Because Table~\ref{tab:network_stats_all} shows that Classes $T$ and $D$ together contain almost all nodes in both networks, these classes dominate the pairwise comparisons underlying the all-node Kendall's $\tau_b$. This observation helps explain why the all-node Kendall's $\tau_b$ is high when $d$ is close to one.

We next consider why Top-100 Overlap can behave differently from Kendall's $\tau_b$ in the \texttt{twitter\_combined} and \texttt{cit-HepPh} networks by examining the class composition of the top-ranked nodes; see Table~\ref{tab:comparison_page_pure_top100}. By definition, teleportation from every node in PageRank becomes less frequent as the damping factor $d$ increases. As a result, PageRank increasingly favors nodes in Class $R$ over nodes in Class $T$, because the leakage of importance from Class~$T$ is compensated less often. Consistent with this shift, the relative advantage of Class $R$ over Class $T$ increases markedly as $d$ increases. In the \texttt{twitter\_combined} network at $d=0.999$, PageRank places $98\%$ of the top-100 nodes in Class $R$, whereas PureRank places $95\%$ of the top-100 nodes in Class $T$. The two top-100 sets therefore differ sharply, and Top-100 Overlap is small. Kendall's $\tau_b$, however, is based on the ordering over all nodes rather than only the top 100 and can therefore remain relatively high even when the highest-ranked nodes differ substantially. By contrast, in the \texttt{cit-HepPh} network, Class $R$ contains only seven nodes; see Table~\ref{tab:network_stats_all}. The top-100 set therefore remains concentrated in Classes $T$ and $D$, which is consistent with a high Top-100 Overlap.

We also consider why the all-node PCC can be low in the \texttt{twitter\_combined} and \texttt{cit-HepPh} networks when $d$ is close to one even though the all-node Kendall's $\tau_b$ is high. At $d=0.999$, the within-class PCCs in Classes $T$ and $D$, which together contain almost all nodes, equal $1.000$ in both networks; see Table~\ref{tab:classwise_tau_pcc}. Thus, the PageRank and PureRank scores are approximately linearly related within Classes $T$ and $D$. However, as $d$ increases, PageRank and PureRank place different relative emphasis on Classes $R$, $T$, and $D$. At $d=0.999$, the classwise ratios of mean PageRank to mean PureRank, in the order of Classes $R$, $T$, and $D$, are $(72.4,0.210,0.040)$ in the \texttt{twitter\_combined} network and $(531,0.777,0.632)$ in the \texttt{cit-HepPh} network; see Table~\ref{tab:comparison_page_pure_top100}. The classwise linear relation may therefore vary across classes, and the PageRank--PureRank scatter plot over all nodes need not be close to a single common line. This difference may explain why the all-node PCC can be very small even when the all-node Kendall's $\tau_b$ is high.

We close with the fully recurrent \texttt{ca-AstroPh} network as a control case. Table~\ref{tab:network_stats_all} shows that the \texttt{ca-AstroPh} network contains no transient or dangling nodes, and Panel~A of Table~\ref{tab:comparison_page_pure_all} shows that all three similarity measures increase with $d$ and indicate near coincidence at $d=0.999$. Remark~\ref{rem:PureRank_in_fully_recurrent_NW} clarifies this behavior: under uniform teleportation, PageRank assigns each closed class $R_k$ the fixed mass $|R_k|/N$ for all $d$. Combined with Theorem~\ref{thm:classwise_PR_limit}(i), this identity yields $\vc{\gamma}^{(d)}\to\vc{\pi}$ as $d\uparrow1$. The monotonic increase of similarity is therefore consistent with the convergence of PageRank to PureRank in the fully recurrent \texttt{ca-AstroPh} network.

\section{Extension: splitting-network construction for multi-attribute networks}\label{sec:multi_attribute_networks}

This section extends PureRank to directed networks with multiple link attributes through the {\it splitting-network construction}. This construction converts such a network into a single-attribute network with nonnegative link weights, called the {\it splitting network}, by splitting each node into copies indexed by link attributes so that each copy receives the in-links of the corresponding attribute. Applying Algorithm~\ref{algo:PureRank} to the splitting network yields, for each original node, an importance vector whose entries are indexed by the link attributes. We refer to this vector as the \textit{multi-attribute PureRank vector}.

We formally represent a multi-attribute network as $G:=(V,(\vc{W}^{(a)})_{a\in\mathcal{A}})$, where $\mathcal{A}=\{1,2,\dots,m\}$ denotes the index set of attributes and $m:=|\mathcal{A}|$. For each $a\in\mathcal{A}$, let $(i,j;a)$ denote a directed link from node $i$ to node $j$ with attribute $a$, and let $\vc{W}^{(a)}:=(w_{i,j}^{(a)})_{i,j\in V}$ denote the $|V|\times |V|$ matrix with
\begin{align*}
w_{i,j}^{(a)}=
\left\{
\begin{array}{ll}
\mbox{the weight of the link $(i,j;a)$}, & \mbox{if such a link exists}, \\
0, & \mbox{otherwise}.
\end{array}
\right.
\end{align*}
We assume without loss of generality that $\vc{W}^{(a)}$ is nonnegative. Indeed, if a link with attribute $a\in\mathcal{A}$ has a negative weight, then its sign can be reversed and the link can be reassigned the new attribute $(a,-)$. Likewise, if a link with attribute $a\in\mathcal{A}$ has a positive weight, then it can be reassigned the new attribute $(a,+)$. Thus, the index set $\mathcal{A}$ is extended to include $(a,+)$ and $(a,-)$ for each original attribute $a\in\{1,2,\dots,m\}$.

We now define the splitting-network construction, which converts the multi-attribute network $G:=(V,(\vc{W}^{(a)})_{a\in\mathcal{A}})$ into a single-attribute network with nonnegative link weights. Each node $i\in V$ is split into $m \in \bbN$ copies, $\{i^{(1)}, i^{(2)}, \dots, i^{(m)}\}$, indexed by the attributes in $\mathcal{A}$. The resulting network is denoted by $G^{\mathcal{A}}:=(V^{\mathcal{A}},\vc{W}^{\mathcal{A}})$, where $V^{\mathcal{A}}$ and $\vc{W}^{\mathcal{A}}$ are the node set and the weight matrix, respectively:
\begin{align*}
V^{\mathcal{A}}&=\{i^{(a)}:i\in V,a\in\mathcal{A}\},
\\
\vc{W}^{\mathcal{A}}&=\begin{blockarray}{ccccc}
 & V^{(1)} & V^{(2)} & \cdots & V^{(m)} \\
\begin{block}{c(cccc)@{~~}}
V^{(1)}    &
\vc{W}^{(1)}   &
\vc{W}^{(2)}   &
\cdots         &
\vc{W}^{(m)}
\\
V^{(2)}    &
\vc{W}^{(1)}   &
\vc{W}^{(2)}   &
\cdots         &
\vc{W}^{(m)}
\\
\vdots & \vdots & \vdots & \ddots & \vdots \\
V^{(m)}    &
\vc{W}^{(1)}   &
\vc{W}^{(2)}   &
\cdots         &
\vc{W}^{(m)}
\\
 \end{block}
\end{blockarray}
~,
\end{align*}
where $V^{(a)}=\{i^{(a)}:i\in V\}$ for each $a\in\mathcal{A}$ and $V^{\mathcal{A}}=\bigsqcup_{a\in\mathcal{A}}V^{(a)}$. Note that the weights of $G^{\mathcal{A}}$ satisfy $w_{i^{(a)},j^{(b)}}^{\mathcal{A}}=w_{i,j}^{(b)}$ for all $i,j\in V$ and $a,b\in\mathcal{A}$. In this construction, the attribute of each link is encoded by the copy of its destination node; hence, the resulting network $G^{\mathcal{A}}=(V^{\mathcal{A}},\vc{W}^{\mathcal{A}})$ can be regarded as a single-attribute network.

We define the multi-attribute PureRank vector for each original node in the multi-attribute network $G$. Since the splitting network $G^{\mathcal{A}}$ is a directed network with nonnegative link weights, the PureRank vector on $G^{\mathcal{A}}$ is well-defined according to Definition~\ref{defn:PureRank} and is computed by Algorithm~\ref{algo:PureRank}. For each $j\in V$ and $a\in\mathcal{A}$, the PureRank score $\pi_{j^{(a)}}$ of the copy $j^{(a)}$ represents the importance score of node $j$ associated with attribute $a$. Thus, the scores $\pi_{j^{(1)}},\pi_{j^{(2)}},\dots,\pi_{j^{(m)}}$ form the multi-attribute PureRank vector of node $j$:
\begin{align*}
\vc{\pi}_j^{\mathcal{A}}:=(\pi_{j^{(1)}},\pi_{j^{(2)}},\dots,\pi_{j^{(m)}}),
\end{align*}
which provides an attribute-wise importance profile for node $j$.

\begin{rem}
PureRank on multi-attribute networks extends to networks with both link and node attributes when each node has a single attribute. Let $\mathcal{N}$ denote the set of node attributes, and let $\mathscr{A}:V\to\mathcal{N}$ be a node-attribute map. For each link $(i,j)$ with attribute $a\in\mathcal{A}$, we use the ordered triple $(a,\mathscr{A}(i),\mathscr{A}(j))\in\mathcal{A}\times\mathcal{N}\times\mathcal{N}$ as the composite attribute of the link. Node attributes are thereby fully encoded into link attributes. Therefore, no further modification is needed to apply the splitting-network construction and Algorithm~\ref{algo:PureRank} to the resulting network. Under this encoding, for each original node $j\in V$, the multi-attribute PureRank vector is indexed by the pairs $(a,u)\in\mathcal{A}\times\mathcal{N}$, and its $(a,u)$th entry is the PureRank score corresponding to the composite attribute $(a,u,\mathscr{A}(j))$.
\end{rem}

We discuss how to aggregate the multi-attribute PureRank vector into a single summary score. In some applications, a single summary score is preferred, for example, a score obtained by a weighted sum. As a simple illustration, we consider a two-signed network in which the ``positive'' ($+$) and ``negative'' ($-$) attributes are indexed by $1$ and $2$, respectively. For node $j$, the multi-attribute PureRank vector is then $(\pi_{j^{(1)}},\pi_{j^{(2)}})$. A natural summary score is the net score, defined by the difference
\begin{align*}
\pi_j^{\pm}:=\pi_{j^{(1)}}-\pi_{j^{(2)}}.
\end{align*}
If $\pi_j^{\pm}>0$, the positive evaluation of node $j$ dominates the negative one. If $\pi_j^{\pm}<0$, the negative evaluation dominates the positive one. If $\pi_j^{\pm}=0$, the two evaluations are balanced.

This paper focuses exclusively on the theoretical development of PureRank for multi-attribute networks and does not report numerical experiments. In multi-attribute networks, node importance depends not only on overall connectivity, including link weights, but also on the distribution of link attributes, which makes the interpretation and systematic evaluation of ranking results more difficult. A meaningful empirical assessment would therefore need to examine how network structure and link attributes jointly affect the ranking results and to compare those results with standard baseline methods. We leave a comprehensive empirical study of PureRank on multi-attribute networks for future work.

\section{Discussion: implications, limitations, and future directions}\label{sec:concluding_remarks}

This section discusses the implications of PureRank as a parameter-free, RDI-based importance (centrality) measure for directed networks and outlines its limitations and future directions. The following subsections address these topics separately.

\subsection{Discussion}

PureRank is a new parameter-free recursive importance measure for arbitrary directed networks. PureRank answers the question of whether an RDI-based importance measure can be constructed without user-specified parameters. The name ``PureRank'' emphasizes that PureRank yields a unique importance score vector that faithfully reflects the intrinsic network structure without empirical or heuristic tuning. PureRank can therefore serve as a neutral reference for comparison with application-specific, parameter-dependent importance measures. The construction of PureRank consists of three steps. First, SCC decomposition together with the out-degree condition determines the node classification $\calC$ in (\ref{eq:defn_C}) into recurrent, transient, and dangling classes. Second, for each class, the local importance vector is obtained by solving the Katz equation on the class-restricted subnetwork, where the Katz parameters are chosen according to the RDI principle. Third, the resulting local importance vectors are aggregated through the {\it RDI-based local-to-global (RDI-L2G) construction} to form the {\it PureRank vector}. This modular design supports both parallel and incremental computation and thereby makes PureRank scalable to large or dynamically evolving networks. The concept of splitting networks further extends PureRank naturally to multi-attribute networks.

Our analysis has focused on the theoretical formulation, the computational procedure, and the random-surfer interpretation of PureRank. As a preliminary empirical assessment, we examined three real-world unsigned networks and compared PureRank with PageRank. The numerical results show that the relationship between PureRank and PageRank depends strongly on the node classification $\calC$ of each network. In the two transient-dominated networks, \texttt{twitter\_combined} and \texttt{cit-HepPh}, the three similarity measures between PureRank and PageRank---Kendall's $\tau_b$, Top-100 Overlap, and PCC---do not follow a common trend as the damping factor $d$ increases. Kendall's $\tau_b$ increases monotonically with $d$, whereas Top-100 Overlap and PCC do not exhibit the same monotonic behavior; see Panel~A of Table~\ref{tab:comparison_page_pure_all}. By contrast, in the fully recurrent \texttt{ca-AstroPh} network, all three similarity measures increase steadily with $d$, which suggests that the two rankings nearly coincide as $d$ approaches one; see also Remark~\ref{rem:PureRank_in_fully_recurrent_NW}. These findings indicate that the relationship between PureRank and PageRank is not universal but depends on the node classification $\calC$ of the underlying network.

\subsection{Limitations}

PureRank has several practical limitations. First, PureRank scores depend fundamentally on the node classification $\calC$ into recurrent, transient, and dangling classes. Accordingly, changes in the SCC decomposition or in node out-degrees can modify $\calC$ and hence alter the resulting scores. This dependence arises because $\calC$ determines both the class-restricted subnetworks on which the local importance vectors are computed and the {\it RDI-based local-to-global (RDI-L2G) construction} that aggregates them into the PureRank vector. A quantitative sensitivity analysis with respect to changes in $\calC$ is beyond the scope of the present paper. Second, for very large networks, computing local importance vectors, that is, stationary distributions or eigenvectors, within large recurrent or transient classes can be computationally demanding. This cost arises because these vectors are typically computed by the power method, without teleportation-based acceleration as in PageRank. Third, PureRank is intentionally designed as a parameter-free recursive importance measure rather than as an application-specific algorithm. PureRank therefore does not incorporate domain-dependent tuning, which may limit its performance in highly specialized tasks, although it still provides a neutral reference against which more elaborate methods can be compared.

\subsection{Future directions}

A broader empirical investigation of PureRank on diverse real-world networks remains an important direction for future work. Section~\ref{sec:multi_attribute_networks} develops the theoretical extension of PureRank to multi-attribute networks, whereas large-scale empirical validation in such networks has yet to be undertaken. Empirical study is particularly important in multi-attribute settings, because the interaction between network connectivity and attribute patterns complicates both the interpretation of the resulting multi-attribute PureRank vectors and their aggregation into summary scores. Systematic comparisons in this setting, for example, between PureRank and a PageRank variant adapted to multi-attribute networks, would help clarify the distinctive properties of PureRank relative to PageRank.

Personalization of PureRank is another important direction for future work. The present PureRank formulation is nonpersonalized, because the random-surfer interpretation uses the specific distribution vector $\vc{\varpi}$ in (\ref{defn:eta}), which reflects the absence of prior information on node importance. A personalized extension would therefore require a new formulation that incorporates prior information through an exogenous preference distribution. Any such extension should be introduced with care, because excessive personalization may conflict with the design philosophy of PureRank.

\section*{Declaration of generative AI and AI-assisted technologies in the\\ manuscript preparation process}
During the preparation of this work, the author used Gemini and ChatGPT to edit and refine the manuscript and Python code. After using these tools, the author reviewed and edited the content as needed and takes full responsibility for the content of the published article.

\section*{Declaration of competing interest}

The author declares the following potential competing interest: Hiroyuki Masuyama is an inventor on Japanese Patent Application No.~2024--103831, which is currently pending. A non-author collaborator is also a co-inventor on the application.

\section*{Acknowledgments}
The author is grateful for valuable discussions with Professor Hiroshige Dan, whose comments and encouragement have contributed significantly to this study. This work was supported in part by JSPS KAKENHI Grant Number JP25K15006.


\appendix
\renewcommand{\thetable}{A.\arabic{table}}
\setcounter{equation}{0}
\renewcommand{\theequation}{A.\arabic{equation}}

\titleformat{\section}{\large\bfseries}{Appendix~\thesection.}{0.5em}{}
\section{Proofs}

\subsection{Proof of Theorem~\ref{thm:Katz-Centrality}}\label{proof:thm:Katz-Centrality}

\subsubsection{Case $S=R_k$ ($k\in\bbK$)}

The vector $\vc{\lambda}_{R_k}$ is the unique vector satisfying (\ref{defn:lambda_R_k}). Hence, the triple $(\vc{x}_{R_k}^{\top},\delta_{R_k},\beta_{R_k})=(\vc{\lambda}_{R_k},0,0)$ satisfies (\ref{cond_01})--(\ref{cond_04}) and simultaneously attains the lower bounds for both $\delta_{R_k}$ and $\beta_{R_k}$. Thus, this triple is optimal. Moreover, (\ref{cond_01}), (\ref{cond_02}), and $\vc{P}_{R_k}\vc{e}_{R_k}=\vc{e}_{R_k}$ together imply that $\delta_{R_k}=\beta_{R_k}|R_k|$. Consequently, the unique optimal pair is $(\delta_{R_k},\beta_{R_k})=(0,0)$, independent of the minimization order.

\subsubsection{Case $S=T$}

Setting $S=T$ in (\ref{cond_01}), and solving the resulting equation for $\vc{x}_T^{\top}$, we obtain
\begin{align}
\vc{x}_T^{\top}
=\beta_T \vc{e}_T^{\top}[\vc{I}-(1-\delta_T)\vc{P}_T]^{-1},
\label{eqn:231126-01}
\end{align}
where the inverse matrix exists for any $\delta_T\in[0,1]$ because $\rho(\vc{P}_T)<1$ (see Remark~\ref{rem:P_T}). From (\ref{eqn:231126-01}) and (\ref{cond_02}), we have
\begin{align}
\beta_T
={1 \over \vc{e}_T^{\top}[\vc{I}-(1-\delta_T)\vc{P}_T]^{-1}\vc{e}_T}
={1 \over \vc{e}_T^{\top}\dm\sum_{n=0}^{\infty}(1-\delta_T)^n(\vc{P}_T)^n\vc{e}_T}.
\label{eq:beta_T}
\end{align}
The denominator of (\ref{eq:beta_T}) is positive and nonincreasing in $\delta_T\in[0,1]$. Hence, $\beta_T$ is nondecreasing in $\delta_T$ and attains its minimum at $\delta_T=0$; setting $\delta_T=0$ in (\ref{eq:beta_T}) yields the minimum value $\beta_T^*$ defined in (\ref{defn:beta_S}) with $S=T$. Therefore, minimizing either $\delta_T$ or $\beta_T$ first yields the same pair $(\delta_T,\beta_T)=(0,\beta_T^*)$. Furthermore, substituting $(\delta_T,\beta_T)=(0,\beta_T^*)$ into (\ref{eqn:231126-01}) and applying (\ref{defn:beta_S}) and (\ref{defn:lambda_T}), we obtain
\begin{align*}
\vc{x}_T^{\top}
={
\vc{e}_T^{\top}(\vc{I}-\vc{P}_T)^{-1}
\over
\vc{e}_T^{\top}(\vc{I}-\vc{P}_T)^{-1}\vc{e}_T
}
={
\vc{\mu}_T(\vc{I}-\vc{P}_T)^{-1}
\over
\vc{\mu}_T(\vc{I}-\vc{P}_T)^{-1}\vc{e}_T
}
=\vc{\lambda}_T.
\end{align*}
Thus, $(\vc{x}_T^{\top},\delta_T,\beta_T)=(\vc{\lambda}_T,0,\beta_T^*)$ is the unique optimal solution, independently of the minimization order.

\subsubsection{Case $S=D$}

Setting $S=D$ in (\ref{cond_01}) and (\ref{cond_02}), and using $\vc{P}_D=\vc{O}$, we obtain
\begin{align}
\vc{x}_D^{\top}=\beta_D\vc{e}_D^{\top},\quad \vc{x}_D^{\top}\vc{e}_D=1.
\label{eqn:231202-01}
\end{align}
Eq.~(\ref{eqn:231202-01}) together with (\ref{defn:lambda_D}) uniquely determines $\beta_D$ and $\vc{x}_D^{\top}$ as
\begin{align*}
\beta_D&={1 \over |D|},
\\
\vc{x}_D^{\top}&={1 \over |D|}\vc{e}_D^{\top}=\vc{\mu}_D=\vc{\lambda}_D>\vc{0}.
\end{align*}
Since $\beta_D$ is determined independently of $\delta_D$, $(\vc{x}_D^{\top},\delta_D,\beta_D)=(\vc{\lambda}_D,0,1/|D|)$ is the unique optimal solution; moreover, the unique optimal pair $(\delta_D,\beta_D)=(0,1/|D|)$ is obtained independently of the minimization order.

\subsection{Proof of Theorem~\ref{thm:Z=N}}\label{proof:thm:GPR-01}

Substituting (\ref{defn:pi_T})--(\ref{defn:pi_D}) into (\ref{eqn:sum_lambda_S}) yields
\begin{align}
1
& =
{1 \over Z}
\left[
\sum_{k=1}^K |R_k| \vc{\lambda}_{R_k}\vc{e}_{R_k} + |D| \vc{\lambda}_D\vc{e}_D
\right.
\nonumber
\\
&\qquad\qquad
\left. {}
+ 
{ |T|  \over 1 + \theta_T}
\left( 
\sum_{k=1}^K \vc{\lambda}_T\vc{P}_{T,R_k}\vc{e}_{R_k} + \vc{\lambda}_T\vc{P}_{T,D}\vc{e}_D 
+ \vc{\lambda}_T\vc{e}_T
\right) 
\right]
\nonumber
\\
& =
{1 \over Z}
\left[
\sum_{k=1}^K |R_k| + |D|
+ 
{ |T|  \over 1 + \theta_T}
\left\{\vc{\lambda}_T
\left( 
\sum_{k=1}^K \vc{P}_{T,R_k}\vc{e}_{R_k} +\vc{P}_{T,D}\vc{e}_D 
\right) 
+ 1
\right\} 
\right],
\label{eq:2604040-01}
\end{align}
where the second equality follows from $\vc{\lambda}_S\vc{e}_S=1$ for all $S \in \calC$. Furthermore, (\ref{sum_rows_in_T}) and (\ref{eqn:theta_T}) imply that
\begin{align*}
\vc{\lambda}_T
\left( 
\sum_{k=1}^K \vc{P}_{T,R_k}\vc{e}_{R_k} +\vc{P}_{T,D}\vc{e}_D 
\right) 
&= \vc{\lambda}_T (\vc{e}_T - \vc{P}_T\vc{e}_T)
= 1 - \vc{\lambda}_T\vc{P}_T\vc{e}_T = \theta_T.
\end{align*}
Combining this equation with (\ref{eq:2604040-01}) leads to
\begin{align*}
Z 
&= \sum_{k=1}^K |R_k| + |D|
+ 
{ |T|  \over 1 + \theta_T}
(\theta_T + 1)
=\sum_{k=1}^K |R_k| + |D| + |T| = N,
\end{align*}
which completes the proof.

\subsection{Proof of Theorem~\ref{thm:classwise_PR_limit}}\label{proof:thm:classwise_PR_limit}

We derive classwise equations for the PageRank vector by rewriting (\ref{defn:PageRank-01}) according to the block-partitioned form in (\ref{partition_P}). Combining (\ref{defn:PageRank-01}), (\ref{eq:olP}), and (\ref{partition_P}), we obtain
\begin{align}
\vc{\gamma}^{(d)}
&= d \vc{\gamma}^{(d)}
\left\{  \vc{P} + 
\begin{pmatrix}
\vc{0}\\
\vc{e}_D
\end{pmatrix}\vc{\mu}_V
\right\}
+ (1-d)\vc{\mu}_V
\nonumber
\\
&= d \vc{\gamma}^{(d)} \vc{P} 
+ \left\{  d\vc{\gamma}_D^{(d)}\vc{e}_D + (1-d) \right\} \vc{\mu}_V.
\label{eq:PR_vector_form}
\end{align}
Using (\ref{partition_P}), we decompose (\ref{eq:PR_vector_form}) into classwise subvectors corresponding to the classes in $\calC$. For $k\in\bbK$,
\begin{align}
\vc{\gamma}_{R_k}^{(d)}
&= d \left\{
\vc{\gamma}_{R_k}^{(d)} \vc{P}_{R_k} + \vc{\gamma}_T^{(d)}\vc{P}_{T,R_k} \right\}
+ \left\{  d\vc{\gamma}_D^{(d)}\vc{e}_D + (1-d) \right\} 
{\vc{e}_{R_k}^{\top} \over N},
\label{eq:PR_Rk}
\end{align}
and
\begin{align}
\vc{\gamma}_T^{(d)}
&= d\vc{\gamma}_T^{(d)}\vc{P}_T
+ \left\{d\vc{\gamma}_D^{(d)}\vc{e}_D+(1-d)\right\}
{\vc{e}_T^{\top} \over N},
\label{eq:PR_T}
\\
\vc{\gamma}_D^{(d)}
&=d\vc{\gamma}_T^{(d)}\vc{P}_{T,D}
+ \left\{d\vc{\gamma}_D^{(d)}\vc{e}_D+(1-d)\right\}
{\vc{e}_D^{\top} \over N}.
\label{eq:PR_D}
\end{align}

\subsubsection{Proof of Statement (i)}
Fix $k \in \bbK$. Solving (\ref{eq:PR_Rk}) for $\vc{\gamma}_{R_k}^{(d)}$, we have
\begin{align}
\vc{\gamma}_{R_k}^{(d)}
&= \Big[ d\vc{\gamma}_T^{(d)}\vc{P}_{T,R_k}
+  \left\{ d\vc{\gamma}_D^{(d)}\vc{e}_D + (1-d) \right\}
{\vc{e}_{R_k}^{\top} \over N}
\Big]
(\vc{I}-d\vc{P}_{R_k})^{-1}
\nonumber
\\
&= \vc{\psi}_k^{(d)}(\vc{I}-d\vc{P}_{R_k})^{-1},
\label{eq:widetilde{gamma}_{R_k}^{(d)}}
\end{align}
where
\begin{align*}
\vc{\psi}_k^{(d)}
=
\Big[ d\vc{\gamma}_T^{(d)}\vc{P}_{T,R_k}
+  \left\{  d\vc{\gamma}_D^{(d)}\vc{e}_D + (1-d) \right\}
{\vc{e}_{R_k}^{\top} \over N}
\Big] > \vc{0}^{\top}.
\end{align*}
Since $\vc{P}_{R_k}\vc{e}_{R_k}=\vc{e}_{R_k}$, it follows that $(\vc{I}-d\vc{P}_{R_k})^{-1}\vc{e}_{R_k}=(1-d)^{-1}\vc{e}_{R_k}$. Hence, (\ref{eq:widetilde{gamma}_{R_k}^{(d)}}) implies that
\begin{align*}
\vc{\gamma}_{R_k}^{(d)}\vc{e}_{R_k}
= {\vc{\psi}_k^{(d)}\vc{e}_{R_k} \over 1 - d}.
\end{align*}
Substituting this equation and (\ref{eq:widetilde{gamma}_{R_k}^{(d)}}) into (\ref{eq:widetilde{gamma}_S^{(d)}}) with $S=R_k$, we have
\begin{align}
\widetilde{\vc{\gamma}}_{R_k}^{(d)}
={\vc{\psi}_k^{(d)} \over \vc{\psi}_k^{(d)}\vc{e}_{R_k}}
(1-d)(\vc{I}-d\vc{P}_{R_k})^{-1}.
\label{eq:gamma_tilde_Rk_def}
\end{align}
For the irreducible stochastic matrix $\vc{P}_{R_k}$, the Abel limit (see \cite[Theorem~3.3]{Scho12}) yields
\begin{align}
\lim_{d \uparrow 1}(1-d)(\vc{I}-d\vc{P}_{R_k})^{-1}
=\vc{e}_{R_k}\vc{\lambda}_{R_k}.
\label{eq:Abel_limit_R_k}
\end{align}
Equivalently, there exists a matrix function $\vc{E}_{R_k}^{(d)}$ such that $\vc{E}_{R_k}^{(d)} \to \vc{O}$ as $d \uparrow 1$ and 
\begin{align}
(1-d)(\vc{I}-d\vc{P}_{R_k})^{-1}
=\vc{e}_{R_k}\vc{\lambda}_{R_k} + \vc{E}_{R_k}^{(d)}.
\label{eq:Abel_limit_Rk}
\end{align}
Moreover, $\vc{\psi}_k^{(d)}/(\vc{\psi}_k^{(d)}\vc{e}_{R_k})$ is a probability vector and hence is bounded for all $d \in (0,1)$. Therefore,
\begin{align*}
\lim_{d \uparrow 1}
{\vc{\psi}_k^{(d)} \over \vc{\psi}_k^{(d)}\vc{e}_{R_k}} \vc{E}_{R_k}^{(d)}
= \vc{0}^{\top}.
\end{align*}
Applying (\ref{eq:Abel_limit_Rk}) to (\ref{eq:gamma_tilde_Rk_def}) now yields
\begin{align*}
\lim_{d \uparrow 1}\widetilde{\vc{\gamma}}_{R_k}^{(d)}=\vc{\lambda}_{R_k},
\end{align*}
which shows (\ref{eq:lim_tilde_gamma_R_k}). Statement (i) follows.

\subsubsection{Proof of Statement (ii)}

Assume $T \neq \varnothing$. Solving (\ref{eq:PR_T}) for $\vc{\gamma}_T^{(d)}$, we obtain
\begin{align}
\vc{\gamma}_T^{(d)}
= \left\{ d \vc{\gamma}_D^{(d)}\vc{e}_D + (1-d) \right\}
{\vc{e}_T^{\top} \over N} (\vc{I}-d\vc{P}_T)^{-1}.
\label{eq:gamma_T_def}
\end{align}
Substituting (\ref{eq:gamma_T_def}) into (\ref{eq:widetilde{gamma}_S^{(d)}}) with $S=T$, we obtain
\begin{align}
\widetilde{\vc{\gamma}}_T^{(d)}
= {\vc{e}_T^{\top} (\vc{I}-d\vc{P}_T)^{-1} 
\over \vc{e}_T^{\top} (\vc{I}-d\vc{P}_T)^{-1}\vc{e}_T
},
\label{eq:gamma_tilde_T_def}
\end{align}
where the denominator is positive because $(\vc{I}-d\vc{P}_T)^{-1}\vc{e}_T \ge \vc{e}_T$. Since $\rho(\vc{P}_T)<1$ and $\vc{P}_T \ge \vc{O}$, applying the monotone convergence theorem entrywise to $(\vc{I}-d\vc{P}_T)^{-1}=\sum_{n=0}^{\infty}(d\vc{P}_T)^n$ yields
\begin{align}
\lim_{d \uparrow 1}(\vc{I}-d\vc{P}_T)^{-1}
= \sum_{n=0}^{\infty} \lim_{d \uparrow 1} (d\vc{P}_T)^n
= \sum_{n=0}^{\infty} (\vc{P}_T)^n
= (\vc{I}-\vc{P}_T)^{-1}.
\label{eq:inv_limit_T}
\end{align}
Combining (\ref{eq:gamma_tilde_T_def}), (\ref{eq:inv_limit_T}), and (\ref{defn:lambda_T}) yields
\begin{align*}
\lim_{d \uparrow 1}
\widetilde{\vc{\gamma}}_T^{(d)}
= {\vc{e}_T^{\top} (\vc{I}-\vc{P}_T)^{-1} 
\over \vc{e}_T^{\top} (\vc{I}-\vc{P}_T)^{-1}\vc{e}_T
}
= \vc{\lambda}_T,
\end{align*}
which shows (\ref{eq:lim_tilde_gamma_T}). Statement (ii) follows.

\subsubsection{Proof of Statement (iii)} 

Assume $D \neq \varnothing$ and distinguish the two cases $T=\varnothing$ and $T\neq\varnothing$. We first consider the case $T=\varnothing$. Eq.~(\ref{eq:PR_D}) then reduces to
\begin{align*}
\vc{\gamma}_D^{(d)}
&= \left\{d\vc{\gamma}_D^{(d)}\vc{e}_D+(1-d)\right\} {\vc{e}_D^{\top} \over N}.
\end{align*}
Combining this equation with (\ref{eq:widetilde{gamma}_S^{(d)}}) for $S=D$, we obtain
\begin{align*}
\widetilde{\vc{\gamma}}_D^{(d)}
= {\vc{e}_D^{\top} \over \vc{e}_D^{\top}\vc{e}_D}
= {\vc{e}_D^{\top} \over |D|}
= \vc{\mu}_D,
\qquad
d \in (0,1),
\end{align*}
which shows that (\ref{eq:lim_tilde_gamma_D}) holds when $T=\varnothing$.

We next consider the case $T\neq\varnothing$. Substituting (\ref{eq:gamma_T_def}) into (\ref{eq:PR_D}) yields
\begin{align*}
\vc{\gamma}_D^{(d)}
&={1 \over N} 
\left\{ d \vc{\gamma}_D^{(d)}\vc{e}_D + (1-d) \right\}
\left\{
\vc{e}_D^{\top} + d\vc{e}_T^{\top} (\vc{I}-d\vc{P}_T)^{-1}\vc{P}_{T,D}
\right\}.
\end{align*}
Eq.~(\ref{eq:widetilde{gamma}_S^{(d)}}) with $S=D$ therefore implies that
\begin{align*}
\widetilde{\vc{\gamma}}_D^{(d)}
={
\vc{e}_D^{\top} + d\vc{e}_T^{\top} (\vc{I}-d\vc{P}_T)^{-1}\vc{P}_{T,D}
\over
\left\{
\vc{e}_D^{\top} + d\vc{e}_T^{\top} (\vc{I}-d\vc{P}_T)^{-1}\vc{P}_{T,D}
\right\}\vc{e}_D
}.
\end{align*}
Taking $d \uparrow 1$ in this equation, and using (\ref{eq:inv_limit_T}), we obtain
\begin{align}
\lim_{d \uparrow 1}\widetilde{\vc{\gamma}}_D^{(d)}
&=
{
\vc{e}_D^{\top} + \vc{e}_T^{\top} (\vc{I}-\vc{P}_T)^{-1}\vc{P}_{T,D}
\over
\left\{
\vc{e}_D^{\top} + \vc{e}_T^{\top} (\vc{I}-\vc{P}_T)^{-1}\vc{P}_{T,D}
\right\}\vc{e}_D
}
\nonumber
\\
&= {
|D|\vc{\mu}_D  + |T|\vc{\mu}_T (\vc{I}-\vc{P}_T)^{-1}\vc{P}_{T,D}
\over
|D| + |T|\vc{\mu}_T (\vc{I}-\vc{P}_T)^{-1}\vc{P}_{T,D}\vc{e}_D
},
\label{eq:260107-01}
\end{align}
where the second equality follows from $\vc{e}_D^{\top}=|D|\vc{\mu}_D$, $\vc{e}_T^{\top}=|T|\vc{\mu}_T$, and $\vc{e}_D^{\top}\vc{e}_D=|D|$. Moreover, (\ref{defn:lambda_T}) and (\ref{eqn:remark_theta_T}) imply that
\begin{align*}
\vc{\mu}_T (\vc{I}-\vc{P}_T)^{-1} 
= \vc{\lambda}_T \cdot \vc{\mu}_T (\vc{I}-\vc{P}_T)^{-1}\vc{e}_T
= {\vc{\lambda}_T \over \theta_T}.
\end{align*}
Applying this equation to (\ref{eq:260107-01}) yields (\ref{eq:lim_tilde_gamma_D}). Statement (iii) follows.

\subsection{Proof of Theorem~\ref{thm:LPR-class_T}}\label{proof:thm:LPR-class_T}

\subsubsection{Proof of Statement (i)}

We first show that $\vc{\lambda}_T$ is a stationary distribution vector of the stochastic matrix $\vc{Q}_T$. Eq.~(\ref{defn:lambda_T}) implies that $\vc{\lambda}_T$ is a probability vector. Substituting (\ref{eqn:theta_T}) into (\ref{eq:decomposition_theta_T}), and using $\vc{\lambda}_T\vc{e}_T=1$, we obtain
\begin{align*}
\vc{\lambda}_T
&= \vc{\lambda}_T\vc{P}_T + (1-\vc{\lambda}_T\vc{P}_T\vc{e}_T)\vc{\mu}_T
\\
&=  \vc{\lambda}_T\vc{P}_T + \vc{\lambda}_T (\vc{e}_T - \vc{P}_T\vc{e}_T)\vc{\mu}_T
\\
&= \vc{\lambda}_T \left[ \vc{P}_T + (\vc{e}_T - \vc{P}_T\vc{e}_T)\vc{\mu}_T \right]
= \vc{\lambda}_T\vc{Q}_T,
\end{align*}
where the last equality follows from (\ref{defn:Q_T}). Hence, $\vc{\lambda}_T$ is a stationary distribution vector of $\vc{Q}_T$.

We next show that $\vc{Q}_T$ is irreducible and aperiodic. Consider the network $\ol{G}=(V,\vc{P})$ with node set $V$ and weight matrix $\vc{P}$. The network $\ol{G}$ has the same reachability relation as the original network $G=(V,\vc{W})$. Let
\begin{align*}
T_{out} &= \{i \in T: 1 - [\vc{P}_T\vc{e}_T]_i>0\},
\end{align*}
and hence
\begin{align}
[\vc{e}_T - \vc{P}_T\vc{e}_T]_i>0
\quad \mbox{for all $i \in T_{out}$.}
\label{eq:ineq:T_{out}}
\end{align}
Eq.~(\ref{sum_rows_in_T}) and Definition~\ref{defn:class_T} imply that $T_{out}$ is nonempty and that every node in $T$ can reach some node in $T_{out}$ while remaining in Class $T$. Therefore, for each $i \in T$, there exists some $j^{(i)} \in T_{out}$ such that
\begin{align}
\left[\sum_{n=0}^{|T|-1} (\vc{P}_T)^n \right]_{i,j^{(i)}} > 0.
\label{eq:reachability_T}
\end{align}
Eqs.~(\ref{eq:ineq:T_{out}}) and (\ref{eq:reachability_T}) yield
\begin{align}
\left[\sum_{n=0}^{|T|-1} 
(\vc{P}_T)^n (\vc{e}_T - \vc{P}_T\vc{e}_T) \right]_i > 0,
\qquad i \in T.
\label{eqn:250601-02}
\end{align}
Using (\ref{defn:Q_T}) and (\ref{eqn:250601-02}), we obtain
\begin{align*}
\left[\sum_{n=1}^{|T|} (\vc{Q}_T)^n \right]_{i,j} 
\ge \left[\sum_{n=0}^{|T|-1} (\vc{P}_T)^n (\vc{e}_T - \vc{P}_T\vc{e}_T)\vc{\mu}_T \right]_{i,j} > 0,\qquad i,j \in T,
\end{align*}
which implies that $\vc{Q}_T$ is irreducible. Moreover, from (\ref{defn:Q_T}) and (\ref{eq:ineq:T_{out}}), we have
\begin{align*}
[\vc{Q}_T]_{i,j} &\ge [\vc{e}_T - \vc{P}_T\vc{e}_T]_i \cdot {1 \over |T|}> 0
\quad \mbox{for all $i \in T_{out}$ and $j \in T$},
\end{align*}
and hence $[\vc{Q}_T]_{i,i}>0$ for all $i \in T_{out}$. Because the irreducible stochastic matrix $\vc{Q}_T$ has a positive diagonal entry, $\vc{Q}_T$ is aperiodic. Statement (i) follows.

\subsubsection{Proof of Statement (ii)}

It is easily confirmed by induction that $\vc{\lambda}_T(n)$ is a probability vector and hence $\vc{\lambda}_T(n)\vc{e}_T =1$ for all $n=0,1,\dots$. Using $\vc{\lambda}_T(n)\vc{e}_T =1$ and (\ref{defn:Q_T}), we rewrite (\ref{lim:lambda_T-02b}) as follows: For $n=0,1,\dots$,
\begin{align*}
\vc{\lambda}_T(n+1)
&= \vc{\lambda}_T(n)\vc{P}_T 
+ \vc{\lambda}_T(n) \left(\vc{e}_T - \vc{P}_T\vc{e}_T \right) \vc{\mu}_T,\nonumber
\\
&= \vc{\lambda}_T(n)
\left[\vc{P}_T + \left(\vc{e}_T - \vc{P}_T\vc{e}_T \right) \vc{\mu}_T \right]
\nonumber
\\
&= \vc{\lambda}_T(n) \vc{Q}_T.
\end{align*}
This recursion together with (\ref{lim:lambda_T-02a}) implies that
\begin{align}
\vc{\lambda}_T(n) = \vc{\xi}_T(\vc{Q}_T)^n,
\qquad n=0,1,\dots.
\label{induction_lambda_T(n)}
\end{align}
Applying (\ref{eqn:lim_(Q_T)^n}) to (\ref{induction_lambda_T(n)}), and using $\vc{\xi}_T \vc{e}_T=1$, we obtain
\begin{align*}
\lim_{n\to \infty} \vc{\lambda}_T(n) 
= \vc{\xi}_T \lim_{n\to \infty} (\vc{Q}_T)^n 
= \vc{\xi}_T \vc{e}_T\vc{\lambda}_T
= \vc{\lambda}_T.
\end{align*}
Therefore, Statement (ii) follows.

\subsection{Proof of Lemma~\ref{lem:random_surfer-01}}\label{proof:lem_random_surfer-01}

For ease of exposition, we prove Statements (ii) and (iii) under the assumption $T\neq\varnothing$. When $T=\varnothing$, the same reasoning applies after deleting from the block expressions below all blocks, block rows, and block columns labeled by $T$, $\widehat{T}$, $R'$, and $D'$ according to the convention introduced in Section~\ref{subsec:node_classification}.

\subsubsection{Proof of Statement (i)}

Because $\vc{M}_{\widehat{T}}$ is a finite stochastic matrix, it has at least one stationary distribution vector. Let $\vc{\varphi}_{\widehat{T}} \ge \vc{0}^{\top}$ denote an arbitrary stationary distribution vector of $\vc{M}_{\widehat{T}}$, that is, a probability vector satisfying
\begin{align}
\vc{\varphi}_{\widehat{T}} = \vc{\varphi}_{\widehat{T}}\vc{M}_{\widehat{T}}.
\label{eqn:varphi-01}
\end{align}
We partition the probability vector $\vc{\varphi}_{\widehat{T}}$ as
\begin{align}
\vc{\varphi}_{\widehat{T}} =
\begin{blockarray}{c@{\,}ccc@{\,~}c}
 & R' & T & D'  \\ 
\begin{block}{c@{\,}(ccc)@{\,~}c}  
		& \vc{\varphi}_{R'} & \vc{\varphi}_{T} & \vc{\varphi}_{D'} \\
 \end{block}
\end{blockarray}
.
\label{partition_varphi}
\end{align}
Eqs.~(\ref{defn:M_T^{(e)}}) and (\ref{eqn:varphi-01}) imply that
\begin{align}
\vc{\varphi}_{R'} &= \vc{\varphi}_T\vc{P}_{T,R},
\label{eqn:xi_R'} \\
\vc{\varphi}_{D'} & = \vc{\varphi}_T\vc{P}_{T,D},
\label{eqn:xi_D'} \\
\vc{\varphi}_T
&= \vc{\varphi}_{R'}\vc{e}_R \vc{\mu}_T  + \vc{\varphi}_T \vc{P}_T
+ \vc{\varphi}_{D'}\vc{e}_D\vc{\mu}_T.
\label{eqn:xi_T}
\end{align}
Substituting (\ref{eqn:xi_R'}) and (\ref{eqn:xi_D'}) into (\ref{eqn:xi_T}) yields
\begin{align}
\vc{\varphi}_T
&=
  \vc{\varphi}_T \vc{P}_{T,R}\vc{e}_R \vc{\mu}_T
+ \vc{\varphi}_T \vc{P}_T
+ \vc{\varphi}_T \vc{P}_{T,D}\vc{e}_D \vc{\mu}_T
\nonumber \\
&= \vc{\varphi}_T
\left[
\vc{P}_T + (\vc{P}_{T,R}\vc{e}_R + \vc{P}_{T,D}\vc{e}_D)\vc{\mu}_T
\right]
\nonumber \\
&= \vc{\varphi}_T
\left[
\vc{P}_T + (\vc{e}_T - \vc{P}_T \vc{e}_T)\vc{\mu}_T
\right]
\nonumber \\
&= \vc{\varphi}_T \vc{Q}_T,
\label{eqn:23_0126-01}
\end{align}
where the third equality follows from (\ref{sum_rows_in_T}), and the last equality follows from (\ref{defn:Q_T}). If $\vc{\varphi}_T=\vc{0}^{\top}$, then (\ref{eqn:xi_R'}) and (\ref{eqn:xi_D'}) yield $\vc{\varphi}_{R'}=\vc{0}^{\top}$ and $\vc{\varphi}_{D'}=\vc{0}^{\top}$, which contradicts the fact that $\vc{\varphi}_{\widehat{T}}$ is a probability vector. Therefore, $\vc{\varphi}_T\ne \vc{0}^{\top}$. Recall that $\vc{Q}_T$ has the unique stationary distribution vector $\vc{\lambda}_T$; see Theorem~\ref{thm:LPR-class_T}. Hence, (\ref{eqn:23_0126-01}) implies that there exists $c > 0$ such that
\begin{align}
\vc{\varphi}_T = c\,\vc{\lambda}_T.
\label{eqn:xi_T-01}
\end{align}
Moreover, (\ref{eqn:xi_R'}) and (\ref{eqn:xi_D'}) yield
\begin{align}
\vc{\varphi}_{R'} &= c\,\vc{\lambda}_T\vc{P}_{T,R},
\label{eqn:xi_R'-01} \\
\vc{\varphi}_{D'} &= c\,\vc{\lambda}_T\vc{P}_{T,D},
\label{eqn:xi_D'-01}
\end{align}
respectively. Using (\ref{partition_varphi}) and (\ref{eqn:xi_T-01})--(\ref{eqn:xi_D'-01}), we obtain
\begin{align}
1
&= \vc{\varphi}_{\widehat{T}}\vc{e}_{\widehat{T}}
= \vc{\varphi}_{R'}\vc{e}_R + \vc{\varphi}_T\vc{e}_T + \vc{\varphi}_{D'}\vc{e}_D
\nonumber\\
&= c \left[
\vc{\lambda}_T \vc{e}_T + \vc{\lambda}_T(\vc{P}_{T,R}\vc{e}_R  + \vc{P}_{T,D}\vc{e}_D )
\right] \nonumber\\
&= c\left[
1 + \vc{\lambda}_T(\vc{e}_T  -  \vc{P}_T\vc{e}_T)
\right] \nonumber\\
&= c\left[
1 + (1  -  \vc{\lambda}_T\vc{P}_T\vc{e}_T)
\right] \nonumber\\
&= c(1 + \theta_T),
\label{eq:constant_c}
\end{align}
where the second last equality follows from $\vc{\lambda}_T\vc{e}_T=1$, and the last equality follows from (\ref{eqn:theta_T}). Eq.~(\ref{eq:constant_c}) shows that
\begin{align}
c = {1 \over 1 + \theta_T}.
\label{eq:constant_c_02}
\end{align}
Substituting (\ref{eqn:xi_T-01})--(\ref{eqn:xi_D'-01}) into (\ref{partition_varphi}), and using (\ref{eq:constant_c_02}) and (\ref{defn:lambda_{T^{(e)}}}), we have
\begin{align*}
\vc{\varphi}_{\widehat{T}}
=
{1 \over 1 + \theta_T}
\begin{blockarray}{ccc}
R' & T  & D'
\\
\begin{block}{(ccc)}
\vc{\lambda}_T\vc{P}_{T,R} &
\vc{\lambda}_T   &
\vc{\lambda}_T\vc{P}_{T,D}
\\
 \end{block}
\end{blockarray}
= \vc{\lambda}_{\widehat{T}},
\end{align*}
which implies that $\vc{\lambda}_{\widehat{T}}$ is the unique stationary distribution vector of $\vc{M}_{\widehat{T}}$. Therefore, Statement (i) follows.

\subsubsection{Proof of Statement (ii)}

Eqs.~(\ref{defn:P_R}) and (\ref{defn:matrix_M}) imply
\begin{align}
\vc{M}^{\nu}
=
\begin{blockarray}{c ccc cc}
       	& 
R_1    	& 
\cdots 	& 
R_K & 
\widehat{T}	& 
D   
\\ 
\begin{block}{c(ccc|c|c)}
R_1          	&
(\vc{P}_{R_1})^{\nu} 	&
			 	&
\mbox{\large $\vc{O}$}		 &
\vc{O}		 &
\vc{O}		 
\\	
\vdots       &
			 &
\ddots		 &
			 &
\vdots		 &
\vdots		 		 
\\
R_K          &
\mbox{\large $\vc{O}$}		 &
			 &
(\vc{P}_{R_K})^{\nu} &
\vc{O}		 &
\vc{O}		 
\\	
\BAhhline{~-----}
\widehat{T}               &
\vc{O}		 &
\cdots		 &
\vc{O}		 &
\rule{0mm}{5mm}(\vc{M}_{\widehat{T}})^{\nu}		 &
\vc{O}		 
\\
\BAhhline{~-----}		
D			&
\vc{O}		&
\cdots		 &
\vc{O}		&
\vc{O}		&
\vc{I}_D		 
\\			 
\end{block}
\end{blockarray}
~.
\label{eqn:23_0123-03}
\end{align}
Statement~(i) shows that the stochastic matrix $\vc{M}_{\widehat{T}}$ has the unique stationary distribution vector $\vc{\lambda}_{\widehat{T}}$. For each $k \in \bbK$, the irreducible stochastic matrix $\vc{P}_{R_k}$ has the unique stationary distribution vector $\vc{\lambda}_{R_k}$. Therefore, \cite[Chapter~V, Section~2, Theorem~2.1, p.~175]{Doob53} implies that
\begin{align*}
\lim_{n\to\infty}{1 \over n}\sum_{\nu=0}^{n-1} (\vc{M}_{\widehat{T}})^{\nu}
&= \vc{e}_{\widehat{T}}\vc{\lambda}_{\widehat{T}},
\notag
\\
\lim_{n\to\infty}{1 \over n}\sum_{\nu=0}^{n-1} (\vc{P}_{R_k})^{\nu}
&= \vc{e}_{R_k}\vc{\lambda}_{R_k}, \qquad k \in \bbK.
\end{align*}
Combining these limits with (\ref{eqn:23_0123-03}) yields (\ref{eqn:cesaro_limit-01}). Therefore, Statement~(ii) follows.

\subsubsection{Proof of Statement (iii)}

Eqs.~(\ref{defn:eta}) and (\ref{defn:lambda_D}) imply that $\vc{\varpi}$ can be expressed as
\begin{align}
\vc{\varpi}
=
{1 \over N}
\begin{blockarray}{ccc ccc c}
R_1 & \cdots & R_K & R' & T & D' & D 
\\ 
\begin{block}{(ccc|ccc|c)} 
\vc{e}_{R_1}^{\top}   & 
\cdots &
\vc{e}_{R_K}^{\top}   &
\vc{0}   &  
\vc{e}_T^{\top}   & 
\vc{0}   &  
|D|\vc{\mu}_D   
\\  
 \end{block}
\end{blockarray}
.
\label{eqn:hat{mu}}
\end{align}
Combining (\ref{eqn:cesaro_limit-01}) with (\ref{eqn:hat{mu}}) yields
\begin{align}
\lim_{n\to\infty}{1 \over n} \sum_{\nu=0}^{n-1} \vc{\varpi}\vc{M}^{\nu}
&=
{1 \over N}
\begin{blockarray}{ccc ccc c}
R_1 & \cdots & R_K & R' & T & D' & D 
\\ 
\begin{block}{(ccc|ccc|c)} 
\vc{e}_{R_1}^{\top}   & 
\cdots &
\vc{e}_{R_K}^{\top}   &
\vc{0}   &  
\vc{e}_T^{\top}   & 
\vc{0}   &  
|D|\vc{\mu}_D  
\\  
 \end{block}
\end{blockarray}
\notag
\\
& \qquad \times
\begin{blockarray}{c ccc c c}
       & 
R_1    & 
\cdots & 
R_K & 
\widehat{T}    & 
D   
\\ 
\begin{block}{c(ccc|c|c)}
R_1          	&
\vc{e}_{R_1}\vc{\lambda}_{R_1} 	&
			 	&
\mbox{\large $\vc{O}$}		 &
\vc{O}		 &
\vc{O}		 
\\	
\vdots       &
			 &
\ddots		 &
			 &
\vdots		 &
\vdots		 		 
\\
R_K          &
\mbox{\large $\vc{O}$}		 &
			 &
\vc{e}_{R_K}\vc{\lambda}_{R_K} &
\vc{O}		 &
\vc{O}		 
\\	
\BAhhline{~------}
\widehat{T}               &
\vc{O}		 &
\cdots		 &
\vc{O}		 &
\vc{e}_{\widehat{T}}\vc{\lambda}_{\widehat{T}}		 &
\vc{O}		 
\\
\BAhhline{~------}		
D			&
\vc{O}		&
\cdots		 &
\vc{O}		&
\vc{O}		&
\vc{I}_D		 
\\			 
\end{block}
\end{blockarray}
\nonumber
\\
&= {1 \over N}
\begin{blockarray}{cccccc}
R_1 & \cdots & R_K & \widehat{T}  & D
\\ 
\begin{block}{(ccccc)@{~}c@{\,}} 
|R_1| \vc{\lambda}_{R_1} &
\dots				&
|R_K| \vc{\lambda}_{R_K} &
|T| \vc{\lambda}_{\widehat{T}} &
|D| \vc{\mu}_D     &
\\
 \end{block}
\end{blockarray}
.
\label{eqn:240120-01}
\end{align}
Using (\ref{defn:lambda_{T^{(e)}}}), the decomposition $R'=R_1'\sqcup R_2'\sqcup \cdots \sqcup R_K'$, and (\ref{defn:P_{T,R}}), we rewrite the $\widehat{T}$-block on the right-hand side of (\ref{eqn:240120-01}) as
\begin{align*}
{|T| \over N}\vc{\lambda}_{\widehat{T}}
&=
{|T| \over N} {1 \over 1+\theta_T}
\begin{blockarray}{ccc}
R' & T  & D' \\ 
\begin{block}{(ccc)}
\vc{\lambda}_T\vc{P}_{T,R} 	&
\vc{\lambda}_T   			&  
\vc{\lambda}_T\vc{P}_{T,D}   
\\  
 \end{block}
\end{blockarray}
\notag
\\
&=
{1 \over N}
\begin{blockarray}{ccc cc}
R_1' & \cdots & R_K' & T  & D' 
\\ 
\begin{block}{(ccccc)}
\dm{ |T|\vc{\lambda}_T\vc{P}_{T,R_1} \over 1+\theta_T} 	&
\cdots 							&
\dm{ |T|\vc{\lambda}_T\vc{P}_{T,R_K} \over 1+\theta_T}	&
\dm{ |T|\vc{\lambda}_T \over 1+\theta_T}  				&  
\dm{ |T|\vc{\lambda}_T\vc{P}_{T,D} \over 1+\theta_T}  
\\  
 \end{block}
\end{blockarray}
~.
\end{align*}
Right-multiplying both sides of (\ref{eqn:240120-01}) by $\vc{F}$ in (\ref{defn:E}) adds each $R_k'$-block to the corresponding $R_k$-block for $k \in \bbK$ and adds the $D'$-block to the $D$-block. Hence,
\begin{align*}
\Bigl(
\mbox{The $R_k$-block of~}
\lim_{n\to\infty}{1 \over n}\sum_{\nu=0}^{n-1}\vc{\varpi}\vc{M}^{\nu}\vc{F} \Bigr)
&= {1 \over N}\left(|R_k| \vc{\lambda}_{R_k} 
+ \dm{|T|\vc{\lambda}_T \over 1+\theta_T} \vc{P}_{T,R_k}\right)
= \vc{\pi}_{R_k},
\\
\Bigl(
\mbox{The $T$-block of~} \lim_{n\to\infty}{1 \over n}\sum_{\nu=0}^{n-1}\vc{\varpi}\vc{M}^{\nu}\vc{F} 
\Bigr)
&= {1 \over N}\dm{|T|\vc{\lambda}_T \over 1+\theta_T} = \vc{\pi}_T,
\\
\Bigl(
\mbox{The $D$-block of~}\lim_{n\to\infty}{1 \over n}\sum_{\nu=0}^{n-1}\vc{\varpi}\vc{M}^{\nu}\vc{F}
\Bigr)
&= {1 \over N}\left(|D| \vc{\mu}_D + \dm{|T|\vc{\lambda}_T \over 1+\theta_T} \vc{P}_{T,D}\right)
= \vc{\pi}_D,
\end{align*}
where the equalities on the right follow from (\ref{eqn:PureRank_vectors}). Therefore, (\ref{eqn:vector_pi}) holds, and Statement~(iii) follows.

\subsection{Proof of Theorem~\ref{thm:random_surfer}}\label{proof:thm:random_surfer}

\subsubsection{Proof of Statement (i)}

We prove only the case $T \neq \varnothing$, because the case $T=\varnothing$ is proved in the same way with appropriate modifications. Eqs.~(\ref{eqn:vector_pi}) and (\ref{defn:E}) show that $\pi_j$ equals the long-run visit frequency to $j$ plus, when $j'$ exists, that to its copy $j'$. More precisely,
\begin{align}
\pi_j
&=
\left\{
\begin{array}{ll}
\dm\lim_{n\to \infty}
{1 \over n}
\sum_{\nu=0}^{n-1} \left[ \vc{\varpi} \vc{M}^{\nu}\right]_j, & j \in T,
\\[1.5em]
\dm\lim_{n\to \infty}
{1 \over n}
\sum_{\nu=0}^{n-1} \left[ \vc{\varpi} \vc{M}^{\nu}\right]_j
+
\dm\lim_{n\to \infty}
{1 \over n}
\sum_{\nu=0}^{n-1} \left[ \vc{\varpi} \vc{M}^{\nu}\right]_{j'}, & j \in V \setminus T= R \sqcup D.
\end{array}
\right.
\label{connection_pi_wh{pi}}
\end{align}
The definition of the Markov chain $\{\widehat{X}_{\nu}\}$ implies that, for all $j \in \widehat{V}$,
\begin{align}
\lim_{n\to \infty}
{1 \over n}
\sum_{\nu=0}^{n-1} \left[ \vc{\varpi} \vc{M}^{\nu}\right]_j
&=
\lim_{n\to \infty}
{1 \over n}
\sum_{\nu=0}^{n-1} \sum_{i \in \widehat{V}} [\vc{\varpi}]_i
\PP \! \left(\widehat{X}_{\nu} = j \mid \widehat{X}_0 = i
\right)
\nonumber
\\
&=
\lim_{n\to \infty}
{1 \over n}
\sum_{i \in \widehat{V}} [\vc{\varpi}]_i
\EE \!
\left[
\sum_{\nu=0}^{n-1} \one(\widehat{X}_{\nu} = j) \mid \widehat{X}_0 = i
\right].
\label{eqn:occupation_identity}
\end{align}
Substituting (\ref{eqn:occupation_identity}) into (\ref{connection_pi_wh{pi}}) yields (\ref{eqn:pi_j}). Therefore, Statement~(i) follows.

\subsubsection{Proof of Statement (ii)}

Conditioning on $\widehat{X}_0 \in T$, the Markov chain $\{\widehat{X}_n\}$ starts in Class $T$ with the uniform distribution vector $\vc{\mu}_T$, evolves according to $\vc{P}_T$ until it exits to Class $R'$ or Class $D'$ (with the exits to Class $R'$ and Class $D'$ governed by $\vc{P}_{T,R}$ and $\vc{P}_{T,D}$, respectively), and then immediately returns to Class $T$ with the same distribution vector $\vc{\mu}_T$. Hence, for each $m \in \bbZ_+$, the evolution after time $\tau_T(m)$ restarts from $\vc{\mu}_T$ and is independent of the past by the strong Markov property. Therefore, the sojourn times $\{C_T(m): m \in \bbZ_+\}$ are i.i.d. Furthermore, since $\widehat{X}_{\tau_T(m)}$ has distribution $\vc{\mu}_T$ for all $m \in \bbZ_+$, an application of the strong Markov property at time $\tau_T(m)$ yields
\begin{align}
&
\EE[C_T(m) \mid \widehat{X}_0 \in T]
\nonumber
\\
&\quad =
\sum_{i \in T} [\vc{\mu}_T]_i
\EE[C_T(m) \mid \widehat{X}_{\tau_T(m)} = i]
\nonumber
\\
&\quad =
\sum_{i \in T} [\vc{\mu}_T]_i
\sum_{n=1}^{\infty}
\PP(C_T(m) \ge n \mid \widehat{X}_{\tau_T(m)} = i)
\notag
\\
&\quad =
\sum_{i \in T} [\vc{\mu}_T]_i
\sum_{n=1}^{\infty}
\PP(\widehat{X}_{\tau_T(m)+\nu} \in T,~ \forall \nu=0,1,\dots,n-1 \mid \widehat{X}_{\tau_T(m)} = i)
\notag
\\
&\quad =
\sum_{i \in T} [\vc{\mu}_T]_i
\sum_{n=1}^{\infty} [ (\vc{P}_T)^{n-1}\vc{e}_T]_i
= \vc{\mu}_T \sum_{n=1}^{\infty} (\vc{P}_T)^{n-1}\vc{e}_T
\notag
\\
&\quad =
\vc{\mu}_T  (\vc{I} - \vc{P}_T)^{-1}\vc{e}_T.
\label{eqn:240121-01}
\end{align}
Combining (\ref{eqn:240121-01}) with (\ref{eqn:remark_theta_T}) yields (\ref{defn:theta_T-a})--(\ref{defn:theta_T-c}). Therefore, Statement~(ii) follows.


\begin{table}[!p]
\centering
\caption{Summary of network statistics for the \texttt{twitter\_combined}, \texttt{cit-HepPh}, and \texttt{ca-AstroPh} networks. Panel~A reports basic network statistics and the numbers of nodes in Classes $R$, $T$, and $D$. Panel~B reports, for the largest class in the node classification $\calC$, its size, together with its percentage of the total number of nodes, the iteration count, and the median wall-clock runtime (sec) of the power iteration used to compute the local importance vector for that class.}
\label{tab:network_stats_all}

\textbf{Panel A: Network statistics.}

\setlength{\tabcolsep}{8pt}
\scalebox{0.8}{%
\begin{tabular}{l|rrr}
\hline
Statistic   & \textbf{\texttt{twitter\_combined}} & \textbf{\texttt{cit-HepPh}} & \textbf{\texttt{ca-AstroPh}} \\
\hline
Number of links      	& 1,768,149 & 421,578 & 396,160 \\
Number of nodes      	& 81,306    & 34,546  & 18,772  \\
Nodes in Class $R$     	& 624       & 7       & 18,772  \\
Nodes in Class $T$     	& 69,473    & 32,151  & 0       \\
Nodes in Class $D$     	& 11,209    & 2,388   & 0       \\
\hline
\end{tabular}
}
\vspace{0.6em}

\textbf{Panel B: Statistics for the largest class in the node classification}

\setlength{\tabcolsep}{10pt}
\scalebox{0.8}{%
\begin{tabular}{l|cccc}
\hline
Network & Class & Size (\%) & Iter. & Median runtime (sec)\\
\hline
\texttt{twitter\_combined} & Transient & $69{,}473(85.4\%)$ & 2{,}816 & 9.506244\\
\texttt{cit-HepPh}         & Transient & $32{,}151(93.1\%)$ &      45 & 0.042390\\
\texttt{ca-AstroPh}        & Recurrent & $17{,}903(95.4\%)$ & 1{,}877 & 0.836499\\
\hline
\end{tabular}
}
\begin{flushleft}
\footnotesize
\emph{Note.~} In Panel~A, Class $R$ denotes the union of all recurrent classes in the node classification $\calC$. In Panel~B, ``Recurrent'' indicates that the largest class in $\calC$ is a recurrent class, and the corresponding reported size is the size of the largest recurrent class rather than the total number of recurrent nodes. The runtime reported in Panel~B is the median of the five measured runs obtained after excluding one warm-up run.
\end{flushleft}
\end{table}

\begin{table}[!p]
\centering
\caption{Distribution of recurrent class sizes in the \texttt{twitter\_combined}, \texttt{cit-HepPh}, and \texttt{ca-AstroPh} networks. Each cell reports the number of recurrent classes of the indicated size.}
\label{tab:recurrent_class_distribution_all}
\setlength{\tabcolsep}{3.5pt} 

\scalebox{0.8}{%
\begin{tabular}{l|r|cccccccccccccc}
\hline
\multirow{2}{*}{Network}& & \multicolumn{14}{c}{Number of Recurrent Classes by Size} \\
\cline{2-16} 
 & Size & 1 & 2 & 3 & 4 & 5 & 6 & 7 & 8 & 9 & 10 & 12 & 14 & 18 & 17,903 \\
\hline
\texttt{twitter\_combined} & \multirow{3}{*}{Number}  & 2 & 222 & 28 & 9 & 5 & 2 & 1 & 0 & 0 & 0 & 0 & 1 & 0 & 0 \\
\texttt{cit-HepPh}         & & 5 & 1   & 0  & 0 & 0 & 0 & 0 & 0 & 0 & 0 & 0 & 0 & 0 & 0 \\
\texttt{ca-AstroPh}       &  & 1 & 140 & 84 & 36& 14& 3 & 3 & 3 & 1 & 2 & 1 & 0 & 1 & 1 \\
\hline
\end{tabular}
}
\end{table}

\begin{table}[!p] 
\centering
\caption{Comparison of PureRank and PageRank with damping factor $d$ for the \texttt{twitter\_combined}, \texttt{cit-HepPh}, and \texttt{ca-AstroPh} networks. Panel~A reports Top-100 Overlap (\%), Kendall's $\tau_b$, and Pearson correlation coefficient (PCC) between PureRank and PageRank with damping factor $d$. Panel~B reports computational cost: the iteration count and the median wall-clock runtime (sec) for PageRank, and the median wall-clock runtime (sec) for PureRank.}
\label{tab:comparison_page_pure_all}
\setlength{\tabcolsep}{4pt} 

{Panel A: Similarity metrics.}

\scalebox{0.8}{%
\begin{tabular}{l|rrr|rrr|rrr}
\hline
    & \multicolumn{3}{c|}{\texttt{twitter\_combined}} & \multicolumn{3}{c|}{\texttt{cit-HepPh}} & \multicolumn{3}{c}{\texttt{ca-AstroPh}} 
\\
\cline{2-4} \cline{5-7} \cline{8-10}
$d$ & \multicolumn{3}{c|}{vs. Pure} & \multicolumn{3}{c|}{vs. Pure} & \multicolumn{3}{c}{vs. Pure} 
\\
& Top-100 & $\tau_b$ & PCC & Top-100 & $\tau_b$ & PCC & Top-100 & $\tau_b$ & PCC 
\\
\hline
0.1   & 64 & 0.434 & 0.806 & 52 & 0.858 & 0.785 & 42 & 0.564 & 0.743 
\\
0.2   & 66 & 0.450 & 0.817 & 57 & 0.870 & 0.825 & 45 & 0.577 & 0.764 
\\
0.3   & 66 & 0.465 & 0.827 & 68 & 0.881 & 0.862 & 48 & 0.591 & 0.786 
\\
0.4   & 67 & 0.481 & 0.835 & 76 & 0.892 & 0.897 & 55 & 0.607 & 0.810 
\\
0.5   & 70 & 0.499 & 0.839 & 80 & 0.903 & 0.927 & 60 & 0.626 & 0.835 
\\
0.6   & 73 & 0.518 & 0.834 & 83 & 0.913 & 0.953 & 67 & 0.650 & 0.863 
\\
0.7   & 75 & 0.541 & 0.811 & 88 & 0.923 & 0.973 & 74 & 0.681 & 0.894 
\\
0.8   & 76 & 0.571 & 0.745 & 90 & 0.932 & 0.987 & 82 & 0.723 & 0.928 
\\
0.85  & 78 & 0.590 & 0.677 & 91 & 0.937 & 0.991 & 85 & 0.752 & 0.947 
\\
0.9   & 81 & 0.615 & 0.570 & 92 & 0.942 & 0.991 & 87 & 0.791 & 0.967 
\\
0.95  & 83 & 0.649 & 0.406 & 97 & 0.947 & 0.976 & 91 & 0.851 & 0.986 
\\
0.99  & 60 & 0.694 & 0.223 & 95 & 0.950 & 0.649 & 98 & 0.940 & 0.999 
\\
0.999 & 3  & 0.709 & 0.174 & 93 & 0.951 & 0.093 & 99 & 0.966 & 1.000 
\\
\hline
\end{tabular}
}

\vspace{0.6em}

{Panel B: Computational cost.}

\scalebox{0.8}{%
\begin{tabular}{l|rr|rr|rr}
\hline
    & \multicolumn{2}{c|}{\texttt{twitter\_combined}} & \multicolumn{2}{c|}{\texttt{cit-HepPh}} & \multicolumn{2}{c}{\texttt{ca-AstroPh}}
\\
\cline{2-3} \cline{4-5} \cline{6-7}
Measure & Iter. & Time (sec) & Iter. & Time (sec) & Iter. & Time (sec)
\\
\hline
PureRank & -- & 18.713799 & -- & 1.944775 & -- & 2.079275
\\
\hline
PageRank (0.1) & 9 & 1.915265 & 8 & 0.321054 & 9 & 0.256067
\\
PageRank (0.2) & 12 & 1.923997 & 11 & 0.322644 & 12 & 0.257042
\\
PageRank (0.3) & 16 & 1.935446 & 14 & 0.324699 & 16 & 0.258966
\\
PageRank (0.4) & 21 & 1.954718 & 18 & 0.327138 & 20 & 0.260982
\\
PageRank (0.5) & 27 & 1.973523 & 23 & 0.330480 & 26 & 0.263849
\\
PageRank (0.6) & 37 & 2.008050 & 31 & 0.335339 & 34 & 0.267984
\\
PageRank (0.7) & 52 & 2.059397 & 45 & 0.344180 & 48 & 0.273482
\\
PageRank (0.8) & 83 & 2.166773 & 71 & 0.363557 & 75 & 0.287322
\\
PageRank (0.85) & 114 & 2.258419 & 97 & 0.380400 & 102 & 0.300424
\\
PageRank (0.9) & 175 & 2.453739 & 150 & 0.414509 & 157 & 0.324843
\\
PageRank (0.95) & 356 & 3.018558 & 306 & 0.513157 & 321 & 0.400706
\\
PageRank (0.99) & 1,801 & 7.515992 & 1,517 & 1.315879 & 1,637 & 1.011284
\\
PageRank (0.999) & 18,091 & 60.488411 & 11,831 & 9.226185 & 16,436 & 8.676888
\\
\hline
\end{tabular}}
\begin{flushleft}
\footnotesize
\emph{Notes.~} For PageRank, the reported time for each $d$ is computed as the sum of (i) the one-time cost of constructing the weight matrix $\vc{W}$ and the normalized weight matrix $\vc{P}$, measured once per network; and (ii) the median wall-clock time of the power iteration for PageRank with damping factor $d$, computed from five measured runs after excluding one warm-up run. This convention is adopted to clarify the impact of $d$ on the iterative component of the computational cost. For PureRank, the reported time is the median of five measured total wall-clock times obtained after excluding one warm-up run, where each run includes node classification, computation of local importance vectors (including construction of $\vc{W}$ and $\vc{P}$), and RDI-L2G construction. No single iteration count is reported for PureRank because the local importance vectors are computed classwise, requiring multiple separate power iterations.
\end{flushleft}
\end{table}

\begin{table}[!p]
\centering
\caption{Top-100 node distribution (\%) and classwise average score per node, computed over all nodes in each class, for PureRank and PageRank ($d$) on the \texttt{twitter\_combined} and \texttt{cit-HepPh} networks. PageRank ($d$) denotes PageRank with damping factor $d$.}
\label{tab:comparison_page_pure_top100}

\par\smallskip
\noindent{Panel A: Results on the \texttt{twitter\_combined} network.}
\par\smallskip

\scalebox{0.8}{%
\begin{tabular}{lcccccc}
\hline
\multirow{2}{*}{Measure} & \multicolumn{3}{c}{Top-100 Node Distribution (\%)} & \multicolumn{3}{c}{Average Score per Node (by Class)} \\
\cline{2-7}
 & Class $R$ & Class $T$ & Class $D$ & Class $R$ & Class $T$ & Class $D$ \\
\hline
PureRank         & 1  & 95 & 4  & $1.81\times10^{-5}$ & $1.19\times10^{-5}$ & $1.46\times10^{-5}$ \\
PageRank (0.1)   & 2  & 92 & 6  & $1.37\times10^{-5}$ & $1.24\times10^{-5}$ & $1.16\times10^{-5}$ \\
PageRank (0.2)   & 2  & 91 & 7  & $1.53\times10^{-5}$ & $1.25\times10^{-5}$ & $1.09\times10^{-5}$ \\
PageRank (0.3)   & 2  & 91 & 7  & $1.72\times10^{-5}$ & $1.26\times10^{-5}$ & $1.01\times10^{-5}$ \\
PageRank (0.4)   & 3  & 90 & 7  & $1.97\times10^{-5}$ & $1.27\times10^{-5}$ & $9.34\times10^{-6}$ \\
PageRank (0.5)   & 3  & 90 & 7  & $2.29\times10^{-5}$ & $1.28\times10^{-5}$ & $8.51\times10^{-6}$ \\
PageRank (0.6)   & 3  & 92 & 5  & $2.74\times10^{-5}$ & $1.29\times10^{-5}$ & $7.61\times10^{-6}$ \\
PageRank (0.7)   & 4  & 91 & 5  & $3.44\times10^{-5}$ & $1.30\times10^{-5}$ & $6.65\times10^{-6}$ \\
PageRank (0.8)   & 11 & 84 & 5  & $4.75\times10^{-5}$ & $1.31\times10^{-5}$ & $5.60\times10^{-6}$ \\
PageRank (0.85)  & 11 & 84 & 5  & $5.99\times10^{-5}$ & $1.30\times10^{-5}$ & $5.02\times10^{-6}$ \\
PageRank (0.9)   & 11 & 85 & 4  & $8.36\times10^{-5}$ & $1.29\times10^{-5}$ & $4.38\times10^{-6}$ \\
PageRank (0.95)  & 13 & 83 & 4  & $1.49\times10^{-4}$ & $1.25\times10^{-5}$ & $3.61\times10^{-6}$ \\
PageRank (0.99)  & 41 & 56 & 3  & $5.11\times10^{-4}$ & $9.43\times10^{-6}$ & $2.30\times10^{-6}$ \\
PageRank (0.999) & 98 & 2  & 0  & $1.31\times10^{-3}$ & $2.50\times10^{-6}$ & $5.78\times10^{-7}$ \\
\hline
\end{tabular}
}

\medskip

\par\medskip
\noindent{Panel B: Results on the \texttt{cit-HepPh} network.}
\par\smallskip

\scalebox{0.8}{%
\begin{tabular}{lcccccc}
\hline
\multirow{2}{*}{Measure} & \multicolumn{3}{c}{Top-100 Node Distribution (\%)} & \multicolumn{3}{c}{Average Score per Node (by Class)} \\
\cline{2-7}
 & Class $R$ & Class $T$ & Class $D$ & Class $R$ & Class $T$ & Class $D$ \\
\hline
PureRank         & 0  & 43  & 57  & $7.06\times10^{-5}$ & $2.24\times10^{-5}$ & $1.17\times10^{-4}$ \\
PageRank (0.1)   & 0  & 73  & 27  & $3.28\times10^{-5}$ & $2.87\times10^{-5}$ & $3.29\times10^{-5}$ \\
PageRank (0.2)   & 0  & 69  & 31  & $3.74\times10^{-5}$ & $2.83\times10^{-5}$ & $3.74\times10^{-5}$ \\
PageRank (0.3)   & 0  & 63  & 37  & $4.32\times10^{-5}$ & $2.79\times10^{-5}$ & $4.26\times10^{-5}$ \\
PageRank (0.4)   & 0  & 58  & 42  & $5.09\times10^{-5}$ & $2.75\times10^{-5}$ & $4.85\times10^{-5}$ \\
PageRank (0.5)   & 0  & 53  & 47  & $6.22\times10^{-5}$ & $2.70\times10^{-5}$ & $5.52\times10^{-5}$ \\
PageRank (0.6)   & 0  & 51  & 49  & $8.01\times10^{-5}$ & $2.64\times10^{-5}$ & $6.26\times10^{-5}$ \\
PageRank (0.7)   & 0  & 48  & 52  & $1.13\times10^{-4}$ & $2.58\times10^{-5}$ & $7.09\times10^{-5}$ \\
PageRank (0.8)   & 0  & 49  & 51  & $1.85\times10^{-4}$ & $2.51\times10^{-5}$ & $7.99\times10^{-5}$ \\
PageRank (0.85)  & 1  & 48  & 51  & $2.63\times10^{-4}$ & $2.48\times10^{-5}$ & $8.47\times10^{-5}$ \\
PageRank (0.9)   & 2  & 47  & 51  & $4.24\times10^{-4}$ & $2.44\times10^{-5}$ & $8.96\times10^{-5}$ \\
PageRank (0.95)  & 2  & 44  & 54  & $9.20\times10^{-4}$ & $2.39\times10^{-5}$ & $9.44\times10^{-5}$ \\
PageRank (0.99)  & 5  & 42  & 53  & $4.83\times10^{-3}$ & $2.29\times10^{-5}$ & $9.59\times10^{-5}$ \\
PageRank (0.999) & 7  & 42  & 51  & $3.75\times10^{-2}$ & $1.74\times10^{-5}$ & $7.39\times10^{-5}$ \\
\hline
\end{tabular}
}
\end{table}


\begin{table}[!p]
\centering
\caption{Classwise Kendall's $\tau_b$ and Pearson correlation coefficient (PCC) between PureRank and PageRank ($d$) for the \texttt{twitter\_combined} and \texttt{cit-HepPh} networks.}
\label{tab:classwise_tau_pcc}

\par\smallskip
\noindent{Panel A: Results on the \texttt{twitter\_combined} network.}
\par\smallskip

\scalebox{0.8}{%
\begin{tabular}{rcccccc}
\hline
\multirow{2}{*}{$d$} & \multicolumn{3}{c}{Kendall's $\tau_b$} & \multicolumn{3}{c}{PCC} \\
\cline{2-7}
 & Class $R$ & Class $T$ & Class $D$ & Class $R$ & Class $T$ & Class $D$ \\
\hline
0.1 & 0.609 & 0.617 & 0.634 & 0.993 & 0.796 & 0.913 \\
0.2 & 0.617 & 0.638 & 0.666 & 0.979 & 0.811 & 0.920 \\
0.3 & 0.616 & 0.659 & 0.693 & 0.956 & 0.827 & 0.928 \\
0.4 & 0.618 & 0.681 & 0.717 & 0.927 & 0.844 & 0.936 \\
0.5 & 0.624 & 0.706 & 0.741 & 0.893 & 0.862 & 0.945 \\
0.6 & 0.632 & 0.733 & 0.767 & 0.856 & 0.881 & 0.955 \\
0.7 & 0.640 & 0.765 & 0.796 & 0.818 & 0.903 & 0.965 \\
0.8 & 0.649 & 0.806 & 0.833 & 0.780 & 0.929 & 0.977 \\
0.85 & 0.651 & 0.833 & 0.856 & 0.762 & 0.944 & 0.983 \\
0.9 & 0.658 & 0.867 & 0.885 & 0.743 & 0.961 & 0.989 \\
0.95 & 0.670 & 0.915 & 0.926 & 0.725 & 0.981 & 0.996 \\
0.99 & 0.675 & 0.976 & 0.979 & 0.711 & 0.997 & 1.000 \\
0.999 & 0.672 & 0.997 & 0.998 & 0.708 & 1.000 & 1.000 \\
\hline
\end{tabular}
}

\medskip

\par\medskip
\noindent{Panel B: Results on the \texttt{cit-HepPh} network.}
\par\smallskip

\scalebox{0.8}{%
\begin{tabular}{rcccccc}
\hline
\multirow{2}{*}{$d$} & \multicolumn{3}{c}{Kendall's $\tau_b$} & \multicolumn{3}{c}{PCC} \\
\cline{2-7}
 & Class $R$ & Class $T$ & Class $D$ & Class $R$ & Class $T$ & Class $D$ \\
\hline
0.1 & 0.619 & 0.901 & 0.799 & 0.132 & 0.830 & 0.829 \\
0.2 & 0.619 & 0.913 & 0.827 & 0.158 & 0.857 & 0.869 \\
0.3 & 0.619 & 0.925 & 0.853 & 0.201 & 0.884 & 0.902 \\
0.4 & 0.619 & 0.936 & 0.877 & 0.274 & 0.910 & 0.930 \\
0.5 & 0.619 & 0.947 & 0.900 & 0.390 & 0.934 & 0.953 \\
0.6 & 0.810 & 0.958 & 0.921 & 0.555 & 0.955 & 0.971 \\
0.7 & 0.810 & 0.969 & 0.942 & 0.747 & 0.974 & 0.984 \\
0.8 & 1.000 & 0.979 & 0.962 & 0.904 & 0.988 & 0.993 \\
0.85 & 1.000 & 0.985 & 0.972 & 0.953 & 0.993 & 0.996 \\
0.9 & 1.000 & 0.990 & 0.981 & 0.982 & 0.997 & 0.998 \\
0.95 & 1.000 & 0.995 & 0.991 & 0.996 & 0.999 & 1.000 \\
0.99 & 1.000 & 0.999 & 0.998 & 1.000 & 1.000 & 1.000 \\
0.999 & 1.000 & 1.000 & 1.000 & 1.000 & 1.000 & 1.000 \\
\hline
\end{tabular}
}
\end{table}
\FloatBarrier
\clearpage
\section*{Data availability}

The datasets analyzed in this study are publicly available from the Stanford Large Network Dataset Collection.


\begin{thebibliography}{10}
\expandafter\ifx\csname url\endcsname\relax
  \def\url#1{\texttt{#1}}\fi
\expandafter\ifx\csname urlprefix\endcsname\relax\def\urlprefix{URL }\fi
\expandafter\ifx\csname href\endcsname\relax
  \def\href#1#2{#2} \def\path#1{#1}\fi

\bibitem{Newm10}
M.E.J.~Newman, Networks: an Introduction, Oxford University Press, Oxford,
2010.

\bibitem{Saxe20}
A.~Saxena, S.~Iyengar, Centrality measures in complex networks: a survey.
\href{http://arxiv.org/abs/2011.07190}{arXiv:2011.07190}, 2020.

\bibitem{Seel49}
J.R.~Seeley, The net of reciprocal influence; a problem in treating
sociometric data, Can. J. Psychol. 3~(4) (1949) 234--240.

\bibitem{Bona72}
P.~Bonacich, Factoring and weighting approaches to status scores and clique
identification, J. Math. Sociol. 2~(1) (1972) 113--120.

\bibitem{Brod00}
A.~Broder, R.~Kumar, F.~Maghoul, P.~Raghavan, S.~Rajagopalan, R.~Stata,
A.~Tomkins, J.~Wiener, Graph structure in the Web, Comput. Netw. 33~(1--6)
(2000) 309--320.

\bibitem{Katz53}
L.~Katz, A new status index derived from sociometric analysis, Psychometrika
18~(1) (1953) 39--43.

\bibitem{Brin98}
S.~Brin, L.~Page, The anatomy of a large-scale hypertextual web search engine,
Comput. Netw. ISDN Syst. 30~(1--7) (1998) 107--117.

\bibitem{Lang06}
A.N.~Langville, C.D.~Meyer, Google's PageRank and Beyond: the Science of
Search Engine Rankings, Princeton University Press, Princeton, NJ, 2006.

\bibitem{Was23}
T.~W\k{a}s, O.~Skibski, Axiomatic characterization of PageRank, Artif. Intell.
318 (2023) 103900.

\bibitem{Park19}
S.~Park, W.~Lee, B.~Choe, S.-G.~Lee, A survey on personalized PageRank
computation algorithms, IEEE Access 7 (2019) 163049--163062.

\bibitem{Yang24}
M.~Yang, H.~Wang, Z.~Wei, S.~Wang, J.-R.~Wen, Efficient algorithms for
personalized PageRank computation: a survey, IEEE Trans. Knowl. Data Eng.
36~(9) (2024) 4582--4602.

\bibitem{Jaya24}
R.~Jayaram, J.~{\L}\k{a}cki, S.~Mitrovi{\'c}, K.~Onak, P.~Sankowski,
Dynamic PageRank: algorithms and lower bounds, in: Proceedings of the 51st
International Colloquium on Automata, Languages, and Programming (ICALP
2024), Vol.~297 of Leibniz International Proceedings in Informatics (LIPIcs),
2024, pp.~90:1--90:19,
\href{https://doi.org/10.4230/LIPIcs.ICALP.2024.90}{https://doi.org/10.4230/LIPIcs.ICALP.2024.90}

\bibitem{Hou21}
G.~Hou, X.~Chen, S.~Wang, Z.~Wei, Massively parallel algorithms for
personalized PageRank, Proc. VLDB Endow. 14~(9) (2021) 1668--1680.

\bibitem{Berg08}
C.T.~Bergstrom, J.D.~West, M.A.~Wiseman, The Eigenfactor metrics,
J. Neurosci. 28~(45) (2008) 11433--11434.

\bibitem{Gyon04}
Z.~Gy{\"o}ngyi, H.~Garcia-Molina, J.O.~Pedersen, Combating web spam with
TrustRank, in: Proceedings of the 30th International Conference on Very Large
Data Bases, Morgan Kaufmann, 2004, pp.~576--587.

\bibitem{He17}
X.~He, M.~Gao, M.-Y.~Kan, D.~Wang, BiRank: towards ranking on bipartite
graphs, IEEE Trans. Knowl. Data Eng. 29~(1) (2017) 57--71.

\bibitem{Morr05}
J.L.~Morrison, R.~Breitling, D.J.~Higham, D.R.~Gilbert, Generank: using
search engine technology for the analysis of microarray experiments,
BMC Bioinform. 6~(1) (2005) 233.

\bibitem{Batt12}
S.~Battiston, M.~Puliga, R.~Kaushik, P.~Tasca, G.~Caldarelli, DebtRank: too
central to fail? Financial networks, the FED and systemic risk, Sci. Rep. 2
(2012) 541.

\bibitem{Miha04}
R.~Mihalcea, P.~Tarau, TextRank: bringing order into texts, in: Proceedings of
the 2004 Conference on Empirical Methods in Natural Language Processing,
Association for Computational Linguistics, 2004, pp.~404--411.

\bibitem{Glei15}
D.F.~Gleich, PageRank beyond the web, SIAM Rev. 57~(3) (2015) 321--363.

\bibitem{Gova08}
A.Y.~Govan, Ranking Theory with Application to Popular Sports (Ph.D. thesis),
North Carolina State University, 2008.

\bibitem{Lang12}
A.N.~Langville, C.D.~Meyer, Who's \#1?: the Science of Rating and Ranking,
Princeton University Press, Princeton, NJ, 2012.

\bibitem{Avra08}
K.~Avrachenkov, N.~Litvak, K.S.~Pham, A singular perturbation approach for
choosing the PageRank damping factor, Internet Math. 5~(1--2) (2008) 47--69.

\bibitem{Bres10}
M.~Bressan, E.~Peserico, Choose the damping, choose the ranking? J. Discrete
Algorithms 8~(2) (2010) 199--213.

\bibitem{Lu11}
L.~L{\"u}, Y.-C.~Zhang, C.H.~Yeung, T.~Zhou, Leaders in social networks,
the Delicious case, PLOS ONE 6~(6) (2011) e21202.

\bibitem{snap_datasets}
J.~Leskovec, A.~Krevl, SNAP Datasets: Stanford Large Network Dataset
Collection, Online Dataset Collection, 2014.

\bibitem{Frah20}
K.M.~Frahm, D.L.~Shepelyansky, Google matrix analysis of bi-functional
SIGNOR network of protein--protein interactions, Phys. A Stat. Mech. Appl.
559 (2020) 125019.

\bibitem{Frah19}
K.M.~Frahm, D.L.~Shepelyansky, Ising-PageRank model of opinion formation on
social networks, Phys. A Stat. Mech. Appl. 526 (2019) 121069.

\bibitem{Brem20}
P.~Br{\'e}maud, Markov Chains: Gibbs Fields, Monte Carlo Simulation and
Queues, second ed, Springer, Cham, 2020.

\bibitem{Avra10}
K.~Avrachenkov, V.~Borkar, D.~Nemirovsky, Quasi-stationary distributions as
centrality measures for the giant strongly connected component of a reducible
graph, J. Comput. Appl. Math. 234~(11) (2010) 3075--3090.

\bibitem{Sedg11}
R.~Sedgewick, K.~Wayne, Algorithms, 4 ed, Addison-Wesley, Boston, 2011.

\bibitem{Tarj72}
R.~Tarjan, Depth-first search and linear graph algorithms, SIAM J. Comput.
1~(2) (1972) 146--160.

\bibitem{snap_twitter}
Stanford Network Analysis Project (SNAP), ego-Twitter: social circles from
Twitter, Online dataset, 2012. (Accessed 31 May 2025).

\bibitem{snap_cit_hepph}
Stanford Network Analysis Project (SNAP), cit-HepPh: arXiv high energy
physics paper citation network, Online dataset, 2003. (Accessed 31 May 2025).

\bibitem{snap_ca_astroph}
Stanford Network Analysis Project (SNAP), ca-AstroPh: collaboration network of
arXiv astro physics, Online dataset, 2007. (Accessed 31 May 2025).

\bibitem{Scho12}
B.~Schomburg, On the approximation of ergodic projections and stationary
distributions of stochastic matrices, Electron. J. Linear Algebra 23 (2012)
989--1000.

\bibitem{Doob53}
J.L.~Doob, Stochastic Processes, Wiley, New York, 1953.

\end{thebibliography}

\end{document}